\documentclass[11pt,reqno,a4paper]{article}
\usepackage{amsmath}
\usepackage{amsbsy}
\usepackage{amsthm}
\usepackage{amssymb}
\usepackage{amscd}
\usepackage{mathabx}
\usepackage{accents}
\usepackage{float}
\usepackage{url}
\usepackage{tikz}
\usetikzlibrary{matrix,arrows,decorations.pathmorphing}
 \usepackage{relsize}
\usepackage{xcolor, mathrsfs}

\usepackage{hyperref}
\hypersetup{
pdftitle={},%
pdfauthor={},%
pdfsubject={},%
pdfkeywords={},%
colorlinks=true,%
linkcolor={black},%
linktocpage={},%
pageanchor={},
citecolor={black},
}

\usepackage[english]{babel}
\usepackage[utf8]{inputenc}
\usepackage[margin=2.41cm]{geometry}

\usepackage{hyperref}
\usepackage{url}

\usepackage{cite}

\usepackage{titlesec}
\titleformat{\section}{\large\bfseries\filcenter}{\thesection}{1em}{}
\titleformat{\subsection}{\bfseries}{\thesubsection}{1em}{}

\input{xy}
\xyoption{all}

\newtheorem{theorem}{Theorem}[section]
\newtheorem{theorem-intro}{Theorem}[]

\newtheorem{lemma}[theorem]{Lemma}
\newtheorem{proposition}[theorem]{Proposition}
\theoremstyle{remark}

\theoremstyle{definition}
\newtheorem{remark}[theorem]{Remark}

\newtheorem{definition}[theorem]{Definition}

\newtheorem{conjecture}[theorem]{Conjecture}
\newtheorem{question}[theorem]{Question}

\numberwithin{equation}{section}
\allowdisplaybreaks[1]

\catcode`@=12

\renewcommand\thanks[1]{%
  \begingroup
  \renewcommand\thefootnote{}\footnote{#1}%
  \addtocounter{footnote}{-1}%
  \endgroup
}

\renewcommand{\tilde}{\widetilde}
\renewcommand{\epsilon}{{\varepsilon}}

\usepackage{bbm}

\newcommand{\II}{{\mathbb{I}}}
\newcommand{\KK}{{\mathbb{K}}}
\newcommand{\RR}{{\mathbb{R}}}
\newcommand{\CC}{{\mathbb{C}}}

\newcommand{\G}{\mathsf{G}}
\newcommand{\Id}{{\operatorname{Id}}}

\newcommand{\g}{\mathsf{g}}
\newcommand{\fa}{\mathfrak{a}}
\newcommand{\E}{\mathsf{E}}

\newcommand{\M}{\mathsf{M}}
\newcommand{\N}{\mathsf{N}}
\renewcommand{\P}{\mathsf{P}}
\newcommand{\Q}{\mathsf{Q}}
\newcommand{\R}{\mathsf{R}}
\renewcommand{\S}{\mathsf{S}}
\newcommand{\T}{\mathsf{T}}

\newcommand{\V}{\mathsf{V}}

\newcommand{\D}{\mathsf{Dom}}

\newcommand{\Ker}{\mathsf{Ker}}

\newcommand{\ff}{{\mathfrak{f}}}
\newcommand{\fg}{{\mathfrak{g}}}
\newcommand{\fh}{{\mathfrak{h}}}

\newcommand{\vol}{{\textnormal{vol}\,}}
\newcommand{\supp}{{\textnormal{supp\,}}}

\usepackage{pict2e}

\makeatletter
\DeclareRobustCommand{\intprod}{%
  \mathbin{\mathpalette\int@prod{(0.1,0)(0.9,0)(0.9,0.8)}}%
}
\DeclareRobustCommand{\intprodr}{%
  \mathbin{\mathpalette\int@prod{(0.1,0.8)(0.1,0)(0.9,0)}}}

\newcommand{\int@prod}[2]{%
  \begingroup
  \sbox\z@{$\m@th#1+$}%
  \setlength\unitlength{\wd\z@}%
  \begin{picture}(1,1)
  \roundcap
  \polyline#2
  \end{picture}%
  \endgroup
}
\makeatother

\begin{document}
\begin{center}
	
	\vspace{5mm}

	{\Large\bf 
	 THE QUANTIZATION OF PROCA FIELDS
		\\[3mm] 
		ON GLOBALLY HYPERBOLIC SPACETIMES:\\[4mm]  HADAMARD STATES AND M{\O}LLER OPERATORS } 
	
	\vspace{5mm}
	
	{\bf by}
	
	\vspace{5mm}
	\noindent
	{  \bf  Valter  Moretti$^1$, Simone Murro$^2$ and Daniele Volpe$^1$}\\[2mm]
	\noindent  $^1$ {\it Dipartimento di Matematica, Universit\`a di Trento \& INdAM\& INFN-TIFPA}\\
	{\it Via Sommarive 14,} {\it I-38123 Povo, Italy}\\[1mm]
	\noindent  $^2$ {\it Dipartimento di Matematica, Universit\`a di Genova \& INdAM\& INFN}\\
	{\it Via Dodecaneso 35 
		16146 Genova,  Italy}\\[2mm]
	
	Emails: \ {\tt  valter.moretti@unitn.it, murro@dima.unige.it, daniele.volpe@unitn.it}
	\\[10mm]
\end{center}

\begin{abstract}
This paper deals with several issues concerning the algebraic quantization of the real  Proca field in a globally hyperbolic spacetime and the definition and existence of Hadamard states for that field.   In particular, extending previous work,  we  construct the so-called M\o ller $*$-isomorphism between the algebras of Proca observables on paracausally related spacetimes, proving that the pullback of these isomorphisms  preserves the Hadamard property of corresponding quasifree states defined on the two spacetimes.
 Then, we pull-back a natural  Hadamard state constructed on ultrastatic spacetimes of bounded geometry, along this $*$-isomorphism, to obtain an Hadamard state on a general globally hyperbolic spacetime. We conclude the paper, by comparing the definition of an Hadamard state, here given in terms of wavefront set,  with the one proposed by Fewster and Pfenning, which makes use  of a supplementary  Klein-Gordon Hadamard form. We  establish an (almost) complete equivalence of the two definitions.
\end{abstract}

\\[0.5mm]

\tableofcontents

\section{Introduction}
The (algebraic) quantization of  a quantum field propagating in a globally hyperbolic  curved spacetime $(\M,g)$ \cite{W,AQFT}
and the definition of meaningful  quantum states 
 has been
and continues to be at the forefront of scientific research. Linearized theories are the first step of all perturbative procedures, so the definition of physically meaningful states for linearized field equations is an important task. 

Gaussian, also known as quasifree,  states $\omega: \mathcal{A} \to \CC$ on the relevant CCR or CAR unital $*$-algebra ${\cal A}$ of observables of a given quantum field are an important family of (algebraic) states \cite{IgorValter}. They are completely defined by assigning the two-point function, a bi-distribution $\omega_2(x,y)$ on the sections used to smear the field operator.

A crucial  physical requirement on $\omega$  is the so-called \emph{Hadamard condition}, which is needed, in particular, for defining  locally-covariant renormalization procedures of Wick polynomials \cite{IgorValter, FV2}  and for the mathematical formulation of locally covariant perturbative renormalization in quantum field theory  \cite{rejzner}.

 \subsection{Generalized Klein-Gordon vector fields}

All the notations and conventions used in this section to briefly summarize our results will be defined precisely later. For a charged (i.e. complex)  Klein-Gordon field $A$, possibly vector-valued, the construction of Hadamard states amounts to finding distributional
bi-solutions  $\Lambda_2^\pm(x,y)$ of the Klein-Gordon equation $\N A=0$ describing the two-point functions\footnote{We use throughout the convention of summation over the repeated indices.}
$$\omega(\hat{\fa}(\ff)\hat{\fa}^*(\ff')) = \int_{\M\times \M} \: \Lambda_2^+(x,y)_{cd} \gamma^{ca}(x) \gamma^{db}(y)\overline{\ff_a(x)} \ff'_b(y) \vol_g\otimes \vol_g =: \Lambda_2^+(\bar{\ff},\ff')\:,$$ and
$$\omega(\hat{\fa}^*(\ff')\hat{\fa}(\ff)) = \int_{\M\times \M} \: \Lambda_2^-(x,y)_{cd} \gamma^{ca}(x) \gamma^{db}(y)\overline{\ff'_a(x)} \ff_b(y)  \vol_g\otimes \vol_g =: \Lambda_2^-(\ff,\bar{\ff}')\:.$$
Above, the generators of the CCR $*$-algebra of the Proca field  $\hat{\fa}(\ff)$ and $\hat{\fa}^*(\ff')= \hat{\fa}(\ff)^*$   are the (algebraic) field operators smeared with smooth compactly supported complex sections $\ff, \ff'$ of the relevant {\em complex} vector bundle $\E \to \M$. That bundle is  equipped with a non-degenerate Hermitian\footnote{In this work to be a Hermitian or real  scalar product  does not include the positivity condition, though it is always assumed to be non-degenerate.}  fiberwise scalar product (not necessarily positive) $\gamma$.  In case of the standard complex vector Klein-Gordon 
field over $(\M,g)$ constructed out the $1$-form Hodge D'Alembertian or the Levi-Civita  vector D'Alembertian, the vector bundle $\E$ is the one of smooth $1$-forms $\T^*\M_\CC := \T^*\M+ i\T^*\M$ and the Hermitian scalar product $\gamma$ is the indefinite one induced by the metric $g$  in $\T^*\M_\CC$, i.e., $\gamma= g^\sharp$. In the general case, a {\em Klein-Gordon operator} $\N$ is by definition a second-order operator on the smooth sections of $\E$ which is  {\em normally hyperbolic}  \cite{Ba,BaGi}:  its principal symbol $\sigma_\N$ satisfies
	\begin{equation*}
 \sigma_\N(\xi)=-g^\sharp(\xi,\xi)\,\Id_\E \quad \mbox{for all $\xi\in\T^*\M$, where $\Id_\E$ is the identity automorphism of $\E$.}
	\end{equation*}
	$\N$ is also required to be 
{\em formally selfadjoint} with respect to the Hermitian  scalar product (generally non-positive!) induced on the space of complex sections $\ff$ by $\gamma$ and  the volume form $\vol_g$,
$$(\ff|\fg):= \int_\M \overline{\ff_a}(x)\gamma^{ab}(x) \fg_b(x) \vol_g(x)\:.$$
The scalar complex Klein-Gordon field is encompassed by  simply taking $\CC$ as canonical fiber of $\E$ and using the trivial {\em positive} scalar product.

The requirements on the  bi-distributions $\Lambda_2^\pm$ are, where $\G_\N$ is the {\em causal propagator} of $\N$,
\begin{align*}
(1) & \quad \N_x \Lambda^\pm_{2}(x,y)= \Lambda^\pm_{2}(x,y) \N_y =0  \,\mbox{  and  }   \,  \Lambda_2^+ - \Lambda_2^- = -i\G_\N;\\
(2) & \quad\Lambda^\pm_{2}(\bar{\ff},\ff)  \geq 0\:, \quad \mbox{where  $\Lambda^\pm(\bar{\ff},\ff)=0$ implies $\ff= \N \fg$ for a compactly supported section $\fg$;}\\
(3) & \quad WF(\Lambda_2^\pm) = \{(x,k_x;y,-k_y)\in T^*\M^2\backslash\{0\}\:|\:(x,k_x)\sim_{\parallel}(y,k_y), k_x\triangleright0\}\:.
\end{align*}
The second part of (1) corresponds to the \emph{canonical commutations relations}, the first part is the "on-shell" condition, while condition~(2) is the \emph{positivity} requirement on two-functions. Then, the {\em Gelfand--Naimark--Segal construction}  gives rise to  a $*$-representation of $\mathcal{A}_g$ in terms of densely defined operators in a 
Hilbert space which, as a consequence of the above requirements (1) and (2)  and the Wick rule, is a Fock space. Here  $\omega$ is the  expectation value referred to vacuum state and the action of the image of the representation on the vacuum state produces  the dense invariant  domain of the representation itself.
Requirement  (3) is the celebrated \emph{Hadamard condition} (also known as the {\em microlocal spectrum condition}) which ensures the correct short-distance behavior of the $n$-point functions of the  field. This condition has a long history which can be traced back to \cite{FSW}, passing to \cite{KW} and \cite{Rad1,Rad2} (see \cite{IgorValter} for a review). It plays a crucial role in various contexts of quantum field theory  in curved spacetime. In particular, but not only, in  perturbative renormalization and semiclassical quantum gravity.  More recently, G\'erard and Wrochna in~\cite{gerard0,gerardKG}, proved that condition (1)-(3) can be controlled at the same time by using methods of pseudodifferential calculus in spacetimes {\em of bounded geometry} (see also the subsequent papers~\cite{gerardYM,gerard2,
gerard3,gerard4,gerard5}).

When dealing with {\em real} quantum fields, as in this work, for instance the  Klein-Gordon real vector  field $A$, a single bidistribution $\omega_2(x,y)$ is sufficient to define a quasifree state 
$\omega$:
$$\omega(\hat{\fa}(\ff)\hat{\fa}(\ff')) = \int_{\M\times \M} \: \omega_2(x,y)_{cd} \gamma^{ca}(x) \gamma^{db}(y)\ff_a(x) \ff'_b(y) \vol_g\otimes \vol_g$$
where $\hat{\fa}(\ff)=\hat{\fa}(\ff)^*$ is  the (algebraic) field operators smeared with smooth {\em real} compactly supported sections $\ff$ of a relevant {\em real} vector bundle $\E \to \M$, equipped with a fiberwise 
real symmetric non-degenerate (but not necessarily positive)  scalar product $\gamma$.
As before, a {\em Klein-Gordon operator}  $\N$ is  by definition a second-order differential operator on the smooth sections of $\E$ which is normally hyperbolic (same definition as for the complex case)  and formally selfadjoint with respect to the real symmetric scalar product (not necessarily positive) 
$$(\ff|\fg):= \int_\M {\ff_a}(x)\gamma^{ab}(x) \fg_b(x) \vol_g(x)\:.$$
In the case of the standard  real vector Klein-Gordon field (constructed out of the Hodge D'Alembertian or the Levi-Civita D'Alembertian) the  bundle is  exactly  $\T^*\M$,  equipped with a real symmetric non-degenerate but indefinite  fiberwise scalar product induced by the metric  $g$ on $\T^*\M$, namely $\gamma=g^{\sharp}$. The theory of the scalar real Klein-Gordon field is encompassed simply by taking $\RR$ as canonical fiber of $\E$ with trivial positive scalar product.

In the real case, defining the symmetric bilinear form $\mu(\ff,\ff'):= \frac{1}{2} (\omega_2(\ff,\ff')+ \omega_2(\ff',\ff))$,  conditions (1)-(3) are replaced by
\begin{align*}
(1)' & \quad \N_x \omega_{2}(x,y)= \omega_{2}(x,y) \N_y  =0  \,\mbox{  and  }   \omega_2(\ff,\ff') - \omega_2(\ff',\ff) =  i\G_\N(\ff,\ff')\:;\\
(2)' & \quad \mu_{2}(\ff,\ff)  \geq 0 \quad \mbox{where  $\mu(\ff,\ff)=0$ implies $\ff= \N \fg$ for a compactly supported section $\fg$;}\\
(3)' & \quad  |\G_\N(\ff,\ff')|^2 \leq 4\mu(\ff,\ff)\: \mu(\ff',\ff')\:;\\
(4)' & \quad WF(\omega_2^\pm) = \{(x,k_x;y,-k_y)\in T^*\M^2\backslash\{0\}\:|\:(x,k_x)\sim_{\parallel}(y,k_y), k_x\triangleright0\}\:.
\end{align*}
The apparently new continuity condition (3)' for the real case is actually embodied in the positivity condition (2) for the complex case \cite{gerardBook}.  As a matter of fact (2)' and (3)' together give rise to positivity of the whole state  $\omega$ on $\mathcal{A}_g$ induced by $\omega_2$  in the real case. Once again,  the GNS construction gives rise to a representation of the (complex) unital $*$-algebra $\mathcal{A}_g$ generated by the field operators $\hat{\fa}(\ff)$ exactly as in the complex case.

Since   the Klein-Gordon equations   are   normally hyperbolic, not only they are {\em Green hyperbolic}  so that the Green operators  $\G^\pm_\P$ and the causal propagator $\G_P=\G_P^+-\G_P^-$ can be therefore defined, but {\em the Cauchy problem} is also {\em  automatically} well posed  \cite{Ba,BaGi}. An important implication of this fact is that  the  two-point function of a quasifree state can be defined as a Hermitian  or real  bilinear form --  in the complex and real case respectively --  on the Cauchy data of solutions of the Klein-Gordon equation (e.g., see \cite{Norm}). We follow this route in the present paper and, to this end, we  will translate (1)-(3) and (1)'-(4)' in the language of Cauchy data. 

\subsection{Issues with the quantization of the Proca field} 
Most of the quantum  theories are described by {\em Green hyperbolic operators} \cite{Ba,BaGi}, as  Klein-Gordon operators $\N$ discussed above or the {\em Proca operator} \cite{Fewster, koProca},  studied in this work,  $$\P = \delta d + m^2$$ acting on smooth $1$-forms $A\in \Omega^1(\M)$ and where $m^2>0$ is a constant. These operators  are usually formally self-adjoint w.r.t. a (Hermitian or real symmetric) scalar   product induced by the analog $\gamma$  on the fibers of the relevant vector bundle.  In general $\gamma$ is {\em not positive definite}. Very common and physical examples are: the {\em standard vector} Klein-Gordon field, the  Proca field, the Maxwell field, more generally, the Yang-Mills field and also the linearized gravity. Referring to the Proca, and in general all 1-form fields, we have that $\gamma= g^\sharp$ is the inverse (indefinite!) Lorentzian metric of the spacetime $(\M,g)$. 

 Unfortunately,  in those situations,  the Hadamard condition  (4) and (5)' are  in conflict with the positivity of states, respectively, (3) and (2)'-(3)'. It is known  that for a  vectorial Klein-Gordon operator that is formally self-adjoint w.r.t. an {\em indefinite Hermitian/real symmetric  scalar product},  the existence of quasifree Hadamard states is forbidden (see  the comment after \cite[Proposition 5.6]{SV} and \cite[Section 6.3]{gerardYM}).

 The case of a (real) Proca field   seems to be  even more complicated at first glance. In fact,  on the one hand  differently from the Klein-Gordon operator, the Proca operator  is not even {\em normally hyperbolic} and this makes more difficult (but not impossible) the proof of the well-posedness of the Cauchy problem, in particular. On the other hand, similarly to the case of the vectorial  Klein-Gordon theory, the Proca theory  deals with an indefinite fiberwise scalar product.  
Actually, as we shall see in the rest of the work, {\em these two apparent drawbacks cooperate to permit the existence of  quasifree Hadamard states}.  Positivity of the two-point function $\omega_2$ is restored when dealing with a {\em constrained} space of Cauchy conditions that make well-posed the Cauchy problem.

In the present paper, we study the existence of quasifree Hadamard states  for the real  Proca field on a general globally hyperbolic spacetime.  A definition of Hadamard states for the Proca field  was introduced by Fewster and Pfenning in \cite{Fewster}, to study {\em quantum energy inequalities},
with a definition more involved than the one based on conditions (3) and (4)' above.  They also managed to prove that such states exist in globally hyperbolic spacetimes whose Cauchy surfaces are compact.

Differently from Fewster-Pfenning's definition, here we adopt a definition of Hadamard state which directly relies on conditions (3) and (4)' above and we consider a generic globally hyperbolic spacetime. At the end of the work, we actually prove that the two definitions of Hadamard states are substantially equivalent.

Before establishing that equivalence, using the technology of the M\o ller operators we introduced in \cite{Norm} for normally hyperbolic operators, and here extended to the Proca field, we prove the existence of quasifree Hadamard states in every globally hyperbolic spacetime, also in the case in which their Cauchy hypersurfaces are not compact. 

As a matter of fact,  it is enough to focus our attention on {\em ultrastatic} spacetimes of {\em bounded geometry}. In this class of spacetimes, we   directly work at the level of initial data for the Proca equation and we establish  the following, also by taking advantage of some technical results of spectral theory applied to {\em elliptic Hilbert complexes} \cite{BuLe}.
\begin{itemize}
\item[1.] The initial data of the Proca equations are a subspace $C_\Sigma$ of the initial data of {\em a couple} of  Klein-Gordon equations, one scalar and the other vectorial, however  both defined on bundles with fiberwise  {\em positive} real symmetric  scalar product;
\item[2.] The difference of a pair of certain  Hadamard two-point functions for two above-mentioned  Klein-Gordon fields becomes positive once that its arguments are  restricted to $C_\Sigma$.  There, it defines a two-point function $\omega_2$ for a quasifree state $\omega$ of the Proca field;
\item[3.] $\omega$  is also Hadamard since it is the difference of two two-point functions of Klein-Gordon fields which are Hadamard. 
They are Hadamard in view of known results of microlocal analysis of pseudodifferential operators on Cauchy surfaces of bounded geometry, for more details the interested reader can refer to \cite{gerardBook}.
\end{itemize}
Every field theory defined on a globally hyperbolic spacetime $(\M,g)$ is connected to one defined on an ultrastatic spacetime of bounded geometry
$(\RR\times \Sigma, -dt^2 +h)$
 through a M\o ller operator:  the associated M\o ller $*$-isomorphism between the algebras of Proca observables preserves the Hadamard condition. We therefore conclude that every globally hyperbolic spacetime $(\M,g)$ admits a Hadamard state for the Proca field. This state is nothing but the Hadamard state on $(\RR\times \Sigma, -dt^2 +h)$ pulled back to $(\M,g)$ by the M\o ller $*$-isomorphism.

One novelty of this paper is in particular a direct control of the positivity of the two-point functions, obtained by  spectral calculus of elliptic  Hilbert complexes. Some  microlocal property of the M\o ller operators then guarantees  the validity of the Hadamard condition without exploiting the classical so called {\em deformation argument}, or better, by re-formulating it into a new form in terms of M\o ller operators.

\subsection{Main results}
We explicitly  state here the principal results established in this paper referring, for the former, to the notions introduced in the previous section.
Below, $\G^\pm_{\P}$ denote the retarded and advanced Green operators of the Proca equation (\ref{ProcaA}),   we shall discuss in Section \ref{SECGREEN}. The symbol $\kappa_{g'g}$ denotes a linear fiber-preserving isometry from the spaces of smooth sections $\Gamma(\V_{g})$ to $\Gamma(\V_{g'})$ constructed in Section \ref{SECGREEN}. Here, $\V_g$ indicates the vector bundle of real $1$-forms over the spacetime $(\M,g)$ whose sections are the argument of the Proca operator $\P$. 

\begin{theorem-intro}[Theorems~\ref{remchain} and~\ref{causalpropR}]\label{thm:main intro 1}

Let $(\M,g)$ and $(\M,g')$ be globally hyperbolic spacetimes,   with associated real  Proca bundles  $\V_{g}$ and $\V_{g'}$ and Proca operators   $\P, \P'$.\\
	If the metric are paracausally related $g\simeq g'$, then  there exists a $\RR$-vector space isomorphism $\R : \Gamma(\V_{g}) \to \Gamma(\V_{g'})$, called {\bf M\o ller operator} of $g,g'$ (with this order), such that the following facts are true.
	\begin{itemize}
	\item[(1)] The restriction, called {\bf M\o ller map}
\begin{align*}
 \S^0 := \R|_{ \Ker_{sc}(\P)}  :  \Ker_{sc}(\P) \to   \Ker_{sc}(\P')  
 \end{align*}
is  well-defined vector space isomorphism with inverse  given by
\begin{equation*} (\S^0)^{-1} := \R^{-1}|_{ \Ker_{sc}(\P')}  :  \Ker_{sc}(\P') \to   \Ker_{sc}(\P) \,.
\end{equation*}
\item[(2)] It holds $\kappa_{gg'} \P'\R = \P $.
		\item[(3)]   The causal propagators 
		$\G_\P:= \G_\P^+-\G_\P^-$ and $\G_{\P'}:= \G_{\P'}^+-\G_{\P'}^-$, respectively
		of $\P$ and $\P'$, satisfy $\R \G_\P  \R^{\dagger_{g g'}} = \G_{\P'}\,.$
		\item[(4)] It holds  $
		\R^{\dagger_{g g'}} \P'\kappa_{g'g}|_{\Gamma_c(\V_g)}= \P|_{\Gamma_c(\V_g)}$, where the adjoint $^{\dagger_{gg'}}$ is defined in Definition \ref{defadjoint}.
	\end{itemize}	
\end{theorem-intro}

The next result (Theorem 2) permits us  to promote $\R$ to a  $*$-isomorphism $\mathcal{R}$ of the algebras of field operators $\mathcal{A}$, $\mathcal{A}'$
respectively associated to the paracausally related metrics $g$ and $g'$, with the associated $\P,\P'$, and
 generated by respective Hermitian  field operators 
$\fa(\ff)$ and $\fa'(\ff')$ with $\ff,\ff'$ compactly supported smooth real sections of $\V$. We will introduce these notions 
in Section \ref{SECALGEBRAS}.
These field operators satisfy respective CCRs
$$[\fa(\ff), \fa(\fh)]= i \G_\P(\ff,\fh) \II\:, \quad [\fa'(\ff'), \fa'(\fh')]= i \G_{\P'}(\ff',\fh')\II'$$
and the said unital $*$-algebra isomorphism $\mathcal{R}: \mathcal{A}' \to \mathcal{A}$ is uniquely determined by the requirement 
$$\mathcal{R}(\fa'(\ff)) = \fa(\R^{\dagger_{gg'}} \ff)\:, \quad \ff \in \Gamma_c(\V_{g'})\:.$$
The final important result regards the properties of $\mathcal{R}$ for the algebras of a pair of paracausally related metrics $g,g'$ when it acts on the states $\omega : \mathcal{A}\to \CC$, $\omega' : \mathcal{A}'\to \CC$ of the algebras in terms of pull-back.  
$$\omega' = \omega \circ \mathcal{R}\:.$$
As is known, the most relevant (quasifree)  states in algebraic QFT are {\em Hadamard states} characterized by  the {\em microlocal spectrum condition} valid for the wavefront set of their  two-point functions or, equivalently, an universal short distance structure of the distribution defining the two-point function.
A  definition of Hadamard state for the Proca field was first stated by Fewster and Pfenning in \cite{Fewster} and corresponds to Definition~\ref{Hadamard} in this paper.  That definition requires the existence of a bisolution of the {\em Klein Gordon} field satisfying the   microlocal spectrum condition. This bisolution is next used to construct the two-point function of the Proca field.
 Differently, in this work we adopt  a  direct definition (Definition~\ref{HadamardMMV})  which only requires  the  validity  of the  microlocal spectrum condition directly  for the two-point function of the Proca two-point function. We also prove that our  definition, exactly as it happens for   Fewster and Pfenning's definition, satisfies some physically relevant properties. In addition to these general results, we also prove that {\em the Hadamard property is preserved by the M{\o}ller operators} as one of main results of this work. 

\begin{theorem-intro}[Theorem~\ref{thm:main intro Had}]\label{thm:main intro 2}
 Let  $g,g'$ be paracausally related metric and consider the corresponding   Proca operators $\P,\P'$. Finally refer to the associated on-shell  $CCR$-algebras $\mathcal{A}$ and $\mathcal{A}'$.\\
 Let
$\omega:\mathcal{A}\to\CC$ be a quasifree Hadamard state.  The pull-back state $\omega':\mathcal{A'}\to\CC$ by
$\omega'=\omega\circ\mathcal{R}$ satisfies,
\begin{enumerate}
\item $\omega'$ is a well-defined state;
\item $\omega'$ is quasifree;
\item $\omega'$ is a Hadamard state.
\end{enumerate}
\end{theorem-intro}

Attention is next focused on the {\em existence problem} of quasifree Hadamard states for the real Proca field in a {\em generic} globally hyperbolic spacetime.
The technology of M{\o}ller operators allows us to reduce the construction of Hadamard
states for the real  Proca field to the special case of an {\em ultrastatic spacetime}  $(\RR\times \Sigma, -dt^2+h)$. In this class of spacetimes, if  assuming the further  geometric hypothesis of {\em bounded geometry},  we provide a direct construction of a Hadamard state just  working on the space of initial data $C_\Sigma $ for the Proca equation  $\P A =0$ where $A\in \Gamma(\V_g)$ has compact Cauchy data. Here,  $A$ decomposes as $A= A^{(0)}dt +A^{(1)}$, where $A^{(0)}$ and $A^{(1)}$ and are, respectively, a $0$-form and a $1$-form on $\{t\}\times \Sigma$.  As we shall prove, this  space of initial data  is actually {\em constrained} in order to satisfy the existence and uniqueness theorem for the Cauchy problem:  
\begin{equation*}
 C_\Sigma := \left\{ (a^{(0)},\pi^{(0)}, a^{(1)},\pi^{(1)})   \in \Omega_c^0(\Sigma)^2 \times \Omega_c^1(\Sigma)^2  \:\left|\: \pi^{(0)}  =- \delta^{(1)}_h a^{(1)} \:, \quad   (\Delta_h^{(0)} + m^2) a^{(0)} = \delta^{(1)}_h \pi^{(1)}\right. \right\}\:,
\end{equation*}
where $(a^{(0)},\pi^{(0)}) := (A^{(0)}, \partial_t A^{(0)})|_{t=0}$ and $(a^{(1)},\pi^{(1)}) := (A^{(1)}, \partial_t A^{(1)})|_{t=0}$.

\begin{theorem-intro}[Propositions~\ref{PROPHAD1} and~\ref{prop:ultraHada}] \label{thm:intro intermezzo}
Consider the  $*$-algebra $\mathcal{A}_g$ of the real Proca field on the ultrastatic spacetime $(\M,g)= (\RR\times \Sigma, -dt\otimes dt+h)$, with $(\Sigma, h)$ a smooth complete Riemannian manifold.
 Let $\eta_0:=-1$, $\eta_1:=1$ and $h^\sharp_{(j)}$ denote the standard inner product of $j$-forms on $\Sigma$ induced by $h$.
 Then the two-point function
\begin{equation}\nonumber 
\omega_\mu(\fa(\ff)\fa(\ff')) = \omega_{\mu2}(\ff,\ff') := \mu(A, A') + \frac{i}{2}\sigma^{(P)}(A, A')\ 
\end{equation}
defines a quasifree state $\omega_\mu$  on  $\mathcal{A}_g$ where $ \ff,\ff' \in \Gamma_c(\V_g) $. Above
$$ A=\G_\P\ff \:,\qquad A'=\G_\P\ff' \:, \qquad\sigma^{(P)}(A, A') =   \int_\M g^\sharp(\ff, \G_\P \ff') \:  \vol_g $$
$$\mu(A,A') := \sum_{j=0}^1 \frac{\eta_j}{2} \int_\Sigma   h_{(j)}^\sharp(\pi^{(j)}, (\overline{\Delta^{(j)} + m^2})^{-1/2} \pi^{(j)'})    +  h^\sharp_{(j)}  (a^{(j)},  (\overline{\Delta^{(j)} + m^2})^{1/2} a^{(j)'}  )\: \vol_h $$
 where $\Delta^{(j)}$ is the Hodge Laplacian for compactly supported real smooth $j$-forms on $(\Sigma, h)$.\\
Finally, $\omega_\mu$ is Hadamard if $(\Sigma,h)$ is of bounded geometry.
\end{theorem-intro}
Above the bar denotes the closure of the considered operators defined in suitable $L^2$-spaces of forms according to the theory of elliptic Hilbert complexes.

Using the fact that every globally hyperbolic spacetime is paracausally related to an ultrastatic spacetime with bounded geometry and 
 combining the two  previous Theorems, we can conclude that Proca fields can be quantized in any globally hyperbolic spacetime and admit Hadamard states.

\begin{theorem-intro}\label{thm:main intro 3}
Let $(\M, g)$ be a globally hyperbolic spacetime and refer to the
associated $CCR$-algebras $\mathcal A_g$ of the real Proca field. Then there exists a quasifree  Hadamard state on $\mathcal A_g$.
\end{theorem-intro}

Eventually,  coming  back to the alternative definition of Hadamard states proposed by Fewster and Pfenning in \cite{Fewster},  we prove an almost complete  equivalence theorem, which is the last main achievement of this work.

\begin{theorem-intro}[Theorem~\ref{teoequiv}]\label{thm:main intro 4}
Consider the globally hyperbolic spacetime $(\M,g)$ and a quasifree  state $\omega:\mathcal{A}_g\to\CC$ 
for the Proca algebra of observables on $(\M,g)$ with two-point function $\omega \in \Gamma_c'(\V_g\boxtimes \V_g)$.   The following facts are true.
\begin{itemize}
\item[(a)] If $\omega$ is Hadamard according to Fewster and Pfenning, then  it is also Hadamard according to Definition~\ref{HadamardMMV}. 
\item[(b)] If $(\M,g)$ admits a Hadamard state according to Fewster and Pfenning  and $\omega$ is Hadamard according to Definition~\ref{HadamardMMV},  then $\omega$ is Hadamard in the sense of Fewster-Pfenning's  definition.
\end{itemize}
\end{theorem-intro}
The existence of Hadamard states according to Fewster-Pfenning's definition was proved in \cite{Fewster} for spacetimes with compact Cauchy surfaces. For these spacetimes the equivalence of the two definitions  is complete.

\subsection{Structure of the paper}

The paper is structured as follows. In Section~\ref{SECGREEN} we will provide a detailed analysis of the M\o ller maps and the M\o ller operator for classical Proca fields. In particular, we will analyze the relation between the M\o ller operators and the causal propagators of Proca operators on paracausally related globally hyperbolic spacetimes. In Section~\ref{SECALGEBRAS} we will extend the   
action of the M\o ller operators to a $*$-isomorphism of the $CCR$-algebras of free Proca fields. This will allow us to pullback quasifree Hadamard states preserving the microlocal spectrum condition. In this section we also analyze the general properties of Hadamard states including their existence. The explicit construction of Hadamard states in an ultrastatic spacetime is performed in Section~\ref{sec:existence}. In Section~\ref{SECHADFP} we show that the microlocal spectrum condition is essentially equivalent to the definition of Hadamard states proposed by Fewster and Pfenning. Finally, we
conclude our paper with Section~\ref{sec:concl}, where open issues and future prospects are presented.

\subsection*{ Acknowledgments}
We are grateful to Nicol\`o Drago, Chris Fewster,   Christian G\'erard, and Igor Khavkine  for helpful discussions related to the topic of this paper.    This work was written  within the activities of the INdAM-GNFM

\subsection*{Funding}
V.M. and D.V. acknowledge the support of the INFN-TIFPA national project ``Bell''.
S.M. acknowledges the support of the INFN  and the GNFM-INDAM.

\section{Mathematical setup}	\label{sec2}

\subsection{Conventions and notation of geometric tools in spacetimes}
Throughout all the paper the symbols $\subset$ and $\supset$ allow the case $=$.\\
We explicitly adopt the signature $(-,+,\cdots, +)$ for Lorentzian metrics.

Throughout  $(\M,g)$ denotes a {\bf spacetime}, i.e., a paracompact, connected, oriented, time-oriented, smooth,
Lorentzian manifold $\M$, whose Lorentzian metric is $g$.  As in \cite{Norm}, the Lorentzian metrics $g$ of  spacetimes  are hereafter  supposed to be  equipped with their own  temporal orientation.

All considered spacetimes  $(\M,g)$ are also {\bf globally hyperbolic}. In other words, 
  a (smooth) {\bf Cauchy  temporal function} $t:\M\to\RR$ exists. By definition $dt$ is timelike, past directed and 
\begin{align*}
	(\M,g)  \quad \mbox{is isometric  to}\quad (\RR\times\Sigma, g')
	\quad \mbox{with metric}\quad 
	g' = -\beta d t\otimes dt+h_t\,,
\end{align*}
where $\beta:\RR \times \Sigma \to \RR$ is a smooth positive function, $h_t$ 
is a Riemannian metric on each slice
$\Sigma_t:=\{t\} \times \Sigma$ varying smoothly with $t$, and these slices are smooth 
{\bf spacelike Cauchy hypersurfaces}. By definition they are   achronal sets intersected 
exactly 
once by every inextensible timelike curve (see~\cite{Minguzzi} for a recent up-to-date survey on the subject).

According to \cite{Norm},  given two globally hyperbolic metrics $g$ and $g'$ on $\M$ ,  $g \preceq g'$ means that $V_p^{g+} \subset V_p^{g'+}$ {\em for all $p\in \M$}, where  $V_p^{g+}\subset \T_p\M$ is the open cone of future directed timelike vectors at $p$ in $(\M,\g)$.

Two globally hyperbolic metrics $g$ and $g'$ on $\M$  are {\bf paracausally related} , written   $g\simeq g'$,  if 
there exists a finite sequence of globally hyperbolic  metrics $g_1=g, g_2\ldots, g_n = g' $ on $\M$ such that for each  pair of consecutive metrics either
$$ g_{k} \preceq g_{k+1}  \qquad\text{ or } \qquad 
 g_{k+1} \preceq g_{k}\:.$$  For a discussion on this notion, its properties, and examples we refer to \cite[Section 2]{Norm}. 

 We henceforth denote by $\Gamma(\E)$ the real vector space of smooth  sections of any real vector bundle $\E \to \M$. More precisely, 
as in~\cite{Norm}, we denote with $\Gamma_c(\E), \Gamma_{fc}(\E),\Gamma_{pc}(\E),\Gamma_{sc}(\E)$  the space of sections respectively with {\bf compact support}, {\bf future-compact} (i.e. whose support stays before a smooth spacelike Cauchy surface), {\bf past-compact}  (i.e. whose support stays after  a smooth spacelike Cauchy surface), and {\bf spatially-compact support}  (i.e. whose support
on every smooth spacelike Cauchy surface  is compact).

If $\textsf{E}\rightarrow \textsf{M}$ and  $\textsf{E}'\rightarrow \textsf{M}'$ are two vector bundles,
$\textsf{E}\boxtimes \textsf{E}'$ denotes the  {\bf external tensor product} of these vector bundles. This vector bundle has base $\textsf{M}\times 
\textsf{M}'$ and fiber at $(p,p')$ given by the tensor products of the respective fibers at $p\in \textsf{M}$ and  $p'\in \textsf{M}'$ respectively.
If $\mathfrak{f} \in \Gamma(\textsf{E})$ and  $\mathfrak{f}' \in \Gamma(\textsf{E}')$, the section $\mathfrak{f}\otimes \mathfrak{f'} \in
 \Gamma(\textsf{E}\boxtimes \textsf{E}')$ is defined by $\mathfrak{f}\otimes \mathfrak{f'}(p,p'):= \mathfrak{f}(p)\otimes \mathfrak{f'}(p')$.
The tensor product of linear operators acting on sections of an external product bundle are denoted by $\otimes$.

\subsection{Smooth  forms, Hodge operators,  and the Proca equation}

In this work we frequently  deal with real smooth $k$-forms $\ff,\fg \in \Omega^k(\M)$,
where $k=0,\ldots, n= \dim \M$ (and one usually adds $\Omega^{n+1}(\M) = \Omega^{-1}(\M) = \{0\}$).
 The {\bf Hodge} {\em real}  {\bf inner product} can be computed by integrating the fiberwise product with respect to the volume form induced by $g$:
\begin{equation*}
( \ff| \fg )_{g, k} := \int_\M \ff  \wedge * \fg  =  \int_\M  g_{(k)}^\sharp(\ff,\fg)\: \vol_g \:,
 \end{equation*}
where at least one of the two forms has compact support and $g_{(k)}^\sharp$ is the natural inner product of $k$-forms induced by  $g$. This symmetric real scalar product $( \cdot| \cdot )_{g, k} $ is always  non-degenerate but it is not positive when $g$ is Lorentzian as in the considered case. It is  positive when $g$ is Riemannian.
If $k=1$, we simply write
\begin{equation}( \ff| \fg )_g = \int_\M g^\sharp(\ff,\fg)\: \vol_g\:.\label{hodge} \end{equation}
In this context,  $d^{(k)}:\Omega^k(\M) \to \Omega^{k+1}(\M)$ is the exterior derivative  and $\delta_g^{(k)}: \Omega^{k}(\M) \to \Omega^{k-1}(\M)$ 
 is the codifferential operator acting on the relevant spaces of smooth $k$-forms $\Omega^k(\M)$ on $\M$ depending on the metric $g$ on $\M$.  $d^{(k)}$ and $\delta_g^{(k+1)}$ are the {\bf formal adjoint} of one another with respect to the Hodge product (\ref{hodge}) in the sense that
$$( d^{(k)}\ff| \fg)_{g,k+1}  = ( \ff| \delta_g^{(k+1)} \fg)_{g, k} \:,  \quad \forall \ff \in  \Omega^k(\M)\:, \: \forall \fg \in \Omega^{k+1}(\M)\quad  \mbox{if $\ff$ or $\fg$ is compactly supported.} $$
In the rest of the paper we will often omit the indices$_{g,k}$ and  $^{(k)}$ referring to the metric and the  order  of the used forms,  when the choice of  the used metric and order will be obvious from the context.

If  $(\M,g)$ is globally hyperbolic,  we call {\bf Proca bundle} the real  vector bundle $\V_g := (\T^*\M, g^{\sharp})$ obtained by endowing the cotangent bundle with the fiber metric given by the dual metric $g^{\sharp}$ (also appearing in (\ref{hodge})) defined by
$$g^\sharp (\omega_{p},\omega'_{p}) := g(\sharp{\omega}_p,\sharp{\omega}'_p)\quad \mbox{for every   $\omega,\omega' \in\Gamma(\T^*\M)$ and $p\in \M$,}$$
where $\sharp:\Gamma(\T^*\M)\to \Gamma(\T\M)$ is the standard musical isomorphism. 

By construction  $\Gamma(\V_g) = \Omega^1(\M)$ and $\Gamma_c(\V_g) = \Omega_c^1(\M)$. Here and henceforth $\Omega^k_c(\M)\subset \Omega^k(\M)$ is the subspace of compactly supported real smooth $k$-forms on $\M$.

The formally selfadjoint {\bf Proca operator}  $\P$ on $(\M,g)$ is    defined by choosing a (mass) constant $m>0$, {\em the same  for all globally hyperbolic metrics we will consider on $\M$ in this work},
 \begin{equation}
	\P=\delta d + m^2:  \Gamma(\V_g) \to \Gamma(\V_g), \label{PROCAOP}
	\end{equation}
where $d:= d^{(1)}$, $\delta:= \delta_g^{(2)}$.
Actually $\P$ depends also  on $g$, but we shall not indicate those dependencies in the notation for the sake of shortness. 

The {\bf Proca equation} we shall consider in this paper reads
\begin{equation}
\P A =0 \quad \mbox{for $A\in \Gamma_{sc}(\V_g)$} \label{ProcaA}\:,
\end{equation}
where, as said above, $\Gamma_{sc}(\V_g)$ is the space of real smooth $1$-forms which have compact support on the Cauchy surfaces of the globally hyperbolic spacetime $(\M,g)$.

\section{M\o ller Maps and M\o ller Operators}\label{SECGREEN}

The construction of the so-called M\o ller operator for hyperbolic PDEs (coming from the realm of quantum field theories on curved spacetimes) has been studied extensively in various contexts in Quantum Field Theory, see e.g.~\cite{Moller, DHP, FPMoller,DefArg1,Norm,DiracMoller}. The key idea was to inspired by the
scattering theory: Starting with two ``free theories'' described  by the space of solutions of {\em normally hyperbolic operators} (see (\ref{NOP}) below) $\N_0$ and $\N_1$ in corresponding  spacetimes $(\M,g_0)$ and $(\M,g_1)$, respectively,  we connected them through  an ``interaction spacetime'' $(\M,g_\chi)$ with a ``temporally localized'' interaction defined by interpolating  the two metrics by means of a smoothing function $\chi$. Here we need two M{\o}ller maps: $\Omega_+$ connecting $(\M,g_0)$ and $(\M, g_\chi)$ --
which reduces to the identity in the past when $\chi$ is switched off --
and a second M{\o}ller map connecting $(\M,g_\chi)$ to $(\M,g_1)$ -- which reduces to the identity in the future when $\chi$ constantly takes the value $1$. The ``$S$-matrix'' given by the composition 
$\S :=\Omega_-\Omega_+$
 will be the M\o ller map connecting $\N_0$ and $\N_1$. 
 
  As remarked in~\cite[Section 6]{Norm}, all the results concern vector-valued normally hyperbolic operators  acting on real vector bundles  whose fiber metric  does not depend on
the globally hyperbolic metrics $g$ chosen on $\M$. These operators are also assumed to be formally selfadjoint with respect to the associated real symmetric  scalar product on the sections of the bundle. 
As already pointed out in the introduction, to  quantize the theory defining quantum states on an associated $*$-algebra of observables, the fiberwise metric on $\E$ should be  assumed to be {\em positive}.

This section aims to extend the construction of the M\o ller operators to Proca fields.
The main difficulties we have to face with respect to the case of the Klein-Gordon equation  are the following:
\begin{itemize}
\item the fiber metric of the Proca bundle depends on the underlying globally hyperbolic metrics $g$ chosen on $\M$ (and it  is not positive definite);
\item Proca operators are not normally hyperbolic.
\end{itemize}
The next two sections are devoted to tackle these technical issues before starting with the construction of the M\o ller maps.

\subsection{Linear fiber-preserving isometry }
As said above, to construct M\o ller maps for the Proca field  we should be able to compare different fiberwise metrics on  $\T^*\M$ when we change the metric $g$ on $\M$. This will be done by defining suitable fiber preserving isometries.

If $g$ and $g'$ are globally hyperbolic on $\M$ and $g \simeq g'$, it is possible to define
a linear fiber-preserving isometry from $\Gamma(\V_{g})$ to $\Gamma(\V_{g'})$ we denote with $\kappa_{g'g}$ and we shall take advantage of it very frequently in the rest of this work. In other words, if $\ff \in \Gamma(\V_{g})$, then $\kappa_{g'g} \ff \in \Gamma(\V_{g'})$, the map $\kappa_{g'g}: \Gamma(\V_{g})\to \Gamma(\V_{g'})$ is $\mathbb{R}$ linear, 
and $$g'^\sharp((\kappa_{g'g} \ff)(p), (\kappa_{g'g} \fg)(p))= g^\sharp(\ff(p), \fg(p))\quad \forall p\in \M\:.$$
Let us describe the (highly non-unique) construction of $\kappa_{gg'}$.
If $\chi \in C^\infty(\M; [0,1])$ and $g_0 \preceq g_1$,   then \begin{equation}g_\chi := (1-\chi) g_0 + \chi  g_1\label{gCHI}\end{equation}
is a Lorentzian metric globally hyperbolic on $\M$ (see ~\cite[Section 2]{Norm} for details) and satisfies 
 \begin{equation*}
  g_0\preceq g_\chi \preceq g_1 \:.
\end{equation*}

Now consider the product manifold  $\N:= \mathbb{R} \times \M$, equipped with the indefinite non-degenerate  metric 
$$h:= -dt\otimes dt + g_t\:,$$
 where $g_t = (1-f(t))g_0 + f(t) g_\chi$
and $f: \mathbb{R} \to [0,1]$ is smooth and $f(t)=0$ for $t\leq 0$, $f(t)=1$ for $t\geq 1$.
Notice that $g_t$ is Lorentzian according to \cite{Norm} because $g_0 \preceq g_\chi$ and $h$ is indefinite non-degenerate  by construction. At this point $\tilde{\kappa}_{\chi 0}: \T\M \to \T\M$ is the fiber preserving diffeomorphism such that $\tilde{\kappa}_{\chi 0}(x,v)$ is the parallel transport form $(0,x)$ to $(1,x)$ of $v\in \T_x\M \subset \T_{(0,x)}\N$ along the complete $h$-geodesic $\mathbb{R} \ni t \mapsto (t,x) \in \N$. Standard theorems on joint smoothness of the flow of ODEs depending on parameters assure that $\tilde{\kappa}_{\chi 0}: \T\M \to \T\M$ is smooth. Notice that 
$\tilde{\kappa}_{\chi 0}|_{\T_x\M}: \T_x\M \to \T_x\M$
is also a $h$-isometry from known properties of the parallel transport and thus it is a $g_0,g_\chi$-isometry by construction  because $h_{(t,x)}(v,v) = g_t(v,v)$ if $v\in \T_{x}\M \subset \T_{(t,x)}\N$. Taking advantage of the musical isomorphisms, $\tilde{\kappa}_{\chi 0}$ induces  a fiber-bundle  map $\kappa_{\chi 0} : T^*\M \to T^*\M$ which can be seen as a map on the sections of  $\Gamma(\V_{g_0})$  and producing sections of $\Gamma(\V_{g_\chi})$,
preserving the metrics $g^\sharp_0$, $g^\sharp_\chi$. 
Then the required Proca bundle isomorphism ${\kappa}_{g'g} ={\kappa}_{g_1g_0}$ is defined  by composition:
\begin{equation*}
\kappa_{1,0}=\kappa_{1 \chi}\kappa_{\chi 0}.
\end{equation*}
where $\kappa_{1 \chi}$ from $\Gamma(\V_{g_\chi})$ to $\Gamma(\V_{g_1})$ is defined analogously to $\kappa_{\chi 0}$. 
The general case $g\simeq g'$ can be defined by  composing the  fiber preserving linear isometries $\kappa_{g_{k+1}g_k}$ or $\kappa_{g_k,g_{k+1}}$.

\subsection{Klein-Gordon operator associated to a Proca operator and Green operators}
We pass to tackle the issue of normal hyperbolicity of $\P$. As we shall see here,  it is not really necessary to construct the M\o ller maps, and the weaker requirement of {\em Green hyperbolicity} is sufficient.

 Let  $\N$ be  the {\bf Klein-Gordon operator} associated to the Proca operator $\P$ (\ref{PROCAOP}) acting on $1$-forms
\begin{equation}\N := \delta d  + d \delta + m^2 : \Gamma(\V_g) \to  \Gamma(\V_g)\:.\label{NtoP}\end{equation}
Notice that  this operator is {\bf normally hyperbolic}:  its principal symbol $\sigma_\N$ satisfies
	\begin{equation}
 \sigma_\N(\xi)=-g^\sharp(\xi,\xi)\,\Id_{\V_g} \quad \mbox{for all $\xi\in\T^*\M$, where $\Id_{\V_g}$ is the identity automorphism of $\V_g$.} \label{NOP}
	\end{equation}
Therefore the Cauchy problem for $\N$ is well-posed \cite{Ba,BaGi}.
Both $\N$ and $\P$ are formally selfadjoint with respect to the Hodge  scalar product (\ref{hodge}) on $\Omega_c^1(\M)= \Gamma_c(\V_g)$.

Since $m^2>0$ and $\delta_g^{(1)} \delta_g^{(2)}=0$, it is easy to prove that the Proca equation  (\ref{ProcaA}) is {\em equivalent} to the pair of equations
\begin{align}
\N A &= 0\:, \quad \mbox{for $A\in \Gamma_{sc}(\V_g)$} \label{PROCAEQ1}\:,\\
\delta A &= 0\:. \label{PROCAEQ2}
\end{align}
As already noticed, differently from $\N$, the Proca operator is not normally hyperbolic. However, it is {\bf Green hyperbolic} \cite{Ba, BaGi,aqft1} as $\N$, in particular there exist linear  maps, dubbed {\bf advanced Green operator} 
	$\G_\P^+\colon \Gamma_{pc}(\V_g)  \to  \Gamma(\V_g)$  and {\bf retarded Green operator} $\G_\P^-\colon  \Gamma_{fc}(\V_g)  \to  \Gamma(\V_g)$ uniquely defined by the requirements
	\begin{itemize}
		\item[(i.a)] $\G_\P^+ \circ \P\, \ff  = \P \circ \G_\P^+ \ff=\ff$ for all $ \ff \in  \Gamma_{pc}(\V_g)$ ,
		\item[(ii.a)]  $\supp(\G_\P^+ \ff ) \subset J^+ (\supp \ff )$ for all $\ff \in  {\Gamma}_{pc} (\V_g)$;
		\item[(i.b)]  $\G_\P^- \circ \P\, \ff  = \P\circ \G_{\P}^- \ff=\ff$ for all $ \ff \in  \Gamma_{fc}(\V_g)$,
		\item[(ii.b)]  $\supp(\G_\P^- \ff ) \subset J^- (\supp \ff )$ for all $\ff \in  {\Gamma}_{fc} (\V_g)$;
	\end{itemize}
The {\bf causal propagator} of $\P$ is defined as 
\begin{equation} \G_{\P}:= \G_{\P}^+-\G_{\P}^-  : \Gamma_c(\V_g) \to \Gamma_{sc}(\V_g)\:.\label{CAUSALPROP}\end{equation} 
All these maps are also continuous with respect to the natural topologies of the definition spaces \cite{aqft1}.
As a matter of fact (see \cite[Proposition 3.19]{aqft1}  and also \cite{BaGi}), the  advanced and retarded Green operator  $\G^\pm_\P : \Gamma_{pc/fc} (\V_g)\to \Gamma_{pc/fc}(\V_g)$ can be written as
\begin{equation*} 
\G^\pm_\P :=  \left(\Id + \frac{d\delta}{m^2}\right)  \G^\pm_\N =  \G^\pm_\N \left(\Id + \frac{d\delta}{m^2}\right) 
\end{equation*}
where $\G^\pm_\N$  are the analogous Green operators for the Klein-Gordon operator $\N$. Therefore 
\begin{equation} 
\G_\P :=  \left(\Id + \frac{d\delta}{m^2}\right)  \G_\N =  \G_\N \left(\Id + \frac{d\delta}{m^2}\right) \label{GPGN}\:.
\end{equation}
The fact that $\P$ is normally hyperbolic can be proved just by checking that the operators above satisfy the requirements which define the Green operators as stated above, using the analogous properties for $\G^\pm_\N$.

Eq. (\ref{GPGN}) and the analogous properties for $\G_\N$
entail 
\begin{equation}
\G_\P(\Gamma_c(\V_g)) = \{ A \in \Gamma_{sc}(\V_g)\:|\: \P A =0\}\label{RAN}\:.
\end{equation}
Indeed, if $\P A=0$ then $\N A=0$ and $\delta A=0$. If $A\in \Gamma_{sc}(\V_g)$, \cite[Theorem 3.8]{Norm} implies $A= \G_\N\ff$ for some $\ff \in \Gamma_c(\V_g)$, so that $A =   \left(\Id + \frac{d\delta}{m^2}\right) A = \G_\P\ff$ as said.\\
Furthermore, 
\begin{equation} \Ker \G_\P = \{\P\fg \:|\: \fg \in \Gamma_c(\V_g)\}\:.
\label{cruciallemma0} \end{equation}
Indeed,  if $\P A=0$ then $m^{2}  \left(\Id + \frac{d\delta}{m^2}\right)\P A = \N A=0$. If $A\in \Gamma_{sc}(\V_g)$,  again  \cite[Theorem 3.8]{Norm} implies that $A = \N \ff$ for some  $\ff \in \Gamma_c(\V_g)$. Since we also know that $\delta A =0$, the form (\ref{NOP}) of $\N$ yields $A=\P \ff$.
On the other hand, if $A= \P \ff$  for some  $\ff \in \Gamma_c(\V_g)$, then  $\G_\P A= \G_\P^+ \ff-  \G_\P^- \ff = \ff-\ff=0$.

On account of~\cite[Proposition 3.6]{Norm}, for any  smooth function $\rho:\M\to (0,+\infty)$ also $\rho\P$ is Green hyperbolic and $\G^\pm_{\rho\P}=\G^\pm_\P \rho^{-1}$.

\subsection{Proca M{\o}ller maps}
A {\bf  smooth Cauchy time function} in  a globally hyperbolic spacetime $(\M,g)$ relaxes the notion of temporal Cauchy function, it is a  smooth  map $t: \M \to \mathbb{R}$ such that
$dt$ is everywhere  timelike and past directed, the 
 level surfaces of $t$ are
 smooth spacelike Cauchy surfaces and $(\M,g)$ is isometric to $(\mathbb{R}
 \times \Sigma, h)$. Here,  $t$ identifies with the natural coordinate on $\mathbb{R}$ and the Cauchy surfaces of $(\M,g)$ identify with the sets $\{t\} \times \Sigma$.

From now on we indicate by $\N_0$,  $\N_1$, $\N_\chi$ the Klein-Gordon operators (\ref{NtoP}) on $\M$ constructed out of $g_0$, $g_1$ and $g_\chi$ respectively, where the globally hyperbolic metric $g_\chi$ is defined  as in (\ref{gCHI}) (and thus  $g_0\preceq g_\chi\preceq g_1$   \cite[Theorem 2.18]{Norm}) and depends on the choice of a function $\chi \in C_0^\infty(\M, [0,1])$.
 Similarly, $\P_0$,  $\P_1$, $\P_\chi$ denote the Proca operators (\ref{PROCAOP}) on $\M$ constructed out of $g_0$, $g_1$ and $g_\chi$ respectively.

We can state the first technical result.
\noindent \begin{proposition} \label{propNN}

\begin{itemize} Let $g_0,g_1$ be globally hyperbolic metrics satisfying $g_0\preceq g_1$ and let be $\chi\in C^\infty(\M; [0,1])$. Choose
\item[(a)]  a smooth  Cauchy time $g_1$-function $t: \M \to \mathbb{R}$ and $\chi \in C^\infty(\M; [0,1])$ such that $\chi(p)=0$ if $t(p) <t_0$ and $\chi(p)=1$ if $t(p)>t_1$ for given $t_0< t_1$;
\item[(b)]
a pair of smooth functions $\rho, \rho' : \M \to (0,+\infty)$  such that $\rho(p) =1$ for $t(p) < t_0$ and $\rho'(p) = \rho(p) =1$ if $t(p)>t_1$.
(Notice that $\rho=\rho'=1$ constantly is allowed.)\\
\end{itemize}

Then the following facts are true where $g_\chi$ is defined as in (\ref{gCHI}).
\begin{itemize}
\item[(1)] The operators
\begin{align*}
  \R_+&: \Gamma(\V_{g_0}) \to \Gamma(\V_{g_\chi}) \qquad \R_+:=\kappa_{\chi 0} - \G_{\rho\P_\chi}^+\left(\rho \P_\chi \kappa_{\chi0}  -\kappa_{\chi0} \P_0\right),   \\
  \R_-&: \Gamma(\V_{g_\chi}) \to \Gamma(\V_{g_1}) \qquad \R_-:=\kappa_{1\chi} - \G_{\rho\P_1}^-  \left(\rho'\P_1\kappa_{1\chi} - \rho\kappa_{1\chi}\P_\chi\right)
  \end{align*}
are linear space isomorphisms, whose inverses are given by
\begin{align*}
\R_+^{-1}&: \Gamma(\V_{g_\chi})\to \Gamma(\V_{g_0}) \qquad \R_+^{-1}=\kappa_{0\chi}+\G_{\P_0}^+(\rho\kappa_{0\chi}  \P_\chi- \P_0\rho\kappa_{0\chi}),\\
  \R_-^{-1}&: \Gamma(\V_{g_1})\to\Gamma(\V_{g_\chi}) \qquad  \R_-^{-1}:=\kappa_{\chi 1}+\G_{\rho\P_\chi}^- \left(\rho'\kappa_{\chi 1} \P_1- \rho\kappa_{1\chi}\P_\chi\right).
   \end{align*}
  By composition we define the \textbf{M{\o}ller operator}:
  $$\R:\Gamma(\V_{g_0})\to\Gamma(\V_{g_1}) \qquad \R:=\R_-\circ \R_+ ,$$
  which is also a linear space isomorphism.
\item[(2)] It holds 
\begin{equation*}\rho \kappa_{0\chi} \P_\chi \R_+ = \P_0\: \qquad \text{ and }  \:\qquad \rho' \kappa_{\chi1}\P_1 \R_- = \rho \P_\chi\,.
 \end{equation*}
and also
\begin{equation*} \rho \kappa_{0\chi} \P_\chi  = \P_0\R_+^{-1}\: \qquad \text{ and }  \:\qquad \rho' \kappa_{\chi1}\P_1=  \P_\chi \R_- ^{-1}\,.
 \end{equation*}
 \item[(3)]  If $\ff \in \Gamma(\V_{g_0})$ or $\Gamma(\V_{g_{\chi}})$ respectively, then 
\begin{align}
(\R_+ \ff)(p) &= \ff(p) \quad \mbox{for} \quad t(p) < t_0,   \label{pastequal}\\
(\R_- \ff)(p) &= \ff(p) \quad \mbox{for} \quad t(p) > t_1 \label{futureequal}\:.
   \end{align}
\end{itemize}

\end{proposition}
\begin{proof}
First of all, we notice that the operator $\R_{+}$ is well defined on the whole space $\Gamma(\V_{g_0})$ since for all sections $\ff\in\Gamma(\V_{g_0})$ we have that $( \P_\chi\frac{\kappa_{\chi0}}{\rho}  -\frac{\kappa_{\chi0}}{\rho} \P_0)\ff\in\Gamma_{pc}(\V_{g_1})$: indeed by definition, there exists a $t_0\in\RR$ such that on $t^{-1}(-\infty,t_0)$ and we have that 
$\P_{\chi}=\P_0$, $\kappa_{\chi 0}=\Id$ and $t$ is a smooth $g_1$-Cauchy time function. Moreover, since $g_{\chi}\preceq g_1$ it follows that $\Gamma_{pc}(\V_{g_1})\subset \Gamma_{pc}(\V_{g_{\chi}})=\D(G_{\P_{\chi}})$.

To prove (1), we can first notice that 
\begin{align*}
\R_+^{-1}\circ\R_+&=\left(\kappa_{0\chi}+\G_{\P_0}^+(\rho\kappa_{0\chi}  \P_\chi- \P_0\kappa_{0\chi})\right) \circ
 \left(\kappa_{\chi 0} - \G_{\rho\P_\chi}^+(\rho \P_\chi\kappa_{\chi0} -\kappa_{\chi0} \P_0)\right)\\
&=\Id-\kappa_{0\chi}\G_{\rho\P_\chi}^+\left(\rho \P_\chi\kappa_{\chi0} -\kappa_{\chi0} \P_0\right)+\G_{\P_0}^+(\rho\kappa_{0\chi}  \P_\chi- \P_0\kappa_{0\chi})\kappa_{\chi 0}\\
&-\G_{\P_0}^+(\rho\kappa_{0\chi}  \P_\chi- \P_0\rho\kappa_{0\chi})\G_{\rho\P_\chi}^+\left(\rho \P_\chi\kappa_{\chi0} -\kappa_{\chi0} \P_0\right).
\end{align*} 
\noindent To conclude it is enough to show that everything cancels out except the identity operator, but that just follows by using basic properties of Green operators and straightforward algebraic steps. We easily see that the last addend can be recast as:
\begin{align*}
&\G_{\P_0}^+\left(\rho\kappa_{0\chi}  \P_\chi- \P_0\kappa_{0\chi}\right)\G_{\rho\P_\chi}^+\left(\rho \P_\chi\kappa_{\chi0} -\kappa_{\chi0} \P_0\right)\\
&=\G^+_{\P_0}\rho\kappa_{0 \chi}\P_{\chi}\G_{\rho\P_\chi}^+\left(\rho \P_\chi\kappa_{\chi0} -\kappa_{\chi0} \P_0\right)-\G^+_{\P_0}\P_0\kappa_{0 \chi}\G_{\rho\P_\chi}^+(\rho \P_\chi\kappa_{\chi0} -\kappa_{\chi0} \P_0)\\
&=\G^+_{\P_0}\kappa_{0 \chi}\left(\rho \P_\chi\kappa_{\chi0} -\kappa_{\chi0} \P_0\right)-\kappa_{0 \chi}\G_{\rho\P_\chi}^+(\rho \P_\chi\kappa_{\chi0} -\kappa_{\chi0} \P_0),
\end{align*}
\noindent which fulfills its purpose.\\
A specular computation proves that $\R_{+}^{-1}$ is also a right inverse. Almost identical reasonings prove that $\R_-^{-1}$ is a two sided inverse of $\R_-$ which is also well defined, then bijectivity of $\R$ is obvious.

(2) follows by the following direct computation
\begin{align*}
&\rho\kappa_{0 \chi}\P_{\chi}\R_{+}=\rho\kappa_{0 \chi}\P_{\chi}\left(\kappa_{\chi 0} - \G_{\rho\P_\chi}^+\left(\rho \P_\chi \kappa_{\chi0}  -\kappa_{\chi0} \P_0\right)\right)\\
&=\kappa_{0 \chi}\kappa_{\chi 0}\P_0=\P_0.
\end{align*}

(3)  Let us prove (\ref{pastequal}).  In the following $P^*$ denotes the formal dual operator of $P$ acting on the sections of the dual bundle $\Gamma_c(V^*_g)$. 
If $\ff' \in \Gamma_c(\V_g^*)$ and  $\ff \in \Gamma_{pc}(\V_g)$ or 
$\ff \in \Gamma_{fc}(V_g)$ respectively,
\begin{equation}\label{eqGstar}
\int_\M  \langle\G_{P^*}^{-}\ff',\ff\rangle \: \mbox{vol}_g = \int_\M \langle\ff', \G_{P}^+\ff\rangle \: \mbox{vol}_g \:,\quad
\int_\M \langle\G_{P^*}^{+}\ff',\ff\rangle \: \mbox{vol}_g = \int_\M \langle\ff', \G_{P}^- \ff \rangle \: \mbox{vol}_g \:,
 \end{equation}
where 
$\G_{P}^{\pm}$ indicate the Green operators of $\P$ and $\G_{P^*}^{\pm}$ indicate the Green operators of $\P^*$.
Consider now a compactly supported  smooth section $\fh$ whose support is included in the set $t^{-1}((-\infty, t_0))$. Taking advantage of the Equation~\eqref{eqGstar},
 we obtain
$$\int_\M \langle \fh , \G_{\rho \P_\chi}^+ (\rho \P_\chi-\P_0) \ff \rangle \: \mbox{vol}_{g_\chi} = 
\int_\M \langle \G_{(\rho \P_\chi)^*}^- \fh ,  (\rho \P_\chi-\P_0) \ff \rangle \: \mbox{vol}_{g_\chi} =0$$
since $\mbox{supp}( \G_{(\rho \P_\chi)^*}^-\fh) \subset J^{g_\chi}_-(\mbox{supp}(\fh))$  and thus that support  does not meet $\mbox{supp}((\rho \P_\chi-\P_0) \ff)$ because 
$((\rho \P_\chi-\P_0) \ff)(p)$ vanishes if $t(p)< t_0$.   As  $\fh$ is an arbitrary smooth section compactly  supported in $t^{-1}((-\infty,t_0))$,
$$\int_\M \langle \fh, \G_{\rho \P_\chi}^+ (\rho \P_\chi-\P_0) \ff \rangle \: \mbox{vol}_{g_\chi} = 0$$
 entails  that $\G_{\rho \P_\chi}^+ (\rho \P_\chi-\P_0) \ff =0$ on $t^{-1}((-\infty,t_0))$. The proof of (\ref{futureequal}) is strictly analogous, so we leave it to the reader.
\end{proof}

Using Proposition~\ref{propNN}, we can pass to the generic case $g\simeq g'$.

\begin{theorem}\label{remchain}
Let $(\M,g)$ and $(\M,g')$ be globally hyperbolic spacetimes,   with associated Proca bundles  $\V_{g}$ and $\V_{g'}$ and Proca operators   $\P, \P'$.\\
If  $g \simeq g'$,
 then there exist (infinitely many) vector space isomorphisms,
 \begin{equation*}
 \R  : \Gamma(\V_{g}) \to \Gamma(\V_{g'})
  \end{equation*}
 such that 
 \begin{itemize}
 \item[(1)]
referring to the said domains,  
$$\mu \kappa_{gg'} \P'\R = \P $$
for some smooth $\mu: \M \to (0,+\infty)$ (which can always be chosen $\mu=1$ constantly in particular), and a smooth  fiberwise isometry
$\kappa_{gg'}:\Gamma(\V_{g'})\to\Gamma(\V_{g})$.
\item[(2)] 
The restriction, called {\bf M\o ller map}
\begin{align*}
 \S^0 := \R|_{ \Ker_{sc}(\P)}  :  \Ker_{sc}(\P) \to   \Ker_{sc}(\P')  
 \end{align*}
is  well-defined vector space isomorphism with inverse  given by
\begin{equation*} (\S^0)^{-1} := \R^{-1}|_{ \Ker_{sc}(\P')}  :  \Ker_{sc}(\P') \to   \Ker_{sc}(\P) \,.
\end{equation*}
%
%
 \end{itemize}
\end{theorem}

\begin{proof}  Since $g\simeq g'$, there exists a finite sequence of globally hyperbolic metrics $g_0=g,g_1,..,g_N=g$ such that at each step $g_k\preceq g_{k+1}$ or $g_{k+1}\preceq g_k$. For all $k\in\{0,..,N\}$ we can associate to the metric a Proca operator $\P_k$.\\
At each step the hypotheses of Proposition~\ref{propNN} are verified, in fact we can choose functions $\rho_k$ and $\rho'_k$ and the M{\o}ller map is given by $\R_k={\R_k}_-\circ{\R_k}_+$. The general map is then built straightforwardly by composing the $N$ maps constructed step by step:
$$\R=\R_N\circ...\circ\R_1.$$
Regarding (1), by direct calculation we get that $\mu=\prod_{k=1}^{N}\rho'_k$, while $\kappa_{gg'}=\kappa_{g_0 g_1}\circ...\circ \kappa_{g_{N-1} g_N}$. The proof of (2) is trivial.
\end{proof}

\subsection{Causal propagator and M{\o}ller operator}
The rest of this section is devoted to study the relation between M\o ller maps and the causal propagator of Proca operators. To this end, we  use a natural extension of the notion of \textit{adjoint operator} introduced in~\cite[Section 4.5]{Norm}.

Let $g$ and $g'$ (possibly $g\neq g'$) globally hyperbolic metric and let $\V_g$ and $\V_{g'}$ be a Proca bundle on the manifold $\M$.
Consider a $\RR$-linear operator 
$$\T : \D(\T) \to \Gamma(\V_{g'})\:,$$
where $\D(\T) \subset \Gamma(\V_g)$ is a $\RR$-linear subspace and $\D(\T) \supset \Gamma_c(\V_g)$. 
\begin{definition} \label{defadjoint}
An operator $$\T^{\dagger_{gg'}} : \Gamma_c(\V_{g'}) \to \Gamma_c(\V_g)$$ is said to be the {\bf adjoint of $\T$ with respect to  $g, g'$} (with the said order) if it satisfies 
\begin{equation*} 
\int_\M g'^\sharp\left({\fh},{\T\ff}\right)(x)\, \vol_{g'}(x) = \int_\M g^\sharp\left(\T^{\dagger_{gg'}} \fh,\ff\right)(x)\,  \vol_{g}(x)\quad 
\forall \ff \in \D(\T)\:,\: \forall \fh \in \Gamma_c(\E).
\end{equation*}
When $g=g'$, we use the simplified notation $\T^\dagger := \T^{\dagger_{gg}}$.
\end{definition}
As in~\cite{Norm}, the adjoint operator satisfies a lot of useful properties which we summarize in the following proposition. Since the proof is analogous to the one of~\cite[Proposition 4.11]{Norm}, we leave it to the reader.  Though the rest of this paper deal with the real case only, we state the theorem encompassing the case  where the sections are complex and the fiber scalar product is made Hermitian by adding a complex conjugation of the left entry in the usual fiberwise real $g^\sharp$ inner product, which becomes $g^\sharp(\overline{\ff},\fg)$, where the bar denotes the complex conjugation. Definition \ref{defadjoint} extends accordingly.
  For this reason $\mathbb{K}$ will denote either $\mathbb{R}$ or $\mathbb{C}$, and the complex conjugate $\overline{c}$ reduces to $c$ itself when $\mathbb{K}= \mathbb{R}$.  We keep the notation $\V_g$ for indicating either the  real or complex vector bundle $\T^*\M$ or $\T^*\M + i\T^*\M$ corresponding to two possible  choices of $\KK$.

\begin{proposition}\label{propadjoint}
	Referring to the notion of adjoint in Definition \ref{defadjoint}, the following facts are valid.
	\begin{itemize}
		\item[(1)] If the adjoint $\T^{\dagger_{gg'}}$ of $\T$  exists, then it is unique.
		\item[(2)] If $\T : \Gamma(\V_g) \to \Gamma(\V_{g'})$ is a differential operator and $g=g'$, then $\T^{\dagger_{gg}}$ exists and is the restriction of the formal adjoint to $\Gamma_c(\E)$. (In turn, the formal adjoint of $\T$ is  the unique extension to $\Gamma(\E)$ of the differential operator $\T^\dagger$ as a differential operator.)
		\item[(3)] 
		Consider a pair of $\KK$-linear operators  $\T : \D(\T) \to \Gamma(\V_{g'})$,  $\T' : \D(\T') \to \Gamma(\V_{g'})$ with $\D(\T),\D(\T')\subset \Gamma(\V_{g})$
		and $a,b \in \KK$.  Then 
		$$(a\T+b\T')^{\dagger_{gg'}} = \overline{a} \T^{\dagger_{gg'}} + \overline{b} \T'^{\dagger_{gg'}}$$
		provided 
		$\T^{\dagger_{gg'}}$ and $\T'^{\dagger_{gg'}}$ exist.
		\item[(4)] Consider a pair of $\KK$-linear operators  $\T : \D(\T) \to \Gamma(\V_{g'})$,  $\T' : \D(\T') \to \Gamma(\V_{g''})$ with $\D(\T)\subset \Gamma(\V_{g})$ and $\D(\T')\subset\Gamma(\V_{g'})$ such that 
		\begin{itemize}
			\item[(i)]  $\D(\T'\circ \T) \supset \Gamma_c(\V_{g})$,
			\item[(ii)] $\T^{\dagger_{gg'}}$ and $\T'^{\dagger_{g'g''}}$ exist,
		\end{itemize} 
		then $(\T'\circ \T)^{\dagger_{gg''}}$  exists and
		$$(\T'\circ \T)^{\dagger_{gg''}} = \T^{\dagger_{gg'}} \circ \T'^{\dagger_{g'g''}}\:.$$
		\item[(5)] If $\T^{\dagger_{gg'}}$ exists, then $(\T^{\dagger_{gg'}})^{\dagger_{g'g}} = \T|_{\Gamma_c(\V_{g})}$.
		\item[(6)] If $\T : \D(T)= \Gamma(\V_g) \to \Gamma(\V_{g'})$ is bijective,  admits $\T^{\dagger_{gg'}}$, and $\T^{-1}$ admits $(\T^{-1})^{\dagger_{g'g}}$, then $\T^{\dagger_{gg'}}$ is bijective and 
		$(\T^{-1})^{\dagger_{g'g}} = (\T^{\dagger_{gg'}})^{-1}$.
	\end{itemize}
\end{proposition}

Now we are ready to prove that the operators  $\R$ admit adjoints and we explicitly compute them.

\begin{proposition}\label{MollerAdj}
Let $g_0,g_1$ be globally hyperbolic metrics satisfying $g_0\preceq g_1$.
Let $\R_{+}$, $\R_{-}$ and $\R$ be the operators defined in Proposition~\ref{propNN} and fix, once and for all, $\rho=c_{0}^{\chi}$ and $\rho'=c_{0}^{1}$ where $c_{0}^{\chi}$, $c_{0}^{1}$ are the unique smooth functions on $\M$ such that:
\begin{equation}\label{volumes}
\vol_{g_{\chi}}=c_0^{\chi}\vol_{g_0} \qquad \vol_{g_{1}}=c_0^{1}\vol_{g_0}.
\end{equation}
Then we have:
\begin{itemize}\label{adjoints}
\item[(1)] $\R_+^{\dagger_{g_0g_{\chi}}}:\Gamma_c(\V_{g_{\chi}})\to\Gamma_c(\V_{g_0})$ satisfies:
$$\R_+^{\dagger_{g_0g_{\chi}}}=\left(c_{0}^{\chi}\kappa_{0 \chi}-\left(c_0^{\chi}\kappa_{0\chi} \P_{\chi}-\P_0\kappa_{0\chi}\right)\G_{\P_\chi}^-\right)|_{\Gamma_c (\V_{\chi})}$$
and can be recast in the form $$\R_+^{\dagger_{g_0g_{\chi}}}=\P_0\kappa_{0 \chi}\G_{\P_\chi}^-|_{\Gamma_c (\V_{\chi})}.$$
\item[(2)]
$\R_-^{\dagger_{g_\chi g_1}}:\Gamma_c(\V_{g_{1}})\to\Gamma_c(\V_{g_{\chi}})$ satisfies
$$\R_-^{\dagger_{g_\chi g_1}}=\left(c_1^{\chi}\kappa_{\chi 1}-\left(c_1^{\chi}\kappa_{\chi 1}\P_1-\P_{\chi}\kappa_{\chi 1}\right)\G_{\P_1}^+\right)|_{\Gamma_c(\V_1)},$$
and can be recast in the form $$\R_-^{\dagger_{g_\chi g_1}}=P_\chi\kappa_{\chi 1}\G_{\P_1}^+|_{\Gamma_c(\V_1)}.$$
\item[(3)] The map $\R^{\dagger_{g_0 g_1}} : \Gamma_c(\V_{g_1})\to\Gamma_c(\V_{g_0})$ defined by  $\R^{\dagger_{g_0g_1}}:=\R_+^{\dagger_{g_0g_{\chi}}}\circ\R_-^{\dagger_{g_\chi g_1}}$
is invertible and 
$$(\R^{\dagger_{g_0 g_1}})^{-1} = (\R^{-1})^{\dagger_{g_1 g_0}}: \Gamma_c(\V_{g_1}) \to \Gamma_c(\V_{g_0})\,.$$
We call it \textbf{adjoint M{\o}ller operator}.\\
 Moreover $\R^{\dagger_{g_0g_1}}$ is a homeomorphism with respect to the natural (test section)  topologies of the domain and of the co-domain.\\

\end{itemize}
\end{proposition}

\begin{proof}
We start by proving points (1) and (2).
For any $\ff \in \D(\R_+)=\Gamma(\V_{g_0})$ and $  \fh \in \Gamma_c(\V_{g_{\chi}})$ we have
\begin{align*}
		&&\int_\M g_\chi^\sharp\left(\fh,\R_+\ff\right) \vol_{g_{\chi}}=
		 \int_\M g_\chi^\sharp\left(\fh , \big(\kappa_{\chi 0} - \G_{c_0^{\chi} \P_\chi}^+\left(c_0^{\chi}  \P_\chi \kappa_{\chi0}  -\kappa_{\chi0} \P_0\right)\big)\ff\right) \vol_{g_{\chi}}=\\
		 && \int_\M g_\chi^\sharp\left(\fh,\kappa_{\chi 0}\ff\right) \vol_{g_{\chi}}- \int_\M g_\chi^\sharp\left({\fh},{\big(\G_{c_0^{\chi} \P_\chi}^+\left(c_0^{\chi}  \P_\chi \kappa_{\chi0}  -\kappa_{\chi0} \P_0\right)\big)\ff}\right) \vol_{g_{\chi}}.
\end{align*}

\noindent We now split the problem and compute the adjoint of the two summands separately.\\
The adjoint of the first one follows immediately by exploiting the properties of the existing isometry and Equations~\eqref{volumes}
\begin{align*}
\int_\M g_\chi^\sharp\left({\fh},{\kappa_{\chi 0}\ff} \right) \vol_{g_{\chi}}=\int_\M g_{0}^\sharp\left({c_0^{\chi}\kappa_{0 \chi}\fh},{\ff}\right) \vol_{g_{0}}.
\end{align*}
For the second summand the situation is trickier and we cannot split the calculation in two more summands since it is crucial that the whole difference $\left(c_0^{\chi}  \P_\chi \kappa_{\chi0}  -\kappa_{\chi0} \P_0\right)$ acts on a general $\ff\in\Gamma(\V_{g_{\chi}})$ before we apply the Green operator whose domain is $\Gamma_{pc}(\V_{g_\chi})$.\\
So we first rewrite $\G_{c_0^{\chi} \P_\chi}^+=\G_{\P_\chi}^+ \frac{1}{c_0^{\chi}}$ use the properties of standard adjoints of Green operators for formally self-adjoint Green hyperbolic differential operators to get
\begin{align*}
	\int_\M g_\chi^\sharp\left({\fh},{\big(\G_{c_0^{\chi} \P_\chi}^+\left(c_0^{\chi}  \P_\chi \kappa_{\chi0}  -\kappa_{\chi0} \P_0\right)\big)\ff}\right) \vol_{g_{\chi}}=\int_\M g_\chi^\sharp\left({\G_{\P_\chi}^-\fh},{\big( \P_\chi \kappa_{\chi0}  -\frac{\kappa_{\chi0}}{c_0^{\chi}} \P_0\big)\ff}\right) \vol_{g_{\chi}}.
\end{align*}
Now we are tempted to exploit the linearity of the integral and of the fiber product, but first, to ensure that the two integrals individually converge, we need to introduce a cutoff function:
\begin{itemize}
	\item We notice again that there is a Cauchy surface of the foliation $\Sigma_{t_0}$ such that for all leaves with $t<t_0$ the operator $\left( \P_\chi \kappa_{\chi0}  -\frac{\kappa_{\chi0}}{c_0^{\chi}} \P_0\right)=0$;
	\item So take a $t'<t_0$ and define a cutoff smooth function $s:\M\to[0,1]$ such that $s=0$ on all leaves with $t<t'$.
\end{itemize}

\noindent In this way we are allowed to rewrite our last integral and split it in two convergent summands without modifying its numerical value.

\begin{align*}
\int_\M g_\chi^\sharp\Big({\G_{\P_\chi}^-\fh},  \big( \P_\chi  \kappa_{\chi0} & -\frac{\kappa_{\chi0}}{c_0^{\chi}} \P_0\big)s\ff \Big) \vol_{g_{\chi}} = \\
 =& \int_\M g_\chi^\sharp\left({\G_{\P_\chi}^-\fh},{ \P_\chi \kappa_{\chi0}  s\ff}\right) \vol_{g_{\chi}} -\int_\M g_\chi^\sharp\left({\G_{\P_\chi}^-\fh},
{\frac{\kappa_{\chi0}}{c_0^{\chi}} \P_0 s\ff}\right) \vol_{g_{\chi}}\\
=&\int_\M \g_0^\sharp\left({c_0^{\chi}\kappa_{0 \chi}\P_{\chi}\G_{\P_\chi}^-\fh},{  s\ff}\right) \vol_{g_{0}}-\int_\M \g_0^\sharp
(\P_0 \kappa_{0 \chi}\G_{\P_\chi}^-\fh,s\ff) \vol_{g_{0}}\\
=&\int_\M \g_0^\sharp\left({\big(c_0^{\chi}\kappa_{0\chi} \P_{\chi}-\P_0\kappa_{0\chi}\big)\G_{\P_\chi}^-\fh},{s\ff}\right) \vol_{g_{0}}\\
=&\int_\M \g_0^\sharp\left({\big(c_0^{\chi}\kappa_{0\chi} \P_{\chi}-\P_0\kappa_{0\chi}\big)\G_{\P_\chi}^-\fh},{\ff} \right) \vol_{g_{0}}.
\end{align*}
\noindent where in the last identities we have used properties of the standard adjoints of the formally self-adjoint operators, of the isometries and of the cutoff function. \\
Since the domain of the operator is just made up of compactly supported sections, we may exploit the inverse property of the Green operators to immediately obtain that
$$
c_{0}^{\chi}\kappa_{0 \chi}-\left(c_0^{\chi}\kappa_{0\chi} \P_{\chi}-\P_0\kappa_{0\chi}\right)\G_{\P_\chi}^-|_{\Gamma_c (\V_{\chi})}=\P_0\kappa_{0 \chi}\G_{\P_\chi}^-|_{\Gamma_c (\V_{\chi})}.
$$
To see that the image of the operators is indeed compactly supported we can focus on $\R^{\dagger_{g_0g_{\chi}}}$, the rest follows straightforwardly. The first summand $c_0^{\chi}\kappa_{0 \chi}$ does not modify the support of the sections, whereas the second does. Let us fix $\ff\in\Gamma_c(\V_{g_{\chi}})$, then $\supp(\G^-_{\P_{\chi}} \ff)\subset J_{g_{\chi}}^-(\supp\ff)$ which means that $\G^-_{\P_{\chi}} \ff\in\Gamma_{sfc}$, i.e it is space-like and future compact. The thesis follows by again observing that there is a Cauchy surface such that in its past $\left( \P_\chi \kappa_{\chi0}  -\frac{\kappa_{\chi0}}{c_0^{\chi}} \P_0\right)\G^-_{\P_{\chi}} \ff=0$.

\noindent The computation of the adjoint of $\R_-$ is almost identical to the one just performed.

The first part of (3) is an immediate consequence of property (4) in Proposition~\ref{propadjoint}, while the invertibility of the adjoint can be proved by explicitly showing that  the operator
$$(\R_+^{\dagger_{g_0 g_{\chi}}})^{-1}=\left(\frac{\kappa_{\chi 0}}{c_0^{\chi}}+\left(\P_{\chi}\kappa_{\chi 0}-\frac{\kappa_{\chi 0}}{c_0^{\chi}}\G^-_{\P_{0}}\right)\right)\Big{|}_{\Gamma_c(\V_{g_0})}$$
serves as a left and right inverse of $\R_+^{\dagger_{g_0 g_{\chi}}}$. An analogous argument can be used for $\R_-^{\dagger_{g_{\chi}g_{1}}}$.\\
The continuity of both the adjoint and its inverse comes by the same arguments used in the proof of \cite[Theorem 4.12]{Norm} (with the only immaterial difference that this time the smooth isometry $\kappa_{\chi0}$ is included in the definition of the M{\o}ller operator.)
\end{proof}
\begin{remark}
An interesting fact to remark is that having defined the adjoints over compactly supported sections makes the dependence on the auxiliary volume fixing functions disappear.
\end{remark}

 We conclude the section, by proving the second part of Theorem~\ref{thm:main intro 1}. 

\begin{theorem} \label{causalpropR}Let $(\M,g)$ and $(\M,g')$ be globally hyperbolic spacetimes,   with associated Proca bundles  $\V_{g}$ and $\V_{g'}$ and Proca operators   $\P, \P'$.\\
	If $g\simeq g'$,  it is  possible to specialize  the  $\RR$-vector space isomorphism $\R : \Gamma(\V_{g}) \to \Gamma(\V_{g'})$
of Proposition \ref{remchain}
 such that the following further facts are true.
	\begin{itemize}
		\item[(1)]   The causal propagators 
		$\G_{\P}$ and $\G_{\P'}$ (\ref{CAUSALPROP}), respectively
		of $\P$ and $\P'$, satisfy
		\begin{equation*} \R \G_{\P}  \R^{\dagger_{g g'}} = \G_{\P'}\:. 
		\end{equation*}
	
		\item[(2)] It holds 
		\begin{equation*}
		\R^{\dagger_{g g'}} \P'\kappa_{g'g}|_{\Gamma_c(\V_{g})}= \P|_{\Gamma_c(\V_g)}\: .
		\end{equation*}

	\end{itemize}
$\R$ as above is called {\bf M\o ller operator} of $g,g'$ (with this order).
\end{theorem}
\begin{proof}
Since $g\simeq g'$ and the M\o ller map is defined as the composition 
$\R=\R_N\circ...\circ\R_1$, we can use properties (4) in Proposition~\ref{propadjoint} and reduce to the case where $g=g_0 \preceq g_1=g'$. With this assumption, (2) can be obtained following the proof of Proposition~\ref{propNN} and (3) is identical to \cite[Theorem 4.12 (5)]{Norm}. So we leave it to the reader.

It remains to prove (1).  Decomposing $\R$ as above, we  define the maps $\R^{g_0g_\chi}_\pm$, $\R^{g_\chi g_1}_\pm$ by choosing   the various arbitrary functions   as in Proposition \ref{MollerAdj}. We first notice
\begin{gather*}
\R_+ \G^+_{\P_0}\R_+^{\dagger_{g_0 g_{\chi}}}=\left(\kappa_{\chi 0} - \G_{c_0^{\chi} \P_\chi}^+\left(c_0^{\chi}  \P_\chi \kappa_{\chi0}  -\kappa_{\chi0} \P_0\right)\right)\G^+_{\P_0}\left( \P_0\kappa_{0 \chi}\G_{\P_\chi}^-\right)|_{\Gamma_c (\V_{\chi})}\\
=\G_{c_0^{\chi} \P_\chi}^+\kappa_{\chi0}\left( \P_0\kappa_{0 \chi}\G_{\P_\chi}^-\right)|_{\Gamma_c (\V_{\chi})}=\G^+_{\P_{\chi}}-\G^+_{\P_{\chi}}\left(\P_{\chi}-\frac{\kappa_{\chi 0}}{c_0^{\chi}}\P_0\kappa_{0 \chi}\right)\G^-_{\P_{\chi}}.
\end{gather*}
where the first equality follows by definition, in the second one we have used the properties of Green operators, while in the third one we have just equated the two expressions for the adjoint operator according to (1) in Proposition~\ref{adjoints} and performed some trivial algebraic manipulations.\\
Another analogous computation can be performed for the retarded Green operator yielding
\begin{gather*}
\R_+ \G^+_{\P_0}\R_+^{\dagger_{g_0 g_{\chi}}}=\G^-_{\P_{\chi}}-\G^+_{\P_{\chi}}\left(\P_{\chi}-\frac{\kappa_{\chi 0}}{c_0^{\chi}}\P_0\kappa_{0 \chi}\right)\G^-_{\P_{\chi}}.
\end{gather*}
Therefore, subtracting the two terms we get
$$\R_+ \G_{\P_0}\R_+^{\dagger_{g_0 g_{\chi}}}=\R_+ (\G^+_{\P_0}-\G^-_{\P_0})\R_+^{\dagger_{g_0 g_{\chi}}}=\G_{\P_{\chi}}.$$
Applying now $\R_-$ and its adjoint we get the claimed result.
\end{proof}

\section{M\o ller $*$-Isomorphisms and Hadamard States}\label{SECALGEBRAS}

\subsection{The CCR algebra of observables of the Proca field}
We now pass to introduce the algebraic formalism to quantize the Proca field \cite{Fewster,koProca}.

Let $(\M,g)$ be a globally hyperbolic spacetime, $\V_g$ be a  Proca bundle  and  denote by $\P:\Gamma(\V_g)\to\Gamma(\V_g)$ the Proca operator. Following~\cite{IgorValter}, we call \textbf{on-shell Proca $CCR$ $*$-algebra}, the $*$-algebra defined as 
$$\mathcal{A}_g=\mathfrak{A}_g / \mathfrak{I}_g$$ where:
\begin{itemize}
 \item[-] $\mathfrak{A}_g$ is the  free complex unital algebra finitely  generated by the set of abstract elements $\II$ (the unit element),  $\fa(\ff)$ and $\fa(\ff)^*$ for all $\ff\in\Gamma_c (\V_g)$, and  endowed with the unique  (antilinear) $*$-involution which associates $\fa(\ff)$ to $\fa(\ff)^*$and satisfies   $\II^* =\II$ and $(ab)^*= b^*a^*$.
 \item[-]  $\mathfrak{I}_g$ is the two-sided $*$-ideal generated by the following elements of $\mathfrak{A}_f$:
\begin{enumerate}
	\item $\fa(a\ff+b\fh)-a\fa(\ff)-b\fa(\fh)\:, \quad \forall a,b\in\RR \quad \forall \ff, \fh\in\Gamma_c(\V_g)$;
	\item $\fa(\ff)^*-\fa(\ff)\:,\quad \forall \ff\in\Gamma_c(\V_g)$;
	\item $\fa(\ff)\fa(\fh)-\fa(\fh)\fa(\ff)-i\G_{\P}(\ff,\fh)\II\:,  \quad \forall \ff,\fh\in\Gamma_c(\V_g)$;
	\item $\fa(\P \ff)$\:, \quad $\forall \ff\in\Gamma_c(\V_g)$.
\end{enumerate}

\end{itemize}
 The four entries of the list respectively implement linearity, hermiticity of the generators, canonical commutation relations and the equations of motion for the quantum field.

\remark\label{reminterp}  As in \cite{Fewster}, we adopt in this paper the interpretation of $\fa(\ff)$ is $( \fa|\ff )$, where the pairing is the Hodge inner  product of $1$-forms (\ref{hodge}). \\
\noindent An equivalence class in $\mathcal{A}_g$ is denoted by $[\fa(\ff)]=\hat{\fa}(\ff)$, the equivalence class corresponding to the identity is denoted by $[\II]=\Id$.
The hermitian elements of the algebra $\mathcal{A}_g$ are called \textbf{observables}.

\remark Requirement 4, when we pass to the quotient algebra corresponds to the distributional relation $\P\hat{\fa}=0$, when we take Remark \ref{reminterp} into account  and the fact that $\P$ is formally  selfadjoint. Since every solution of the Proca equation is a co-closed solution of the Klein-Gordon equation and {\em vice versa}, we conclude that 
$\delta \hat{\fa}=0$, i.e. $\hat{\fa}(d\ff)=0$ for every $\ff \in \Gamma_c(\V_g)$, must be valid.\\
If, moreover, we deprive the ideal $\mathfrak{I}_g$ of the generators in 4, the quotient algebra is said to be \textbf{off-shell}, however it would still be convenient to assume the-closedness constraint when defining the off-shell algebra. That is  when defining  the relevant ideal of the free off-shell algebra,  we should keep 1-3, we should  drop 4,  and we should replace it with the weaker condition
 \begin{enumerate}
\item[4'.] $\hat{\fa}(d \ff)$\:, \quad $\forall \ff\in\Gamma_c(\V_g)$.
 \end{enumerate}
{\em This work however deals with the on-shell algebra only, we shall  indicate by $\mathcal{A}_g$ throughout. A study of the off-shell algebra, which may result in some relevance in perturbative renormalization procedure  will be done elsewhere.}

\subsection{M\o ller $*$-isomorphism and Hadamard states}

From now on let $X$ be a topological vector space, we indicate by $X'$ its topological dual. For example $\Gamma'_c(V_g)$ represents the space of distributions acting on compactly supported test functions, and shall not be confused with the space of compactly supported distributions.

Having built the $CCR$-algebra, the subsequent step in quantization consists in finding a way to associate numbers to the abstract operators in $\mathcal{A}_g$ by identifying a distinguished state. For sake of completeness, let us recall that a \textbf{state} over the Proca algebra $\mathcal{A}_g$ a $\CC$-linear functional $\omega:\mathcal{A}_g\to\CC$ which is
	\begin{itemize}
		\item[(i)]\textbf{Positive:} $\omega(a^*a)\geq0\quad\forall a\in\mathcal{A}_g$,
		\item[(ii)]\textbf{Normalized:} $\omega(\II)=1$
	\end{itemize}

A generic element of the $CCR$-algebras $\mathcal{A}_g$ of a quantum field can be written as a finite polynomial of the generators $\hat\fa(f)$, where the zero grade term is proportional to $\II$. To specify the action of a state it is sufficient to know its action on the monomials, i.e its \textbf{n-point functions}:
\begin{equation*}
  \omega_n(\ff_1,\dots,\ff_n) :=\omega(\hat\fa(\ff_1)\dots\hat\fa(\ff_n)) 
\end{equation*} 
with $\ff_1,\dots,\ff_n\in\Gamma_c(\V_g)$.

If we impose continuity with respect to the usual topology on the space of compactly supported test sections we can uniquely extend  the $n$-point functions to distributions 
in  $\Gamma_c'(\V_g^{n\boxtimes})$ we shall hereafter indicate by the symbol $\tilde\omega_n$. 

Among all possible states the physical ones are the so-called {\em quasifree} (or {\em Gaussian}) {\em Hadamard} states. Quasifree means that the $n$-point distributions of the state have a structure resembling the one of a free theory, i.e they all can be recovered just by knowing the two-point distribution.

\begin{definition}\label{quasifree} Consider the globally hyperbolic spacetime $(\M,g)$ and a  state $\omega:\mathcal{A}_g\to\CC$ 
for the Proca algebra of observables on $(\M,g)$. $\omega$ is  called
	 \textbf{quasifree},  if for all choices of $\ff_i \in \Gamma_c(V_g)$ \begin{itemize}
		\item[(i)] $\omega_{n}(\ff_1,\dots,\ff_{n})=0$, if $n\in \mathbb{N}$ is odd,
		\item[(ii)]  $\omega_{2n}(\ff_1,\dots,\ff_{2n})=\sum_{\Pi}  \omega_2(f_{i_1},f_{i_2}) \cdots  \omega_2(f_{i_{n-1}},f_{i_n})$,
		if $n\in \mathbb{N}$ is even,
	\end{itemize}
	where $\Pi$  refers to the class of all possible decompositions of the set $\{1,2,\ldots, 2n\}$ into $n$ pairwise disjoint subsets of $2$ elements $\{i_1,i_2\}$, $\{i_3,i_4\}$, $\ldots$, $\{i_n-1,i_n\}$ with $i_{2k-1} < i_{2k}$ for 
	$k=1,2,\ldots, n$.
\end{definition}

Regarding the notion of Hadamard state for the Proca field, which is a vector field, we adopt the notions of microlocal analysis for vector-valued distributions introduced in \cite{SV}.

\remark\label{remdist}  The interpretation of the action of a distribution on test sections  is  formalized  in the sense of the Hodge product (\ref{hodge}). This interpretation is necessary in order to agree with the interpretation of the symbol $\hat{\fa}(\ff)$ stated  in Remark \ref{reminterp}, since some of the distributions we shall consider in the rest of the paper arise from field operators, as the two-point functions $\omega_2(\ff,\fg) := \omega(\hat{\fa}(\ff)\hat{\fa}(\fg))$.  If 
$$\Gamma_c(\V_g)\ni \fg \mapsto  \omega_2(\cdot ,\fg)  \in  \Gamma'_c(\V_g)$$
is well-defined and  continuous, $\omega_2$  actually defines a distribution of $\Gamma_c'(\V_g\boxtimes \V_g)$ and {\em vice versa}, as a consequence of the {\em Schwartz kernel theorem} as clarified below.\\
From now on, if $F \in \Gamma'_c(\V_g)$ and $\ff \in \Gamma_c(\V_g)$, the action of the former on the latter is therefore interpreted as the Hodge product (\ref{hodge}) $$F(\ff)= ( F| \ff)  =  ( \ff|F)   = \int_\M \g^\sharp(F, \ff) \vol_g\:.$$
With a straightforward extension of the Definition \ref{defadjoint},  operators working on a generic space of $k$ test-forms  $\T: \Omega^k_c(\M)\to \Omega^k_c(\M)$ can be extended to the topological duals, i.e the associated distributions, in terms of the action  $\T^\dagger$  on the argument of the distribution:
$$(\T F)(\ff) := F(\T^\dagger \ff)\:.$$
For instance, if $F \in {\Omega^2}'_c(\M)$ and $H \in {\Omega^0}'_c(\M)$,
$$(\delta F)(\ff) := F(d\ff)\:, \quad (dH)(\ff) := H(\delta \ff)\:, \quad \ff \in \Omega_c^1(\M)\:.$$
If $\S : \Gamma_c(\V_g) \to \Gamma'_c(\V_g)$ is continuous (in particular if $\S : \Gamma_c(\V_g) \to \Gamma_c(\V_g)$ is continuous), the standard Schwartz kernel theorem permits to introduce the distribution indicated with the same symbol $\S \in \Gamma_c'(\V_g\boxtimes \V_g)$, which is the unique distribution such that
$$\S(\ff \otimes \fg) := \S(\ff,\fg) := (\S\fg)(\ff)\:\: ``= (  \ff| \S \fg )''\:.$$
Conversely, a distribution of $\Gamma_c'(\V_g\boxtimes \V_g)$ defines a unique map $ \Gamma_c(\V_g) \to \Gamma'_c(\V_g)$ that fulfills the identity above.
In the rest of the work we shall take advantage of these  facts and notations above. Furthermore,  we adopt the notion of {\em  wavefront set}  of a distribution on test sections of  a vector bundle on $\M$ as defined in \cite{SV}.\\

\begin{definition}\label{HadamardMMV} Consider the globally hyperbolic spacetime $(\M,g)$ and a  state $\omega:\mathcal{A}_g\to\CC$ 
for the Proca algebra of observables on $(\M,g)$. $\omega$ is  called
	  \textbf{Hadamard} if it is quasifree and its two-point function $\omega_2\in \Gamma_c'(\V_g\boxtimes \V_g)$ satisfies the {\bf microlocal spectrum condition}\footnote{The notion of wavefront set  refers to distributions acting on {\em complex} valued test sections in view of the pervasive use of the Fourier transform. For this reason, when dealing with these notions we consider the natural complex extension of the involved distributions, by imposing that they are also  $\mathbb{C}$-linear.}, i.e.
	\begin{equation} \label{WHadMMV} WF(\omega_2)={\cal H}:= \{(x,k_x;y,-k_y)\in T^*\M^2\backslash\{0\}\:|\:(x,k_x)\sim_{\parallel}(y,k_y), k_x\triangleright0\}\:.\end{equation}
Above,  $(x,k_x)\sim_{\parallel}(y,k_y)$ means that $x$ and $y$ are connected by a lightlike geodesic and $k_y$ is
the co-parallel transport of $k_x$ from $x$ to $y$ along said geodesic, whereas $k_x \triangleright 0$ means that the
covector $k_x$ is future pointing.
\end{definition}

As for Klein-Gordon scalar field theory, Hadamard states for Proca fields have two important properties which were also established in \cite{Fewster} for the notion of Hadamard state adopted there. We present here independent proofs only based on Definition~\ref{HadamardMMV}. Indeed, \cite{Fewster} uses a definition of Hadamard states which is apparently different from our definition. A comparison of the two definitions  and an equivalence result appear in Section \ref{SECHADFP}.\\
The first property of Hadamard states is the fact that the difference between the two-point functions of a pair of Hadamard states is a smooth function. This fact plays a crucial role in the point-splitting renormalization procedure (for instance of Wick polynomials and time-ordered polynomials \cite{HW1,HW2, KM, KMM} and of the stress-energy tensor \cite{W,stress,HM}) and is, in fact, one of the reasons for assuming that Hadamard states are the physically relevant ones.

\begin{proposition}\label{smoothnessMMV}
Let $\omega, \omega' \in \Gamma'_c(V_g\boxtimes V_g)$ be a pair of Hadamard states  on the algebra $\mathcal{A}_g$ of  the Proca field according to Definition~\ref{HadamardMMV}. Then,   $\omega-\omega' \in \Gamma(V_g\boxtimes V_g)$, i.e., $\omega-\omega'$ is smooth.\\
More generally, $\omega-\omega'$ is smooth  if  $\omega,\omega'$ are distributions satisfying (\ref{WHadMMV}) such that  their  antisymmetric parts coincide mod. $C^\infty$.
\end{proposition}

\begin{proof} Let us first prove the second statement.
Let us  define $\omega_2^+(\ff,\fg):= \omega_2(\ff,\fg)$ and $\omega_2^-(\ff,\fg):= \omega_2(\fg,\ff)$,
$$N^{+}:= \{(x,k)\in T^*\M\backslash\{0\}\:|\: k_ak^a =0\:, \: k\triangleright0\}\:, \quad N^{-}:= \{(x,k)\in T^*\M\backslash\{0\}\:|\: k_ak^a =0\:,\: k\triangleleft0\}\:,$$
\begin{equation}\Gamma' := \{(x,k_x;y,-k_y)\in T^*\M^2\backslash\{0\}\:|\:(x,k_x;y, k_y)\in \Gamma\}\:.\label{Gammaprime}\end{equation}
for every $\Gamma \subset \T^*\M^2\setminus \{0\}$.
If both distributions satisfy (\ref{WHadMMV}), then 
\begin{equation}\label{WHadPpm} WF(\omega_2^\pm)'  \subset   N^\pm \times N^\pm\:.\end{equation}
With the hypotheses of the proposition  define
$R^\pm := \omega_2^\pm - \omega'^\pm_2$. Since $\omega_2^+ - \omega^-_2= \omega'^+_2 - \omega'^-_2 + F$ where $F$ is a smooth function, we have that 
$R^+=-R^-$ mod. $C^\infty$. At this juncture, (\ref{WHadPpm}) yields $WF(R^+)'  \cap WF(R^-)' =\emptyset$ because $N^+\cap N^-=\emptyset$.
Since $WF(R^+)=WF(-R^- +F) = WF(-R^-)= WF(R^-)$, we conclude that the wavefront set of the distributions  $R^\pm$ is empty and thus they are smooth functions. This is the thesis of the second statement. The latter  statement implies the former because, since both $\omega$ and $\omega'$ are states on the Proca $*$-algebra, their antisymmetric part must be identical and it amounts to $i\G_P$, furthermore $\omega$ and $\omega'$ satisfy (\ref{WHadMMV}) in view of Definition~\ref{HadamardMMV}.
\end{proof}

The second property regards the so called propagation property of the  Hadamard singularity or also the local-global feature of Hadamard states. It has a long history which can be traced back to \cite{FSW} passing through  \cite{KW}, \cite{Rad1,Rad2} and \cite{SV} (and the recent \cite{Valter}) at least. 

\begin{proposition}\label{propagationMMV}
Consider a globally hyperbolic spacetime $(\M,g)$ and a globally hyperbolic neighborhood $\mathcal{N}$  of a smooth spacelike Cauchy surface $\Sigma$ of $(\M,g)$.
Finally, let $\omega_{\mathcal{N}}$ be a quasifree state for the on-shell algebra of the Proca field in $(\mathcal{N}, g|_{\mathcal{N}})$. The following facts are valid.
\begin{itemize}
\item[(a)] There exists a unique a quasifree  state $\omega : \mathcal{A}_g \to \mathbb{C}$ for the Proca field on the whole $(\M,g)$ which restricts to  $\omega_{\mathcal{N}}$ on 
the Proca algebra on $\mathcal{N}$.
\item[(b)] If $\omega_{\mathcal{N}}$ is Hadamard according to Definition \ref{HadamardMMV}, then $\omega$
is.
\end{itemize}
\end{proposition}

\begin{proof} (a) 
According to (\ref{cruciallemma0}),
$\G_\P \ff =0 \: \mbox{for $\ff \in \Gamma_c(\V_g)$}$ if and only if $\ff = \P\fg \: \mbox{for some $\fg \in \Gamma_c(\V_g)$}$.
 We will use this fact to construct $\omega$ out of $\omega_{\mathcal{N}}$. 
Consider two other  smooth spacelike surfaces (for both $\M$ and $\mathcal{N}$) $\Sigma_+$ in the future of $\Sigma$ and $\Sigma_-$ in the past of $\Sigma$. Let $\chi^+,\chi^- :\M \to [0,1]$ be smooth maps such that $\chi^+(p)=0$ if $p$ stays in the past of $\Sigma_-$ and $\chi^+(p)=1$ if $p$ stays in the future of $\Sigma_+$ and $\chi^++\chi^-=1$.
 Then, defining \begin{equation}\T \ff := \P \chi^+ \G_\P \ff\:, \quad \ff \in \Gamma_c(\V_g)\label{PREdefT}\end{equation} we have that $\T \ff \in \Gamma_c(\V_g|_\mathcal{N})$  (more precisely $supp(\T\ff)$ stays between $\Sigma_-$ and $\Sigma_+$), and
\begin{equation}
 \T \ff - \ff = \P \fg \quad \mbox{for some $\fg \in \Gamma_c(\V_g)$}\:,\label{defT}
 \end{equation} 
because by standard properties of Green operators:
$$
\begin{gathered}
 \G_\P \T\ff = \G_\P^+ \T\ff - \G_\P^- \T\ff = 
 \left(\G_\P^+\P\right) \chi^+ \G_\P \ff - \G_\P^-\P (1-\chi^-) \G_\P \ff =\\
  \chi^+ \G_\P \ff - \G_\P^-\left(\P  \G_\P \ff\right)+\G_\P^-\P\chi^-\G_\P\ff=
  \chi^+ \G_\P \ff +\chi^-\G_\P\ff =\G_\P \ff.
 \end{gathered}
 $$
Therefore we can apply (\ref{cruciallemma0}) obtaining (\ref{defT}).\\
\noindent With these results, let us define 
\begin{equation} 
\omega_2(\ff, \fg) := \omega_{\mathcal{N}2}(\T\ff,\T\fg) \:,\quad \ff,\fg \in \Gamma_c(\V_g)\:.\label{omegaT}
\end{equation}
Taking  the continuity properties of $\G_\P$ into account,  we leave to the reader the elementary proof of the fact that there is a unique distribution $\Gamma_c'(\V_g\boxtimes \V_g)$ such that its value on $\ff \otimes  \fg$ coincides with\footnote{If $\omega_2$ indicates the distribution associated to the two-point function: $\omega_2 = \omega_{\mathcal{N}2} \circ \T\otimes \T$.} $\omega_2(\ff, \fg)$. 
(We will indicate that distribution by $\omega_2$ with the usual misuse of language.)
Furthermore, in view of the definition of $\T$, it is obvious that $\omega_2$ is also a bisolution of the Proca equation, since $\G_\P \P = \P\G_\P =0$. Using   Definition \ref{quasifree} to construct a candidate quasifree state 
$\omega$ on $\mathcal{A}_g$ out of its two-point function $\omega_2$, it is clear that the positivity requirement is guaranteed because $\omega_{\mathcal{N}}$ satisfies it.
We conclude that there is a quasifree state 
$\omega$ on  $\mathcal{A}_g$, whose two point function is (\ref{omegaT}), and this two point function is a distribution which is also bisolution of the Proca equation. Finally, observe that $\omega$ extends to the whole $\mathcal{A}_g$ the state  $\omega_{\mathcal{N}}$ since  the states are  quasifree and the  two-point function of the former  extends the two point function of the latter. Indeed, 
$$ \omega_2(\ff, \fg) = \omega_{\mathcal{N}2}(\T\ff,\T\fg) = \omega_{\mathcal{N}2}(\ff,\fg) \quad \mbox{if $\ff,\fg \in \Gamma_c(\V_g|_{\mathcal{N}})$}\:.$$
This is because, specializing (\ref{cruciallemma0}) and (\ref{PREdefT})-(\ref{defT}) to the globally hyperbolic spacetime $(\mathcal{N}, g|_{\mathcal{N}})$ since $\ff \in \Gamma_c(\V_g|_{\mathcal{N}})$,
we have that $\T\ff -\ff = \P\fg$ with $\fg \in  \Gamma_c(\V_g|_{\mathcal{N}})$ and $ \omega_{\mathcal{N}2}$
vanishes when one argument has the form $\P\fg$, because it is a bisolution of the Proca equation in $\mathcal{N}$. \\
  There is only one such quasifree state which is an extension of $\omega_{\mathcal{N}}$
 to the whole algebra $\mathcal{A}_g$,
and such that its two-point function is a bisolution of the Proca equation.
In fact, another such extension  $\omega'$  would satisfy
$$\omega_2'(\ff,\fg) = \omega'_2(\T\ff, \T\fg) = \omega_{\mathcal{N}}(\T\ff,\T\fg) =  \omega_2(\T\ff, \T\fg)= \omega_2(\ff,\fg)\:, \quad \mbox{for all  $\ff, \fg \in \Gamma_c(\V_g)$}. $$
(b) We pass to the proof that $\omega$ is Hadamard if $\omega_{\mathcal{N}}$ is. We have to prove that (\ref{WHadMMV}) is valid if it is valid for $\omega_{\mathcal{N}}$ in $(\mathcal{N}, g|_{\mathcal{N}})$.
Interpreting the two-point functions as distributions of $\Gamma_c'(\V_g\boxtimes \V_g)$,
\begin{equation}\omega_2 = \omega_{\mathcal{N}2} \circ \P \chi^+ \G_\P \otimes  \P \chi^+ \G_\P\:.\label{compG}\end{equation}
The wavefront sets of $\G_\P$ and $ \P \chi^+ \G_\P$ can be computed as follows. First of all, from (\ref{GPGN}),
\begin{equation}\label{PQN}\G_\P = \Q \G_\N = \G_\N\Q	\end{equation}
where $\Q= I + m^{-2} d\delta_g $.  It is known that 
$$WF(\G_\N) =  \{(x,k_x;y,-k_y)\in T^*\M^2\backslash\{0\}\:|\:(x,k_x)\sim_{\parallel}(y,k_y) \}$$
Notice that, in particular $k_x \neq 0$ and  $k_y \neq 0$ nor simultaneously by definition, nor separately since they are connected by a coparallel transport.\\
So, since $\Q$ is a differential operator we immediatly deduce by \ref{PQN} that $WF(\G_\P)\subset WF(\G_N)$. Then we associate to the two operator their distributional kernels $\G_\P(x,y)$ and $\G_\N(x,y)$ and recast equation \ref{PQN} in the form:
$$\G_\P(x,y)=\left(\Id_x\otimes \Q_y\right)\G_\N(x,y)
,$$
which, by standard microlocal analysis results, implies that
$$WF(\G_\N)\subset Char(\Id_x\otimes \Q_y)\cup WF(\G_\P).
$$
However explicit computations give that $Char(\Id_x\otimes\Q_y)=\{(x,k_x;y,0)|(x,k_x)\in\T^*\M,y\in\M\}$ which does not intersect $WF(\G_\N)$ at any point, implying $$WF(\G_\N)\subset WF(\G_\P)\subset WF(\G_\N).$$ So $\G_P$ and $\G_\Q$ have the same wavefront set. Therefore, since $\P \chi^+$ is a smooth differential operator,
$$ WF(\P \chi \G_\N)  \subset  \{(x,k_x;y,-k_y)\in T^*\M^2\backslash\{0\}\:|\:(x,k_x)\sim_{\parallel}(y,k_y) \}$$
A this point, a standard estimate of composition of wavefront sets in (\ref{compG}) yields (see, e.g.,  \cite{IgorValter})
$$WF(\omega_2) \subset {\cal H}\,$$
where the Hadamard wavefront set ${\cal H}$ is the one in (\ref{WHadMMV}).
To conclude the proof, it is sufficient to establish the converse inclusion. To this end, 
observe that, since the antisymmetric part of  $\omega_2$ is $\omega_2^+-\omega_2^-= i\G_\P$, 
$$WF(\G_\P)\subset  WF(\omega^+_2) \cup WF(\omega_2^-)\:,$$
where we adopted the same notation as at the beginning of  the proof of Proposition~\ref{smoothnessMMV}: $\omega_2^+=\omega_2$, $\omega_2^-(\ff,\fg)= \omega_2(\fg,\ff)$. If, according to that notation, 
the prime applied to wavefront sets is defined as in (\ref{Gammaprime}), the above inclusion can be re-phrased to
\begin{equation} \{(x,k_x;y,k_y)\in T^*\M^2\backslash\{0\}\:|\:(x,k_x)\sim_{\parallel}(y,k_y) \} = WF(\G_\P)' \subset  WF(\omega^+_2)' \cup WF(\omega_2^-)' \label{tripinc}\end{equation}
Above 
$$ WF(\omega^+_2)' \subset {\cal H}' = \{(x,k_x;y,k_y)\in T^*\M^2\backslash\{0\}\:|\:(x,k_x)\sim_{\parallel}(y,k_y), k_x\triangleright0\}$$
and, with a trivial computation,
$$ WF(\omega^-_2)' \subset  \{(x,-k_x;y,-k_y)\in T^*\M^2\backslash\{0\}\:|\:(x,k_x)\sim_{\parallel}(y,k_y), k_y\triangleright0\}\:,$$
Now suppose that $(x,k_x;y,k_y)\in {\cal H}'$ does not belong to  $WF(\omega^+_2)'$. According to (\ref{tripinc}),   $ (x,k_x;y,k_y)\not \in  WF(\G_\P)'$ (notice that  ${\cal H}' \ni (x,k_x;y,k_y) \not \in  WF(\omega_2^-)'$ since the two sets are disjoint). This is impossible because every $(x,k_x;y,k_y)\in {\cal H}'$
belongs to $ WF(\G_\P)'$ as it immediately arises by comparing the explicit expressions of these two sets written above. In summary ${\cal H}' \subset WF(\omega_2)'$, that is ${\cal H}\subset WF(\omega_2)$, concluding the proof.
\end{proof}

\noindent Hadamard  states turned also out to be relevant in the study of quantum energy conditions \cite{Fewster, SF, FV} and in black hole physics  \cite{DMPu,Ko,MP,KPV,gerardUnruh} (see references in \cite{Valter} for a summary)

We are finally ready to extend the M{\o}ller operator to the quantum algebras, proving that they are indeed isomorphic. To this end, for any paracausally related metric $g\simeq g'$, 
	we define an isomorphism of the free algebras $\mathcal{R}_{gg'}:\mathfrak{A}_{g'}\to\mathfrak{A}_{g}$ as the unique unital $*$-algebra isomorphism between the said free unital $*$-algebras such that 
	$$\mathcal{R}_{gg'}(\fa'(\ff))=\fa(\R^{\dagger_{g g'}}\ff)\quad\forall\ff\in\Gamma_c(\V_{g'})\,,$$
	where  $\R$ is a M\o ller operator of $g,g'$ and the adjoint $\R^{\dagger_{g g'}}$ is defined as in Proposition~\ref{MollerAdj}.

\subsection{M{\o}ller $*$-isomorphism and the pullback of Hadamard states}
When we pass to the quotient algebras, the preservation of the causal propagators discussed in the previous sections, immediately implies that the induced map on the quotient algebras is an isomorphism, that we call 
\textbf{M{\o}ller $*$-isomorphism.}

\begin{proposition}
Let now $\mathcal{R}_{gg'}:\mathcal{A}_{g'}=\mathfrak{A}_{g'}/\mathfrak{I}_{g'}\to\mathcal{A}_g=\mathfrak{A}_g/\mathfrak{I}_g$ be the quotient morphism constructed out of $\mathcal{R}_{gg'}$. Then $\mathcal{R}$ is well defined and is indeed a $*$-algebra isomorphism.  
\end{proposition}
\begin{proof}
The proof of this statement is identical to the one performed in \cite[Proposition 5.6]{Norm}. Indeed it just relies on the preservation of the causal propagators proved in Theorem~\ref{causalpropR}, which implies that the associated $CCR$-ideals are $*$-isomorphic.
\end{proof}

The final step in our construction is to define a pullback of the M{\o}ller $*$-isomorphism to the states and then to prove that the Hadamard condition is preserved, as done in \cite[Theorem 5.14]{Norm} for normally hyperbolic field theories.

\begin{theorem}\label{thm:main intro Had}
Let $\mathcal{R}_{gg'}$ be the M{\o}ller *-isomorphism and let $\omega:\mathcal{A}_g\to\CC$ be a quasifree Hadamard state, we define the pull-back state $\omega':\mathcal{A}_{g'}\to\CC$ by
$\omega'=\omega\circ\mathcal{R}_{gg'}$. 
The following facts are true:
\begin{enumerate}
\item[1] $\omega'$ is a well-defined state;
\item[2] $\omega'$ is quasifree;
\item[3] $\omega'$ is a Hadamard state.
\end{enumerate}
\end{theorem}
\begin{proof}
The proof of 1-2 is trivial and discussed in \cite[Proposition 5.11]{Norm}. The proof of 3 follows from the Hadamard  propagation property stated  in Proposition \ref{propagationMMV}. To prove the statement we can just focus on the case in which the M{\o}ller operator is constructed out of two spacetimes such that $g\preceq g'$, the reasoning can then be iterated at each step of the paracausal chain.\\
The two-point function of the pullback state can be written as
$$\omega'_2(\ff,\fh)=\omega'(\hat{\mathfrak{a}}'(\ff)\hat{\mathfrak{a}}'(\fh))=\omega(\mathcal{R}_{gg'}(\hat{\fa}'(\ff)\hat{\fa}'(\fh)))=\omega(\hat{\fa}(\R^{\dagger_{g g'}}\ff)\hat{\fa}(\R^{\dagger_{g g'}}\fh))=\omega_2(\R^{\dagger_{g g'}}\ff,\R^{\dagger_{g g'}}\fh). $$
We recall that the operator is the composition of two pieces $\R^{\dagger_{g g'}}=\R_+^{\dagger_{g g_\chi}}\circ\R_-^{\dagger_{g_\chi g'}}$ and split the proof in two steps.\\
First we focus on the bidistribution $\omega_2^{\chi}(\ff,\fh):=\omega_2(\R_+^{\dagger_{g g_\chi}}\ff,\R_+^{\dagger_{g g_\chi}}\fh)$ on $(\M,g_\chi)$ defining a quasifree state therein.
 By the property \ref{pastequal}, in the region in which $g_\chi=g$, there is a $t_0$ a Cauchy surface $\Sigma_{t_0}$ in common for $g$ and $g_\chi$, a common  globally hyperbolic neighborhood  $\mathcal{N}$ of that Cauchy surface such that $\omega_2^{\chi}(\ff,\fh)
= \omega_2(\ff,\fh)$ when the supports of $\ff$ and $\fg$ are chosen in $\mathcal{N}$ and thus the corresponding state is Hadamard in $(\mathcal{N}, g_\chi)$. Now  Proposition \ref{propagation} implies that 
$\omega_2^{\chi}$ is Hadamard in the whole $(\M, g_\chi)$.
Secondly, the same argument can be used once again for the operator $\R_-^{\dagger_{g_\chi g'}}$ on the Hadamard state $\omega^{\chi}$ on $(\M, g_\chi)$,
proving that the state induced by $\omega_2(\R^{\dagger_{g g'}}\cdot ,\R^{\dagger_{g g'}}\cdot) = \omega_2^{\chi}(\R_-^{\dagger_{ g_\chi g'}}\cdot ,\R_-^{\dagger_{g_\chi g'}}\cdot)$ is Hadamard as well on $(\M, g')$.
 In other words  the full M{\o}ller operator preserves  the Hadamard form.
\end{proof}

\section{Existence of Proca Hadamard states in globally hyperbolic spacetimes}\label{sec:existence} This section is devoted to the construction of Hadamard states for the real Proca field in a generic globally hyperbolic spacetime.
Actually, the technology of M{\o}ller operators, in particular Theorem~\ref{thm:main intro Had}, allows us to reduce the construction of Hadamard
states for the Proca equation to the special case of an ultrastatic spacetime with Cauchy hypersufaces of bounded geometry.
Indeed, as shown in\cite[Corollary 2.23]{Norm}, for any globally hyperbolic spacetime  $(\M,g)$, there exists a paracausally related  globally hyperbolic  spacetime $(\M,g_0)$ which is ultrastatic. In other words, first of all $(\M, g_0)$ is isometric to $\RR \times \Sigma$ where $(\Sigma, h_0)$ is a $t$-independent complete Riemannian manifold and  $g_0 = - dt\otimes dt + h_0$, where  $t$  is the natural coordinate on $\RR$ and $dt$ is past directed.  We  also  denote by $\partial_t$  the tangent vector to the submanifold $\RR$ of $\RR \times \Sigma$.
In view of the completeness of $h$, these spacetimes are  globally hyperbolic  (see e.g. \cite{Fulling}) and  $\Sigma$ is a Cauchy surface of this spacetime. 
In turn, it is possible to change the metric on $\Sigma$ in order that the final metric, indicated by $h$ is both complete and {\em of bounded geometry} \cite{greene}.  By construction, the final ultrastatic spacetime $(\M,  - dt\otimes dt + h)$ is paracausally related to  $(\M,g_0)$ because the intersection of the corresponding open cones  is non-empty as it always contains $\partial_t$.
 By transitivity  $(\M,g)$ is paracausally related with  $(\RR \times \Sigma, - dt\otimes dt + h)$.

Hence, we assume without loss of generalities, that $(\M, g)=(\RR\times \Sigma,  - dt\otimes dt + h)$ is a globally hyperbolic ultrastatic spacetime, with $dt$ past directed,  whose spatial metric $h$ is complete. When dealing with the construction of Hadamard states  we also assume that the spatial manifold $(\Sigma, h)$ is also  of bounded geometry. In the final part of the section, we will come back  to consider  a generic globally hyperbolic spacetime $(\M,g)$

\subsection{The Cauchy problem in ultrastatic spacetimes}\label{SECCP} We study here the Cauchy problem for the Proca (real and complex) field in ultrastatic spacetimes  $(\M,g) = (\RR \times \Sigma,  - dt\otimes dt + h)$, where $(\Sigma,h)$ is complete. A  more general treatise appears in \cite{koProca} where the Cauchy problem is studied, also in the presence of a source of the Proca field, in a generic globally hyperbolic spacetime and the continuity of the solutions with respect to the initial data is focused.

 Let us consider the Proca equation (\ref{ProcaA}) (where $m^2 >0$)  on the above ultrastatic  spacetime.
As observed in \cite{Fewster}, every  smooth  $1$-form $A \in \Omega^1(\M)$ naturally uniquely decomposes as \begin{equation} A(t,p) = A^{(0)}(t,p) dt + {A}^{(1)}(t,p) \quad 
\label{DECAA}
\end{equation}
where ${A}^{(i)}(t, \cdot) \in \Omega^i(\Sigma)$  for $i=0,1$ and  $t\in \RR$.   By direct inspection and taking  the equivalence of (\ref{ProcaA}) and (\ref{PROCAEQ1})-(\ref{PROCAEQ2}) into account, one sees that Eq.
\eqref{ProcaA} is equivalent to the constrained double Klein-Gordon  system
\begin{eqnarray}
\partial_t^2  A^{(0)} &=& -(\Delta_h^{(0)} + m^2)  A^{(0)}\:, \label{EQ1}\\
\partial_t^2 A^{(1)} &=& -(\Delta_h^{(1)} + m^2)  A^{(1)}\:,  \label{EQ2}\\
\partial_t  A^{(0)} &=& -\delta^{(1)}_h  A^{(1)}\:.  \label{EQ3}
\end{eqnarray}
Above,  $\Delta_h^{(k)} := \delta_h^{(k+1)}d^{(k)} +  d^{(k-1)}\delta_h^{(k)}$ is the Hodge Laplacian on $(\Sigma,h)$ for $k$-forms  and the last condition (\ref{EQ3}) is nothing but the constraint $\delta^{(1)}_g A=0$ arising from (\ref{ProcaA}). 

The theory for the fields $ A^{(1)}$ and $ A^{(0)}$ is a special case  of the theory of {\em normally hyperbolic equations on corresponding  vector bundles with positive inner product} over a globally hyperbolic spacetime \cite{Ba,BaGi}.  In our case,
\begin{itemize}
\item[(1)] there is a real vector bundle  $\V_g^{(1)}$with  basis $\M$, canonical fiber isomorphic to $T^*_q\Sigma$, and equipped with a fiberwise real symmetric scalar product 
induced by  $h_q^\sharp$.  $A^{(1)} \in \Gamma(V_g^{(1)})$;
\item[(2)]  there is another  real vector bundle  $\V_g^{(0)}$ with basis $\M$, canonical fiber isomorphic to $\RR$, and equipped with a positive fiberwise real symmetric scalar product  given by the natural product in $\RR$. $A^{(0)} \in \Gamma (V_g^{(0)})$.
\end{itemize}
Evidently 
\begin{equation}
\V_g = \V_g^{(0)} \oplus \V_g^{(1)}\label{decVV}\:.
\end{equation}
 Equations  (\ref{EQ1}) and (\ref{EQ2}) admit  existence and uniqueness theorems for smooth  compactly supported Cauchy data and corresponding  smooth  spacelike compact solutions
in $\Gamma_{sc}(\V_g^{(0)})$  and  $\Gamma_{sc}(\V_g^{(1)})$   respectively, as a consequence of very well-known results in the theory of normally hyperbolic equations  \cite{Ba,BaGi,GiMu}. 
 However, when viewing $A^{(0)}$ and $A^{(1)}$ as parts of the Proca field $A$,  we have also to deal with the additional constraint~\eqref{EQ3}. 
 Notice that (\ref{EQ3}) imposes two constraints on the Cauchy data of $ A^{(0)}$ and $ A^{(1)}$ on $\Sigma$: 
 $$\partial_t  A^{(0)}(0,p) = -\delta^{(1)}_h  A^{(1)}(0,p) \qquad \partial^2_t  A^{(0)}(0,p) =- \delta^{(1)}_h \partial_t A^{(1)}(0,p)\,.$$ The second constraint is only apparently of the second order. Indeed, taking (\ref{EQ1}) into account, it can be re-written as a condition of the Cauchy data 
 $$ (\Delta_h^{(0)} + m^2)  A^{(0)}(0,p) = \delta^{(1)}_h \partial_t A^{(1)}(0,p)\,.$$  At this juncture  we observe that, with some elementary computation (use $\Delta^{(0)}_h \delta_h^{(1)} = \delta_h^{(1)} \Delta_h^{(1)}$),  Equations~(\ref{EQ1}) and~(\ref{EQ2})
imply also the crucial condition  $$(\partial_t^2 + \Delta^{(0)}_h - m^2) (\partial_t  A^{(0)} + \delta^{(1)}_h  A^{(1)})=0$$ which, in turn, implies   Equation~(\ref{EQ3}) , 
if the initial condition of that scalar Klein-Gordon equation for $(\partial_t  A^{(0)}+ \delta^{(1)}_h  A^{(1)})$ are the zero initial conditions. This exactly amounts to  
$$\partial_t  A^{(0)}(0,p) = -\delta^{(1)}_h  A^{(1)}(0,p) \qquad\text{ and }\qquad (\Delta_h^{(0)} + m^2)  A^{(0)}= \delta^{(1)}_h \partial_t A^{(1)}(0,p)\,.$$
In summary, we are naturally led to focus on this Cauchy problem
\begin{eqnarray}
&&\partial_t^2 A^{(0)} + (\Delta_h^{(0)} + m^2)  A^{(0)}=0\:, \label{EQ11}\\
&&\partial_t^2  A^{(1)} +(\Delta_h^{(1)} + m^2)  A^{(1)}=0\:,  \label{EQ22}\\
&&(\partial_t^2 + \Delta^{(0)}_h - m^2) (\partial_t  A^{(0)} + \delta^{(1)}_h  A^{(1)})=0\:, \quad   
 \label{EQ33}
\end{eqnarray}
with initial data 
\begin{equation}
\mbox{$ A^{(0)}(0,\cdot) = a^{(0)}(\cdot)$, \quad $\partial_t  A^{(0)}(0,\cdot)= \pi^{(0)}(\cdot)$, \quad $ A^{(1)}(0,\cdot)= {a}^{(1)}(\cdot)$, \quad  $\partial_t A^{(1)}(0,\cdot)= {\pi}^{(1)}(\cdot)$} \label{Cdata}
\end{equation}
where $a^{(0)},\pi^{(0)},  a^{(1)}, \pi^{(1)}$ are pairs of  smooth compactly supported, respectively $0$ and  $1$,  forms on $\Sigma$,
and the constraints are valid
\begin{equation}
\pi^{(0)}  =- \delta^{(1)}_h a^{(1)} \:, \quad   (\Delta_h^{(0)} + m^2) a^{(0)} = \delta^{(1)}_h \pi^{(1)}\:. \label{CCdata}
\end{equation}
If $A$ is a spacelike compact solution of the Proca equation (\ref{ProcaA}), then it satisfies (\ref{EQ1})-(\ref{EQ3})  and its Cauchy data (\ref{Cdata})  satisfy the constraints (\ref{CCdata}).
On the other hand, if we have smooth compactly supported  Cauchy data  (\ref{Cdata}),  then the two Klein-Gordon equations (\ref{EQ1}) and (\ref{EQ2}) admit  unique spacelike compact smooth solutions which 
also satisfies (\ref{EQ33}) as a consequence. If the said Cauchy data satisfy the constraint  (\ref{CCdata}), then also  
(\ref{EQ3}) is satisfied, because it is equivalent to the unique solution of (\ref{EQ33}) with zero Cauchy data. In that case,  the two solutions $ A^{(0)}$ and $ A^{(1)}$ define a unique solution of the Proca equation with the said Cauchy data.

We have established the following result completely extracted from the theory of normally hyperbolic equations.
 
\begin{proposition}\label{TEOEUP} Let $(\M, g) = (\Sigma,  - dt\otimes dt + h)$ be a smooth  globally hyperbolic  ultrastatic spacetime with $dt$ past directed, where $h$ is a smooth complete Riemannian metric on $\Sigma$.
Consider the Cauchy problem on $(\M, g)$ for the 
smooth $1$-form $A$ satisfying the 
Proca equation (\ref{ProcaA}) for $m^2>0$, with smooth compactly supported  Cauchy data (\ref{Cdata}) on $\Sigma$ viewed as the $t=0$ time slice. \\ 
The Proca Cauchy problem for $A$  with constraints (\ref{CCdata})  is equivalent, regarding  existence and uniqueness of spacelike compact smooth solutions ,   to the double normally hyperbolic  Klein-Gordon constrained  Cauchy problem (\ref{EQ1})-(\ref{EQ3}),  for the fields $ A^{(0)} \in \Gamma_{sc}(\V_g^{(0)})$ and $ A^{(1)} \in  \Gamma_{sc}(\V_g^{(1)})$,  with the same initial data (\ref{Cdata}) and constraints  (\ref{CCdata}). 
As a consequence,
\begin{itemize}
\item[(1)]  every smooth spacelike compact  solution of the Proca equation  $A\in \Gamma_{sc}(\V_g)$  (\ref{ProcaA}) defines compactly supported smooth Cauchy data on $\Sigma$ which satisfy  the constraints  (\ref{CCdata});
\item[(2)] if the Cauchy data are smooth, compactly supported and satisfy  (\ref{CCdata}),  then there is a unique smooth spacelike compact solution of the Proca equation $A\in \Gamma_{sc}(\V_g)$  (\ref{ProcaA}) associated to them;
\item[(3)] the support of a solution $A\in \Gamma_{sc}(\V_g)$ with smooth compactly supported initial data  satisfies $supp(A) \subset J^+(S) \cup J^-(S)$, where $S \subset \Sigma$ is the union of the supports of the Cauchy data.\\
\end{itemize}
\end{proposition}

\remark
\begin{itemize}
\item[(1)] All the discussion above, and Proposition~\ref{TEOEUP} in particular, extends to the case of a {\em complex} Proca field and corresponding associated complex  Klein Gordon fields. The stated results can be extended easily to the case of the non-homogeneous Proca equation and also considering continuity properties of the solutions with respect to the source and the initial data referring to natural topologies. (See \cite{koProca} for a general discussion.)
\item[(2)]  A naive idea may be  that we can  freely fix smooth compactly supported Cauchy data for $ A^{(1)}$ and then define associated Cauchy conditions for $ A^{(0)}$ by solving the constraints~(\ref{CCdata}). 
In this case the true degrees of freedom of the Proca field would be the vector part $ A^{(1)}$, whereas $ A^{(0)}$ would be a constrained degree of freedom. This viewpoint is incorrect, if we decide to deal with spacelike compact solutions, because the second constraint in  Equation~(\ref{CCdata}) in general does not produce a compactly supported function $a^{(0)} $ when the source $\delta^{(1)}_h \pi^{(1)}$ is smooth compactly supported (the smoothness of $a^{(0)} $ is however guaranteed by elliptic regularity from the smoothness of $\delta^{(1)}_h \pi^{(1)}$). $a^{(0)}$ is  compactly supported  only for some smooth compactly supported  initial conditions $ \pi^{(1)}$. Therefore the linear subspace of  initial data (\ref{Cdata})  compatible with the constraints   (\ref{CCdata}) does not include {\em all}  possible compactly supported  initial conditions $ \pi^{(1)}$ which, therefore, cannot be freely chosen.
\item[(3)]
 However this space of constrained Cauchy data is non-trivial, i.e., it does not contain only zero initial conditions and in particular there are couples $(a^{(0)}, \pi^{(1)})$ such that both elements do not vanish. This is because, for every smooth compactly supported $1$-form  $f^{(1)}$ (with $\delta^{(1)}f^{(1)} \neq 0$ in particular) and for every  smooth compactly supported $2$-form  $f^{(2)}$, 
$$a^{(0)}:=\delta_h^{(1)}f^{(1)} \qquad \pi^{(1)}:=\left(\Delta_h^{(1)}+m^2\right)f^{(1)} + \delta_h^{(2)} f^{(2)}$$
 are  smooth, and compactly supported,   they solve the nontrivial constraint in (\ref{CCdata})
 $\delta_h^{(1)} \pi^{(1)} = (\Delta_{(0)} + m^2) a^{(0)}$ and $f^{(1)}, f^{(2)}$ can be chosen in order that  neither of $a^{(0)}$ and $\pi^{(1)}$  vanishes. The easier constraint $\pi^{(0)}= - \delta^{(1)}_h a^{(1)}$ is solved by every smooth compactly supported  $1$-form  $a^{(1)}$ by defining the smooth compactly supported $0$-form  $\pi^{(0)}$
correspondingly.
\end{itemize}

\subsection{The Proca symplectic form in ultrastatic spacetimes} \label{SECSYMPL}Consider two  solutions $A, A' \in \Gamma_{sc}(\V_g) \cap \Ker \P$ of the Proca  equation in our ultrastatic spacetime, choose $t\in \RR$  and consider the bilinear   form
\begin{equation}
\sigma^{(\P)}_t(A,A') := \int_{\Sigma} h^\sharp(a^{(1)}_t,  \pi^{(1)'}_t- da^{(0)'}_t) -   h^\sharp(a^{(1)'}_t,  \pi^{(1)}_t-da^{(0)}_t) \:  \vol_h \:,
 \label{SYMPr}
\end{equation}
where we are referring to the Cauchy data  on $\Sigma$ of the smooth spacelike compact  solutions of the Proca equation. $\Sigma$ is viewed as the time slice at time $t$. As is well known, it is possible to define a natural symplectic form for the Proca field in general globally hyperbolic specetimes \cite{aqft1} with properties analogous to the ones we are going to discuss here. In this section we however stick to the ultrastatic spacetime case which is enough for our ends.
 
According to \cite{aqft1} (with an argument very similar to the proof of Propositions 3.12 and 3.13
in \cite{Norm}) we have immediately that 
\begin{equation*}
\sigma^{(\P)}_t(A,A') = \sigma^{(\P)}_{t'}(A,A')\quad \forall t,t' \in \RR\:,
\end{equation*}
and, omitting the index $t$ as the symplectic form is independent of it, 
\begin{equation}
\sigma^{(\P)}(A,A') = \int_\M g^\sharp\left(\ff, \G_\P \ff'\right) \:  \vol_g \label{GSIGMA}
\end{equation}
where $A,f$ (resp. $A,f'$) are related by $A:= \G_\P \ff$ (resp. $A' := \G_\P \ff'$).  

\begin{remark} The important identity (\ref{GSIGMA}) is also valid in a generic globally hyperbolic spacetime when $\sigma^{(P)}$ is interpreted as the general symplectic form of the Proca field according to \cite{aqft1}.\end{remark}

Let us suppose to  deal with the Cauchy data of  the real vector space $C_\Sigma\subset  \Omega^{0}_c(\Sigma)^2\times \Omega^{1}_c(\Sigma)^2$   of smooth compactly supported  Cauchy data  $(a_0,\pi_0, a_1,\pi_1) $ subjected to the linear constraints~(\ref{CCdata}), 
\begin{equation}
 C_\Sigma := \left\{ (a^{(0)},\pi^{(0)}, a^{(1)},\pi^{(1)}) \in  \Omega^{0}_c(\Sigma)^2\times \Omega^{1}_c(\Sigma)^2 \:\left|\: \pi^{(0)}  =- \delta^{(1)}_h a^{(1)} \:, \quad   (\Delta_h^{(0)} + m^2) a^{(0)} = \delta^{(1)}_h \pi^{(1)}\right. \right\}\:. \label{CCdata2}
\end{equation}
Not only the Cauchy problem is well behaved in that space as a consequence of Proposition \ref{TEOEUP},  but we also have the following result which, in particular, implies that the Weyl algebra of the real Proca field has trivial center.
\begin{proposition}
The bilinear antisymmetric map $\sigma^{(P)}: C_\Sigma \times C_\Sigma \to \RR$ defined in (\ref{SYMPr}) is non-degenerate and therefore it is a symplectic form on $C_\Sigma$.
\end{proposition}

\begin{proof} Taking (\ref{RAN}) into account,
suppose that $\Gamma_{sc}(\V_g) \cap  \Ker \P \ni  A' = \G_\P \ff$ whose Cauchy data are $(a^{(0)'},\pi^{(0)'}, a^{(1)'},\pi^{(1)'})  \in C_\Sigma$ is such that $\sigma^{(\P)}(A,A')=0$ for all $A = \G_\P \ff
\in \Gamma_{sc}(\V_g) \cap  \Ker \P \equiv  C_\Sigma$, we want to prove that  $A'=0$ namely, its initial conditions are  $(0,0,0,0)$.
From (\ref{GSIGMA}), using the fact that $g^\sharp$ is non-degenerate, we have that $A' = \G_\P \ff' =0$ so that its Cauchy data are the zero data in view of the well-posedness of the Cauchy problem Proposition \ref{TEOEUP}.
\end{proof}

To conclude this section we prove that, when using Cauchy data in $C_\Sigma$,  the expression of $\sigma^{(\P)}$ can be re-arranged in order to make contact with the analogous symplectic forms of the two Klein-Gordon fields $ A^{(0)}$ and $ A^{(1)}$ the solution $A$ is made of,  as discussed  in Section \ref{SECCP}. Indeed, remembering the constraint $\pi^{(0)} =- \delta^{(1)}_h  a^{(1)}$, and using the duality of $\delta$ and $d$, part of the  integral in the right-hand side of (\ref{SYMPr}) can be rearranged to
\begin{align*}
 \int_\Sigma h^\sharp( a^{(1)}_t, da^{(0)'}_t) -   h^\sharp( a^{(1)'}_t, da^{(0)}_t)  \:  \vol_h &=   \int_\Sigma h^\sharp(\delta_h^{(1)} a^{(1)}_t, a^{(0)'}_t) -   h^\sharp(\delta_h^{(1)}a^{(1)'}_t, a^{(0)}_t)  \:  \vol_h \\
 &=  -\int_\Sigma h^\sharp(\pi^{(0)}_t, a^{(0)'}_t) -   h^\sharp(\pi^{(0)'}_t, a^{(0)}_t)  \:  \vol_h
\:.
\end{align*}
As a consequence, if $\eta_i =0$ for $i=1$ and $\eta_i=-1$ for $i=0$ and $h^\sharp_{(i)}$ is $h^\sharp$ for $i=1$ and the pointwise product for $i=0$,
\begin{equation}
\sigma^{(\P)}(A,A') = \sum_{i=0}^1 \eta_i \int_{\Sigma_t} h_{(i)}^\sharp(a^{(i)}_t, \pi^{(i)'}_t) -   h_{(i)}^\sharp(a^{(i)'}_t,  \pi^{(i)}_t) \:  \vol_h  \label{NEWSP}\:.
\end{equation} 
In other words,  referring to the (Klein-Gordon) symplectic forms introduced in \cite{Norm} for normally hyperbolic equations (\ref{EQ1}) and (\ref{EQ2})
\begin{equation*} \sigma^{(\P)}(A,A') =\sigma^{(1)}( A^{(1)},  A^{(1)'})  - \sigma^{(0)} ( A^{(0)},  A^{(0)'})
\end{equation*}
where $\sigma^{(k)}$ is the symplectic form for a normally hyperbolic field operator on a real vector bundle 
defined, e.g.,  \cite[Proposition 3.12]{Norm}.  

A similar result is valid for the causal propagators. Decomposing $\ff = \ff^{(0)}dt + \ff^{(1)} \in \Gamma_c(\V_g)$ where $\ff^{(0)} \in \Gamma_c(\V^{(0)}_g)$ and  $\ff^{(1)} \in \Gamma_c(\V^{(1)}_g)$, (\ref{GSIGMA}), the analogs for scalar and vector Klein Gordon fields \cite{Norm} and (\ref{NEWSP}) imply
\begin{equation*}
 \int_\M g^\sharp(\ff, \G_\P \ff')  \vol_g =  \int_\M h^\sharp(\ff^{(1)}, \G^{(1)} \ff^{(1)'})  \vol_g -  \int_\M \ff^{(0)} \G^{(0)} \ff^{(0)'} \vol_g 
\end{equation*}
where $\G^{(i)}$, $i=0,1$ are the causal propagators  for the normally hyperbolic operators 
$$\N^{(i)}:=\partial_t^2  + \Delta^{(i)}_h + m^2I  : \Gamma_{sc}(\V_g^{(i)}) \to  \Gamma_{sc}(\V_g^{(i)})\quad i=0,1$$ according to the theory of \cite{Norm}. Here $\Delta^{(0)}_h$ coincides with the standard Laplace-Beltrami operator for scalar fields on $\Sigma$.

\begin{remark}\label{REMSYMP} With the same argument, the found results immediately generalize to the case of {\em complex} $k$-forms. More precisely, if the Cauchy data belong to $C_\Sigma + i C_\Sigma$,
\begin{equation*}
 \sigma^{(\P)}(\overline{A},A') =\sigma^{(1)}( \overline{A^{(1)}},  A^{(1)'})  - \sigma^{(0)} ( \overline{A^{(0)}},  A^{(0)'})\:, 
\end{equation*}
where the left-hand side is again (\ref{SYMPr}) evaluated for complex Proca fields, i.e., complex Cauchy data.
Above, the bar denotes the complex conjugation and  the Cauchy data of the considered  complex Proca fields
 satisfy the constraints  (\ref{CCdata}).  Furthermore 
\begin{equation*}
 \int_\M g^\sharp(\overline{\ff}, \G_\P \ff')  \vol_g =  \int_\M h^\sharp(\overline{\ff^{(1)}}, \G^{(1)} \ff^{(1)'})  \vol_g -  \int_\M \overline{\ff^{(0)}} \G^{(0)} \ff^{(0)'} \vol_g
 \end{equation*}
where the smooth compactly supported sections are complex. We have used the same symbols as for the real case for the causal propagators since the associated operators commute with the complex conjugation. As a consequence, a standard argument about the uniqueness of Green operators implies that the causal propagators for the real case are nothing but the restriction of the causal propagator of the complex case which, in turn, are the trivial complexification of the real ones.
\end{remark}

\subsection{The Proca energy density  in ultrastatic spacetimes} Starting from the Proca Lagrangian in every curved spacetime (see, e.g, \cite{ProcaFT})
$${\cal L} =   - \frac{1}{4}F_{\mu\nu}F^{\mu \nu} - \frac{m^2}{2} A_\mu A^\mu\quad \mbox{with}\quad F_{\mu\nu}:= \partial_\mu A_\nu - \partial_\nu A_\mu$$
and referring to local coordinates $(x^0,\ldots, x^{n-1})$ adapted to the split $\M= \RR \times \Sigma$ of our ultrastatic spacetime, where $x^0=t$ runs along the whole $\RR$ and $x^1,\ldots, x^{n-1}$ are local coordinates on $\Sigma$, the energy density reads in terms of initial conditions on $\Sigma$ of the considered Proca field
\begin{equation}\label{POSEN}
\begin{aligned}
 T_{00} = \frac{1}{2}\, & h^\sharp( \pi^{(1)}-da^{(0)},  \pi^{(1)}-da^{(0)}) + \frac{1}{2} h_{(2)}^\sharp (da^{(1)}, da^{(1)}) \\
 &+ \frac{m^2}{2} \left( h^\sharp(a^{(1)}, a^{(1)})
+ a^{(0)}a^{(0)}\right) \geq 0\:. 
\end{aligned}
\end{equation}
Above $ h_{(2)}^\sharp$ is the natural scalar product for the $2$-forms on $\Sigma$ induced by the metric tensor.
It is evident that the energy density is non-negative since the metric $h$ and its inverse $h^\sharp$ are positive by hypothesis. 
The total energy at time $t$ is the integral of $T_{00}$ on $\Sigma$, using the natural volume form, when replacing $A^{(0)}$ and $ A^{(1)}$ for the respective Cauchy data. As $\partial_t$ is a Killing vector and the solution is spacelike compact, the total energy is finite and constant in time.  
\begin{equation} \label{EP}
\begin{aligned}
E^{(P)} = \frac{1}{2}\int_\Sigma &\Big( h^\sharp( \pi^{(1)}-da^{(0)},  \pi^{(1)}-da^{(0)}) +  h_{(2)}^\sharp (da^{(1)}, da^{(1)}) \\
& + m^2\big( h^\sharp(a^{(1)}, a^{(1)})
+ a^{(0)}a^{(0)}\big) \Big)\vol_h\:. 
\end{aligned}
\end{equation}
Using Hodge duality of $d$ and $\delta$ and the definition of the Hodge Laplacian,   the expression of the total energy can be re-arranged to
$$E^{(P)} = \frac{1}{2}\int_\Sigma \Big( h^\sharp( \pi^{(1)},  \pi^{(1)}) + h^\sharp (da^{(0)},da^{(0)}) -2 h^\sharp( \pi^{(1)}, da^{(0)}) - \delta^{(1)}_h  a^{(1)}  \delta^{(1)}_h a^{(1)}$$ $$+ h^\sharp\big( a^{(1)}, \Delta_h^{(1)} a^{(1)})
+ m^2(a^{(0)}a^{(0)} + h^\sharp(a^{(1)},a^{(1)})\big) \Big) \vol_h\:.$$
 Using again the Hodge duality of $d$ and $\delta$   the third term in the integral can be rearranged to $$- \int_\Sigma h^\sharp( \pi^{(1)}, da^{(0)}) \vol_h =- \int_\Sigma a^{(0)} \delta^{(1)}_h \pi^{(1)} \vol_h\:.$$ The term  $\delta^{(1)}  \pi^{(1)}$ above and the term $\delta^{(1)}_h a^{(1)}  \delta^{(1)}_h a^{(1)}$ appearing in the expression for the total energy
can be worked out 
exploiting  the constraints
(\ref{CCdata}). Inserting the results in the found formula for the total energy, we finally find, with the notation already used for the symplectic form, 
\begin{equation}E^{(P)} = \sum_{i=0}^1 \eta_i \frac{1}{2}\int_\Sigma h_{(i)}^\sharp( \pi^{(i)},  \pi^{(i)}) + h_{(i)}^\sharp(a^{(i)}, (\Delta_h^{(i)} + m^2I)a^{(i)})\: \vol_h
 \label{2T00}\:,
\end{equation}
when the used Cauchy data belong to the constrained space $C_\Sigma$.
It is now  clear that the total energy of the Proca field is the difference between the total energies of the two Klein-Gordon fields composing it exactly as it happened for the symplectic form. This difference is however positive when working on smooth compactly supported  initial conditions satisfying the constraints (\ref{CCdata}), because the found expression of the energy is the same as the one computed with the density (\ref{POSEN}).

\begin{remark}
We notice that the negative energy component of the field can be interpreted as a ghost, in this case however no issues arise since dynamical constraints covariantly remove such a state. A different approach to the problem by generalizing to curved spacetime the Stuckelberg lagrangian, can be found in \cite{Folacci}, where it is appearently argued the no Hadamard states exist for the Proca field, contrarily to the results of \cite{Fewster} and of this work.
\end{remark}

\begin{remark} With the same argument, the found result immediately generalizes to the case of {\em complex} $k$-forms and one finds
\begin{equation}
\begin{aligned}
 \sum_{i=0}^1 \eta_i \frac{1}{2}\int_\Sigma & h_{(i)}^\sharp( \overline{\pi^{(i)}},  \pi^{(i)}) + h_{(i)}^\sharp( \overline{a^{(i)}}, (\Delta_h^{(i)} + m^2I)a^{(i)})\vol_h =\\
=  \frac{1}{2}\int_\Sigma & \Big( h^\sharp(\overline{\pi^{(1)}-da^{(0)}},  \pi^{(1)}-da^{(0)}) +  h_{(2)}^\sharp (\overline{da^{(1)}}, da^{(1)}) \\ & + m^2\big( h^\sharp(\overline{a^{(1)}}, a^{(1)})
+ \overline{a^{(0)}}a^{(0)}\big) \vol_h \Big) \geq 0
\label{POSCOMP}
\end{aligned}
\end{equation}
where the bar over the forms denotes the complex conjugation and 
 $ (a^{(0)}, \pi^{(0)}, a^{(1)}, \pi^{(1)})$ are {\em complex} forms of  $C_\Sigma + i C_\Sigma$.
\end{remark}

\subsection{Elliptic Hilbert complexes and Proca quantum states in ultrastatic  spacetimes}
We can proceed to the construction of quasifree  states. As we shall see shortly, this construction for the Proca field uses some consequences of the spectral theory applied to the theory of {\em elliptic Hilbert complexes}  \cite{BuLe} defined  in terms of the closure of Hodge operators in natural $L^2$ spaces of forms.  

Some of the following ideas were inspired by~\cite{Fewster}. However we now work in the space of Cauchy data instead of in the space of smooth supportly compacted forms and/or modes. Furthermore we do not assume restrictions on the topology of the Cauchy surfaces used in \cite{Fewster} to impose a pure point spectrum to the Hodge Laplacians.

To define quasifree states for the Proca field we observe that,  as $\P$ is Green hyperbolic,  the CCR algebra $\mathcal{A}_g$ is isomorphic to the analogous unital $*$-algebra $\mathcal{A}^{(symp)}_g$  generated by the {\bf solution-smeared field operators} $\sigma^{(\P)}(\hat{\fa}, A)$, for $A\in \Ker_{sc}(P)$, which are $\RR$-linear in $A$, Hermitian,  and  satisfy the commutation relations\footnote{Notice that,  as $\sigma^{(\P)}(A,A')$ is non degenerate, we have that $\sigma^{(\P)}(\hat{\fa}, A)=0$ only if $A=0$.}
\begin{equation}\left[\sigma^{(\P)}(\hat{\fa}, A),\sigma^{(\P)}(\hat{\fa}, A')\right] = i\sigma^{(\P)}(A,A') I\:.\label{CARSYMP}\end{equation}
The said unital $*$-algebra isomorphism $F: \mathcal{A}_g \to \mathcal{A}^{(symp)}_g$ is completely defined as the unique homomorphism of unital $*$-algebras that satisfies  $$F: \hat{\fa}(\ff) \mapsto \sigma^{(\P)}\left(\hat{\fa}, \G_\P \ff\right)\quad \mbox{with $A= \G_\P\ff$, $\quad\ff\in\Gamma_c(\V_g)$}\:.$$
The definition is well-posed in view of (\ref{GSIGMA}), (\ref{RAN}),  (\ref{cruciallemma0}), and the  definition of $\mathcal{A}_g$.
Within this framework, the two point function  $\omega_2$ is interpreted as the integral kernel of $$\omega\left(\sigma^{(\P)}(\hat{\fa}, A) \sigma^{(\P)}(\hat{\fa}, A')\right)\:.$$ In particular, its antisymmetric part is universally given by  $\frac{i}{2} \sigma^{(P)}(A,A')$ due to (\ref{CARSYMP}). The specific part of the two point function is therefore completely embodied in its symmetric part $\mu(A,A')$.

   According to this observation, a  general recipe for real  (bosonic) CCR in generic globally hyperbolic spacetimes  to  define a quasifree state on the $^*$-algebra $\mathcal{A}_g$  
   (e.g., see \cite{KW,W,IgorValter} for the scalar case and   \cite[Chapter 4, Proposition 4.9]{gerardBook} for the generic case of real bosonic CCRs)  is to assign  a   real scalar product on the space of spacelike compact solutions
$$\mu : \Ker_{sc}(\P) \times  \Ker_{sc}(\P) \to \RR$$
satisfying
\begin{itemize}
\item[(a)] the strict  positivity requirement $\mu(A,A)\geq 0$ where $\mu(A,A) =0$ implies $A =0$;  
\item[(b)] the  continuity requirement with respect to the relevant  symplectic form $\sigma^{(P)}$ (see, e.g., \cite[Proposition 4.9]{gerardBook}), 
\begin{equation}
\sigma^{(\P)}(A,A')^2 \leq 4 \mu(A,A) \mu(A',A') \label{diseq}\:.
\end{equation}
\end{itemize}
 The  continuity requirement  directly arises form the fact that the quasifree state induced by $\mu$ on the whole $*$-algebra $\mathcal{A}_g\equiv \mathcal{A}_g^{symp}$
according to Definition \ref{quasifree} 
 is a positive functional. The converse implication, though true,  is less trivial \cite{KW,gerardBook}. 
  The two mentioned requirements are nothing but the direct translation of  (2)' and (3)' stated in the introduction. (Regarding the latter, observe that  $\sigma^{(P)}$  corresponds to the causal propagator at the level of solutions -- Eq.~(\ref{GSIGMA}) in our case --
 as discussed in Section \ref{SECSYMPL}.)
At this point, it should be clear that the quasifree state defined by $\mu$ has two-point function, viewed as bilinear map on $\Gamma_c(\V_g)\times \Gamma_c(\V_g)$,
$$
\omega_\mu(\fa(\ff)\fa(\ff')) = \omega_{\mu2}(\ff,\ff') := \mu(\G_\P\ff, \G_\P\ff') + \frac{i}{2}\sigma^{(P)}(\G_\P\ff, \G_\P\ff')\:. 
$$

However, since the Cauchy problem is well posed 
on the time slices $\Sigma$ of an ultrastatic spacetime $(\RR\times \Sigma,  - dt\otimes dt + h)$,
as proved in
Proposition~\ref{TEOEUP}, we can directly define $\mu$ (and $\sigma^{(\P)}$) in the space of Cauchy data $C_\Sigma$ on $\Sigma$, for smooth spacelike compact solutions, viewed as the time slice at $t=0$,
\begin{equation*}
\mu : C_\Sigma \times C_\Sigma \to \RR\:.
\end{equation*}
In view of the peculiarity of the Cauchy problem for the Proca field as discussed in Section~\ref{SECCP},  the real vector space of the Cauchy data $C_\Sigma$ is {\em constrained}. We underline that working at the level of constrained initial data does not affect the construction of quasifree states. Indeed, it is sufficient that the space of constrained initial conditions is a real (or complex) vector space and that the constrained Cauchy problem is well posed. 
With this in mind, referring to  the canonical decomposition $A = A^{(0)} dt + A^{(1)}$ of a real  smooth spacelike compact solution $A$ of the Proca equation, we remember that
\begin{equation}
 C_\Sigma := \left\{ (a^{(0)},\pi^{(0)}, a^{(1)},\pi^{(1)}) \in \Omega^{0}_c(\Sigma)^2\times \Omega^{1}_c(\Sigma)^2 \:\left|\: \pi^{(0)}  =- \delta^{(1)}_h a^{(1)} \:, \quad   (\Delta_h^{(0)} + m^2) a^{(0)} = \delta^{(1)}_h \pi^{(1)}\right. \right\}\:. \nonumber
\end{equation}
Above $(a^{(0)},\pi^{(0)}) := (A^{(0)}, \partial_t A^{(0)})|_{t=0}$ and $(a^{(1)},\pi^{(1)}) := (A^{(1)}, \partial_t A^{(1)})|_{t=0}$.

With the said definitions and 
{\em where $A$ denotes both a solution of Proca equation and its Cauchy data on $\Sigma$},
we have the first result.
\begin{proposition} \label{PROPHAD1}
Consider the  $*$-algebra $\mathcal{A}_g$ of the real Proca field on the ultrastatic spacetime $(\M,g)= (\RR\times \Sigma, -dt\otimes dt+h)$, with $dt$ past directed and  $(\Sigma, h)$ a smooth complete Riemannian manifold.
 Let $\eta_0:=-1$, $\eta_1:=1$ and $h^\sharp_{(j)}$ denote the standard inner product of $j$-forms on $\Sigma$ induced by $h$.
The bilinear  map on the space $C_\Sigma$ of real smooth compactly supported Cauchy data (\ref{CCdata2}) \begin{equation}\mu(A,A') := \sum_{j=0}^1 \frac{\eta_j}{2} \int_\Sigma   h_{(j)}^\sharp(\pi^{(j)}, (\overline{\Delta^{(j)} + m^2})^{-1/2} \pi^{(j)'})    +  h^\sharp_{(j)}  (a^{(j)},  (\overline{\Delta^{(j)} + m^2})^{1/2} a^{(j)'}  )\: \vol_h 
 \label{DEFMU} 
 \end{equation}
is a well defined symmetric positive inner product which  satisfies (\ref{diseq}) and thus it defines a quasifree state $\omega_\mu$  on  $\mathcal{A}_g$ completely defined by its two-point function
\begin{equation}\label{omega2R}
\omega_\mu\left(\fa(\ff)\fa(\ff')\right) = \omega_{\mu2}(\ff,\ff') := \mu\left(\G_\P\ff, \G_\P\ff'\right) + \frac{i}{2}\sigma^{(P)}\left(\G_\P\ff, \G_\P\ff'\right)\ 
\end{equation}
where $ \ff,\ff' \in \Gamma_c(\V_g) $ satisfy
$$\sigma^{(P)}\left(\G_\P\ff, \G_\P\ff'\right) =   \int_\M g^\sharp(\ff, \G_\P \ff') \:  \vol_g \:.$$
\end{proposition}

The bar over the operators in (\ref{DEFMU})  denotes the closure in suitable Hilbert spaces of the operators originally defined on domains of compactly supported smooth functions. To explain  this formalism,  before starting with the proof  we have to permit some technical facts about the properties of the Hodge operators at the level of $L^2$ spaces.
Given the complete Riemannian manifold $(\Sigma, h)$, with $n:= \dim (\Sigma)$
consider  the Hilbert space $\mathcal{H}_h:= \bigoplus_{k=0}^{n} L_k^2(\Sigma, \vol_h)$, where the sum is orthogonal and $ L_k^2(\Sigma, \vol_h)$ is the complex Hilbert space   of the square-integrable $k$-forms with respect to the relevant {\em Hermitian} Hodge inner product:
$$( a| b )_k :=   \int_\Sigma  h_{(k)}^\sharp(\overline{a},b)\: \vol_h \:, \quad a,b \in L_k^2(\Sigma, \vol_h)\:,$$
where  $\overline{a}$ denotes the pointwise complex conjugation of the complex form $a$.
The overall inner product on $\mathcal{H}_h$ will be  indicated by $(\cdot|\cdot)$ and the Hilbert space adjoint of a densely-defined  operator $A: D(A) \to \mathcal{H}_h$, with $D(A)\subset \mathcal{H}_h$, will be denoted by $A^*:D(A^*)\to \mathcal{H}_h$. The closure of $A$   will be denoted by the bar: $\overline{A}: D(\overline{A})\to \mathcal{H}_h$.

If $\Omega_c(\Sigma)_\CC := \bigoplus_{k=0}^{n} \Omega^k_c(\Sigma)_\CC$ denotes the dense subspace of complex  complactly supported smooth forms  $\Omega^k_c(\Sigma)_\CC := \Omega^k_c(\Sigma) + i\Omega^k_c(\Sigma)$,
define the two  operators (we omit the index $h$ for shortness)  $$d :=  \oplus_{k=0}^n d^{(k)}: \Omega_c(\Sigma)_\CC \to  \Omega_c(\Sigma)_\CC \:, \quad \delta:= \oplus_{k=0}^n \delta^{(k)} : \Omega_c(\Sigma)_\CC \to  \Omega_c(\Sigma)_\CC$$ 
 with $d^{(n)} :=0$ and  $\delta^{(0)} :=0$.
Finally, introduce the Hodge Laplacian as $$\Delta := \sum_{k=0}^n \Delta^{(k)}:  \Omega_c(\Sigma)_\CC   \to  \Omega_c(\Sigma)_\CC \quad \mbox{with $\Delta^{(k)} := \delta^{(k+1)} d^{(k)} + d^{(k-1)}\delta^{(k)}$}. $$
Since $(\Sigma, h)$ is complete,  $\Delta$ can be proved to be essentially selfadjoint,  for instance exploiting the well-known argument by Chernoff \cite{Chernoff} (or directly referring to  \cite{AV}).  Since $\Delta$ is essentially selfadjoint, if $c\in \RR$, also $\Delta +cI$ is essentially selfadjoint. In particular,  its unique selfadjoint extension is its closure  $\overline{\Delta +cI}$.

Referring to the theory of {\em elliptic Hilbert complexes} developed in \cite[Section 3]{BuLe} and focusing in particular on   \cite[Lemma 3.3]{BuLe} based on previous achievements established in \cite{AV},  we can conclude that  the following couple of facts are true.  (The compositions of operators are henceforth defined with  their natural domains: $D(A+B):= D(A) \cap D(B)$,  $D(AB) =\{x\in D(B)\:|\: Bx\in D(A)\}$, $D(aA):= D(A)$ for $a\neq 0$, $D(0A) := \mathcal{H}_h$, and $A\subset B$ means $D(A)\subset D(B)$ with $B|_{D(A)}=A$.)
\begin{itemize}
\item[(a)] The identities hold
\begin{equation}
\overline{d}^* = \overline{\delta}\:, \quad \overline{\delta}^* = \overline{d} \label{agg}
\end{equation}
where $^*$ denotes the adjoint  in the Hilbert space $\mathcal{H}_h$.
\item[(b)]  The  unique selfadjoint extension $\overline{\Delta}$ of $\Delta$ satisfies
\begin{equation}
\overline{\Delta} = \overline{d}\:\overline{\delta}+\overline{\delta}\overline{d} =  \sum_{k=0}^n \overline{\Delta^{(k)}} \quad \mbox{with $\overline{\Delta^{(k)}} := \overline{\delta^{(k+1)}}\:\overline{d^{(k)}} + \overline{d^{{k-1}}}\:\overline{\delta^{(k)}}$}. \label{Deltaov}
\end{equation}
A trivial generalization of the  decomposition as in (\ref{Deltaov}) holds for $\overline{\Delta + cI} = \overline{\Delta} + cI$ with $c\in \RR$.
\end{itemize}

We are now prompt to prove  a preparatory technical lemma --  necessary to establish   Proposition \ref{PROPHAD1} --
 that will be fundamental for showing that the bilinear map $\mu$ is positive on the space $C_\Sigma$.

\begin{lemma}\label{LEMMAFORME}
 For every given $k=0,1,\ldots, n$, $c>0$, and $\alpha \in \RR$, the identites hold
\begin{align*}
(\overline{\Delta^{(k+1)} +cI})^{\alpha}  \overline{d^{(k)}} x &= \overline{d^{(k)}} (\overline{\Delta^{(k)}+cI})^{\alpha}x\:, \quad 
\forall x \in D((\overline{\Delta^{(k)}+ cI})^\alpha) \cap D((\overline{{\Delta}^{(k+1)}+ cI})^\alpha  \overline{d^{(k)}})\\
 (\overline{\Delta^{(k-1)} +cI})^{\alpha}  \overline{\delta^{(k)}} y &= \overline{\delta^{(k-1)}} (\overline{\Delta^{(k)}+cI})^{\alpha}y\:, \quad 
\forall y \in D((\overline{\Delta^{(k)}+ cI})^\alpha) \cap D((\overline{{\Delta}^{(k-1)}+ cI})^\alpha  \overline{\delta^{(k)}}) \:.
\end{align*}
\end{lemma}

\begin{proof}
 Since $dd=0$ and $\delta\delta=0$, from (\ref{agg}), we also have 
 $\overline{d}\:\overline{d}x=0$ if $x\in D(\overline{d})$ and $\overline{\delta}\:\overline{\delta}y=0$ 
if $y\in D(\overline{\delta})$, 
and thus (\ref{Deltaov}) yields\footnote{It holds $(B+C)A = BC+BA$, but $AB+AC \subset A(B+C)$.}
$$\overline{d}\: \overline{\Delta} \supset  \overline{d}\:\overline{\delta}\: \overline{d} = \overline{\Delta}\: \overline{d}\:.$$ 
However, if $D(\overline{d}\: \overline{\Delta}) \supsetneq D(\overline{d}\:\overline{\delta}\: \overline{d})$, we would have $x \in D(\overline{\Delta})= D(\overline{\delta}\:\overline{d}) \cap 
D(\overline{d}\:\overline{\delta})$ such that $\overline{\Delta} x= \overline{\delta}\:\overline{d}x + \overline{d}\:\overline{\delta}x \in D(\overline{d})$, but
$x\not \in D(\overline{d}\:\overline{\delta}\:\overline{d} )$, namely $\overline{\delta}\:\overline{d}  x \not \in D(\overline{d})$. This is impossible since 
$\overline{\delta}\:\overline{d}x + \overline{d}\:\overline{\delta}x \in D(\overline{d})$,
$D(\overline{d})$ is a subspace and 
$\overline{d}\:\overline{\delta}x \in D(\overline{d})$ (and more precisely $\overline{d}\:\overline{d}\:\overline{\delta}x =0$ as stated above). Therefore 
$$\overline{d}\: \overline{\Delta} =  \overline{d}\:\overline{\delta}\: \overline{d} = \overline{\Delta}\: \overline{d}$$
and the same result is valid  with $\delta$ in place of $d$.
Evidently, in both cases $\overline{\Delta}$ can be replaced by  the selfadjoint operator $\overline{\Delta}+cI$ = $\overline{\Delta + cI}$ for every $c\in \RR$:
\begin{equation}
\overline{d}\: \overline{\Delta +cI} = \overline{\Delta + cI}\: \overline{d}\:, \quad  \overline{\delta}\: \overline{\Delta +cI} = \overline{\Delta+cI}\: \overline{\delta} \:. \label{commut}
\end{equation} 
{\em We henceforth assume  $c>0$}. In that case,  as $\Delta$ is already positive on its domain,   the spectrum of  the selfadjoint operator $\overline{\Delta+cI}$  is strictly positive and thus $ \overline{\Delta +cI}^{-1}: \mathcal{H}_h \to D(\overline{\Delta +cI})$  is well defined, selfadjoint  and bounded. The former identity in (\ref{commut}) also implies that $D(\overline{d}\: \overline{\Delta +cI})=D(\overline{\Delta +cI} \: \overline{d})$, so that
$$ \overline{\Delta +cI}^{-1} \overline{d}\:  \overline{\Delta +cI}|_{D(\overline{d}\: \overline{\Delta +cI})}  x =  \overline{d}|_{D(\overline{d}\: \overline{\Delta +cI})}x \:.$$
By construction, we can choose $x=  \overline{\Delta +cI}^{-1} y$ with $y\in D(  \overline{d})$  in view of the definition of the natural domain  of the composition 
$\overline{d}\: \overline{\Delta +cI})$. In summary
$$  \overline{\Delta +cI}^{-1} \overline{d} y =  \overline{d}\:   \overline{\Delta +cI}^{-1} y\:,  \quad \forall y   \in D(  \overline{d})\:.$$
Since the argument is also valid for $\delta$, we have established that
\begin{equation*}
 \overline{\Delta +cI}^{-1} \overline{d} \subset   \overline{d}\:   \overline{\Delta +cI}^{-1}\:,  \quad 
 \overline{\Delta +cI}^{-1} \overline{\delta} \subset   \overline{\delta}\:   \overline{\Delta +cI}^{-1}
\end{equation*}
Iterating the argument, for every $n=0,1, \ldots$,
$$
(\overline{\Delta +cI}^{-1})^n \overline{d} \subset   \overline{d}\:   (\overline{\Delta +cI}^{-1})^n\:,  \quad 
 (\overline{\Delta +cI}^{-1})^n \overline{\delta} \subset   \overline{\delta}\:   (\overline{\Delta +cI}^{-1})^n\:.
$$ This result extends to complex polynomials of $\overline{\Delta +cI}^{-1}$ in place of powers by linearity. Using the spectral calculus  (see e.g. \cite{FFM}) where $\mu_{xy}(E) = ( x| P_Ey)$  and $P$ is the projector-valued spectral measure of $\overline{\Delta +cI}^{-1}$, the found result for $\overline{d}$ can be written
\begin{equation}\int_{[0,b]} p(\lambda) d\mu_{x, \overline{d} y}(\lambda) = \int_{[0,b]} p(\lambda) d\mu_{\overline{\delta}x, y}(\lambda)  \label{estimate}\end{equation}
for every complex polynomial $p$, where $[0,b]$ is a sufficiently large interval to include the (bounded positive) spectrum of $\overline{\Delta +cI}^{-1}$, $x\in D(\overline{\delta})$,
$y \in D(\overline{d})$, and where we have used $\overline{\delta}= \overline{d}^*$. Since the considered regular Borel complex measures are finite and supported on the compact $[0,b]$, we can pass in (\ref{estimate}) from polynomials $p$ to generic continuous functions $f$ in view of the Stone-Weierstrass theorem.  At this juncture, $P_E^*=P_E$ and  the uniqueness part of Riesz' representation theorem for regular complex Borel measures, implies that
$$( P_E \overline{\delta}  y |  x) =  (P_E y| \overline{d}x) \quad \mbox{for all $x\in D(\overline{\delta})$, 
$y \in D(\overline{d})$, and every Borel set $E\subset \RR$.}$$
which means $P_E \overline{\delta} \subset \overline{d}^* P_E$, namely $P_E \overline{\delta} \subset \overline{\delta} P_E$. Analogously, we also have
 $P_E \overline{d} \subset \overline{d} P_E$. 

If $f: \RR \to \CC$ is measurable and {\em bounded}, the standard spectral calculus and 
(\ref{agg}), with a procedure similar to the one used to prove $P_E \overline{\delta} \subset \overline{\delta} P_E$ and taking into account the fact that $D(f(\overline{\Delta +cI}^{-1}))=\mathcal{H}_h$, yields
\begin{equation} f(\overline{\Delta +cI}^{-1}) \overline{\delta} \subset  \overline{\delta} f(\overline{\Delta +cI}^{-1})\:, \quad
f(\overline{\Delta +cI}^{-1}) \overline{d} \subset  \overline{d} f(\overline{\Delta +cI}^{-1}) \label{twoint} \end{equation}
If $f$ is unbounded, we can choose a sequence of bounded measurable functions $f_n$ such that $f_n \to f$ pointwise. It is easy to prove that (see, e.g. \cite{FFM}) 
$x \in D(\int_\RR f dP)$ entails  $\int_\RR f_n dP x \to  \int_\RR f dP x $. 
This is the case for instance for $f(\lambda) = \lambda^{\beta}$ with $\beta <0$ restricted to $[0,b]$.
Referring to this function and the pointed out result for some sequence of bounded functions with $f_n \to f$ pointwise,  the latter of (\ref{twoint}) implies that\footnote{Below, $\alpha > 0$ otherwise $(\overline{\Delta+ cI})^\alpha$ is bounded in view of its spectral properties and (\ref{twoint})  is enough to conclude the proof.} ,  
$$(\overline{\Delta+ cI})^\alpha \overline{d} x = \overline{d} (\overline{\Delta+ cI})^\alpha  x\quad \mbox{if $x\in D((\overline{\Delta+ cI})^\alpha) \cap D(\overline{d})$ and  $\overline{d}x \in D((\overline{\Delta+ cI})^\alpha) $, }$$
where we used also the fact that $\overline{d}$ is closed. The case of $\delta$ can be worked out similarly.
Summing up, we have proved that, if  $\alpha \in \RR$,
\begin{align*}
(\overline{\Delta +cI})^{\alpha}  \overline{d} x &= \overline{d} (\overline{\Delta +cI})^{\alpha}x\:, \quad 
\forall x \in D((\overline{\Delta+ cI})^\alpha) \cap D((\overline{\Delta+ cI})^\alpha  \overline{d}) \\
(\overline{\Delta +cI})^{\alpha}  \overline{\delta} y &= \overline{\delta} (\overline{\Delta +cI})^{\alpha}y\:, \quad 
\forall y \in   D((\overline{\Delta+ cI})^\alpha) \cap D((\overline{\Delta+ cI})^\alpha  \overline{\delta}) \:.
\end{align*}
Let us remark that for $\alpha\leq 0$ it is sufficient to choose $x\in D(\overline{d})$  and $y\in D(\overline{\delta})$.
 For every given $k=0,1,\ldots, n$, $c>0$, and $\alpha \in \RR$, taking the decomposition of $\mathcal{H}_h$ into account  the above formulae imply
\begin{align*}
(\overline{\Delta^{(k+1)} +cI})^{\alpha}  \overline{d^{(k)}} x &= \overline{d^{(k)}} (\overline{\Delta^{(k)}+cI})^{\alpha}x\:, \quad 
\forall x \in D((\overline{\Delta^{(k)}+ cI})^\alpha) \cap D((\overline{{\Delta}^{(k+1)}+ cI})^\alpha  \overline{d^{(k)}})\\
 (\overline{\Delta^{(k-1)} +cI})^{\alpha}  \overline{\delta^{(k)}} y &= \overline{\delta^{(k-1)}} (\overline{\Delta^{(k)}+cI})^{\alpha}y\:, \quad 
\forall y \in D((\overline{\Delta^{(k)}+ cI})^\alpha) \cap D((\overline{{\Delta}^{(k-1)}+ cI})^\alpha  \overline{\delta^{(k)}}) \:.
\end{align*}
That is the thesis. 
\end{proof}

We are now prompted to prove that the bilinear map defined by Equation~\eqref{DEFMU} defines a quasifree state defined by the two-point function given by~\eqref{omega2R} establishing the thesis of 
 Proposition~\ref{PROPHAD1}.

\begin{proof}[Proof of Proposition~\ref{PROPHAD1}]
To continue with the proof of the proposition, we now demonstrate that $\mu$ is well-defined and positive. 
That bilinear form is well-defined because  $\Omega_c^{(j)}(\Sigma) \subset D( \overline{\Delta^{(j)} + m^2I}^{\alpha})$ for $\alpha \leq 1$ as one immediately proves from spectral calculus. Furthermore, the integrand in the right-hand side of Equation~(\ref{DEFMU}) is the linear combination of  products of $L^2$ functions (of which one of the two has also compact support). Let us pass to the positivity issue.  
Our strategy is to re-write $\mu(A,A)$, where $A=  (a^{(0)},\pi^{(0)}, a^{(1)},\pi^{(1)})\in C_\Sigma $,  as the quadratic form of the energy 
$\mu(A,A) =E^{(P)}(A_o)$, where the right-hand side is defined in Equation~(\ref{EP}), for a new set of initial data $A_o$ which are not necessarily smooth and compactly supported but such that $E^{(P)}(A_o)$ is well defined.
If $A \in C_\Sigma$, define  for $j=0,1$
\begin{equation} 
\begin{aligned}
& A_o=  (a^{(0)}_o,\pi^{(0)}_o, a^{(1)}_o,\pi^{(1)}_o) \\   & a^{(j)}_o := (\overline{\Delta^{(j)} + m^2I})^{-1/4} a^{(j)}\\
&  \pi^{(j)}_o := (\overline{\Delta^{(j)} + m^2I})^{-1/4} 
\pi^{(j)}\label{newA} 
\end{aligned}
\end{equation}
Notice that the definition is well posed
and the forms $a^{(j)}_o$ and $\pi^{(j)}_o$ belong to the respective Hilbert spaces of $j$-forms, 
 because $\Omega_c^{(j)}(\Sigma) \subset D( \overline{\Delta^{(j)} + m^2I}^{\alpha})$ for $\alpha \leq 1$ as said above. Furthermore the new forms are real since the initial ones are real and  $ \overline{\Delta^{(j)} + m^2I}^{\alpha}$ commutes with the complex conjugation\footnote{It easily arises from spectral calculus using the fact that the complex conjugation is bijective from $\mathcal{H}_h$ to $\mathcal{H}_h$, continuous, and commutes with $\overline{\Delta^{(j)} + m^2I}$.}.
At this juncture, we have from (\ref{DEFMU})
\begin{equation}\mu(A,A) =  \sum_{j=0}^1 \eta_j \int_\Sigma h^\sharp_{(j)}(\pi^{(j)}_o, \pi^{(j)}_o)  + h^\sharp_{(j)}(a^{(j)}_o, (\overline{\Delta^{(j)}+ m^2 I}) a^{(j)}_o) \vol_h\label{TRASF1}\end{equation}
Furthermore, the new Cauchy data, though they stay outside $C_\Sigma$  in general, they however satisfy the natural generalization of the constraints defining $C_\Sigma$ in view of Lemma~\ref{LEMMAFORME}: 
\begin{equation}
\pi^{(0)}_o  =- \overline{\delta^{(1)}_h} a^{(1)}_o \:, \quad   \overline{(\Delta_h^{(0)} + m^2)} a^{(0)}_o = \overline{\delta^{(1)}_h} \pi^{(1)}_o\:. \label{CCdataGEN}
\end{equation}
These identities arise immediately from Definitions~(\ref{newA}),  the constraints (\ref{CCdata}), and by applying Lemma~\ref{LEMMAFORME} and paying attention to the fact that $\Omega_c^{(j)}(\Sigma)\subset  D((\overline{{\Delta}^{(j-1)}+ cI})^\alpha  \overline{\delta^{(j)}})$ for every  $\alpha \leq 1$ and also using 
$(\overline{\Delta^{(j)} + m^2I}) (\overline{\Delta^{(j)} + m^2I})^{-1/4}  =  (\overline{\Delta^{(j)} + m^2I})^{-1/4}   \overline{\Delta^{(j)} + m^2I}$ (for, e.g., \cite[(f) in Proposition 3.60 ]{FFM}).  Using (\ref{agg}) and (\ref{CCdataGEN}) in the right-hand side of (\ref{TRASF1}), we can proceed backwardly as in the proof that (\ref{EP}) is equivalent to (\ref{2T00}). Indeed, the only ingredients we used in that proof were the constraint equations which are valid also for $A_o$ and the duality of $\delta$ and $d$ with respect to the Hodge inner product,  which extends to $\overline{\delta}$ and
$\overline{d}$. In summary, 
\begin{align*}
\mu(A,A) = \frac{1}{2}\int_\Sigma & \Big( h_{(1)}^\sharp( \pi^{(1)}_o-\overline{d^{(0)}}a_o^{(0)},  \pi^{(1)}_o-\overline{d^{(0)}}a^{(0)}_o) + h_{(2)}^\sharp (\overline{d^{(1)}}a^{(1)}_o, \overline{d^{(1)}}a^{(1)}_o) \\
& + 
m^2 \big( h_{(1)}^\sharp(a^{(1)}_o, a^{(1)}_o)
+ a^{(0)}_oa^{(0)}_o\big) \Big) \vol_h\:.
\end{align*}
From that identity, it is clear that $\mu(A,A) \geq 0$ and $\mu(A,A)=0$ implies $A_o=0$, which in turn yields $A=0$ because the operators $\overline{\Delta^{(j)} + m^2I}^{1/4}$ are injective. We have established that $\mu : C_\Sigma \times C_\Sigma \to \RR$ is a positive real symmetric inner product.

Let us pass to prove (\ref{diseq}). First of all, we change the notation concerning the scalar product $\mu$ making explicit the decomposition of $A$, and we work with {\em complex} valued forms. We use
$$ A = (a,\pi) = (a^{(0)}, \pi^{(0)}, a^{(1)}, \pi^{(1)})\:,\quad  a:= (a^{(0)}, a^{(1)})\:, \quad  \pi:= (\pi^{(0)}, \pi^{(1)})$$
so that, if  $(a,\pi), (a',\pi') \in (L_0^2(\Sigma, \vol_h) \oplus L_1^2(\Sigma, \vol_h)) \times  (L_0^2(\Sigma, \vol_h) \oplus L_1^2(\Sigma, \vol_h)) $ are such that the right-hand side below is defined, we can write
$$\mu( (\overline{a},\overline{\pi}), (a',\pi')) := \sum_{j=0}^1 \frac{\eta_j}{2} \int_\Sigma   h_{(j)}^\sharp(\overline{\pi^{(j)}}, H^{-1}_{(j)} \pi^{(j)'})    +  h^\sharp_{(j)}  (\overline{a^{(j)}}, H_{(j)} a^{(j)'}  )\vol_h  $$
where $H_{(j)} := \overline{\Delta^{(j)} + m^2I}^{1/2}$, 
and the bar on forms denotes the complex conjugation.
Finally, for $\alpha= \pm 1$, we defined
$$H^{\alpha} a := ( H_{(0)}^{\alpha} a^{(0)},  H_{(1)}^{\alpha} a^{(1)})\:, \quad H^{\alpha} \pi := ( H_{(0)}^{\alpha} \pi^{(0)},  H_{(1)}^{\alpha}\pi^{(1)})\:. $$
By direct inspection one sees that, if  $(a,\pi), (a',\pi') \in C_\Sigma + iC_\Sigma$, then the right-hand side of the first identity below  is well-defined and
\begin{align*}
\Lambda((a,\pi), (a',\pi')):= & \frac{1}{2}\mu\left((\overline{\pi}+iH^{-1} \overline{a}, \overline{a}- iH\overline{\pi}),(\pi'-iH^{-1} a', a'+i H\pi')\right) \\
=& \mu( (\overline{a},\overline{\pi}), (a',\pi'))+\frac{i}{2} \sigma^{(\P)} ((\overline{a},\overline{\pi}), (a',\pi'))
\end{align*}
where $\sigma^{(P)}$ is the right-hand side of (\ref{NEWSP}), which however coincides with the original symplectic form 
(\ref{SYMPr})  evaluated on complex Cauchy data because $(a,\pi), (a',\pi')\in C_\Sigma+iC_\Sigma$ and Remark~\ref{REMSYMP} holds.  Finally notice that if $(a,\pi) \in C_\Sigma + iC_\Sigma$ then $a_o := \pi -i H a$ and $\pi_o:= a+ iH^{-1}\pi$  satisfy the constraints (though they do not belong to $ C_\Sigma + iC_\Sigma$ in general)
\begin{equation*}
\pi^{(0)}_o  =- \overline{\delta^{(1)}_h} a^{(1)}_o \:, \quad   H_{(0)} a^{(0)}_o = \overline{\delta^{(1)}_h} \pi^{(1)}_o\:. 
\end{equation*}
The proof is direct, using   Lemma~\ref{LEMMAFORME} once more. As a consequence, exploiting the same argument to prove (\ref{POSCOMP}) and observing that $H^\alpha$ commutes with the complex conjugation --  so that it holds $\overline{\pi}-iH^{-1} \overline{a}= \overline{\pi +i H^{-1}a}$ for instance --  we have that  $$2\Lambda((a,\pi), (a'\pi'))=
\mu\left((\overline{\pi}+iH^{-1} \overline{a}, \overline{a}- iH\overline{\pi}),(\pi-iH^{-1} a, a +i H\pi)\right) $$ $$ =
\mu\left((\overline{\pi- iH^{-1} a, a+ iH\pi}),(\pi-iH^{-1} a, a +i H\pi)\right)
 \geq 0\:.$$
The final inequality is due to the fact that  $\mu$ is (the complexification of)  a {\em real}  positive bilinear  symmetric form.
All that means in particular that the {\em  Hermitian} form $\Lambda$ on $(C_\Sigma + iC_\Sigma)\times  (C_\Sigma + iC_\Sigma)$ is (semi)positively defined and thus it satisfies the 
Cauchy-Schwartz inequality. In particular,
$$(Im \Lambda( (a,\pi),(a',\pi')))^2 \leq | \Lambda( (a,\pi),(a',\pi'))|^2 \leq  \Lambda((a,\pi),(a,\pi))  \:\Lambda((a',\pi'),(a',\pi'))\:. $$
If choosing $ (a,\pi),(a',\pi') \in C_\Sigma$ (thus {\em real} forms), the above inequality specialises to
$$\sigma^{(\P)} ((a,\pi),(a',\pi'))^2 \leq 4\mu((a,\pi),(a,\pi))\: \mu((a',\pi'),(a',\pi'))$$
which is the inequality (\ref{diseq})  we wanted to prove.
\end{proof}

\subsection{Hadamard states in ultrastatic and generic globally hyperbolic spacetimes}
With the next proposition, we show that the quasifree states defined in Proposition~\ref{PROPHAD1} is a Hadamard state when $(\Sigma, h)$ is of bounded geometry. To prove the assertion we will take advantage of the general formalism developed in \cite{gerardBook} and \cite{gravity}. An alternative proof,  which does not assume that the manifold is of bounded geometry (however  we here take advantage of~\cite{greene}),  could be constructed along  the procedure developed in~\cite{FNW} and extending it to the vectorial Klein-Gordon field. 

\begin{proposition}\label{prop:ultraHada}

 If the metric $h$ on the time slice $\Sigma$ is of bounded geometry, then  the quasifree state $\omega_\mu: \mathcal{A}_g \to \CC$ defined in Proposition~\ref{PROPHAD1}  is Hadamard according to Definition~\ref{HadamardMMV}.
\end{proposition}

\begin{proof} 
Consider a pair of {\em complex}  Klein-Gordon fields $A^{(0)}$ and $A^{(1)}$ in the ultrastatic spacetime
$(\M,g)= (\RR\times \Sigma, -dt\otimes dt+h)$, with $(\Sigma, h)$ a smooth complete Riemannian manifold of bounded geometry  obeying the normally hyperbolic equations  (\ref{EQ1}) and (\ref{EQ2}) in the respective vector bundles on $\M$, according to Section \ref{SECCP}. We stress that we now assume that the two fields are complex. Referring to  \cite[Chapter 4]{gerardBook}, we define the {\em covariances}, for $j=0,1$
\begin{align}\label{lambdaplus}
 \lambda_{(j)}^+(A^{(j)}, A^{(j)'}) :=& \frac{1}{2} \int_\Sigma   h_{(j)}^\sharp(\overline{\pi^{(j)}}, H^{-1}_{(j)} \pi^{(j)'})    +  h^\sharp_{(j)}  (\overline{a^{(j)}},  H_{(j)} a^{(j)'}  )\: \vol_h 
 + \frac{i}{2}\sigma^{(j)}(\overline{A^{(j)}},A^{(j)'})\\
 \lambda_{(j)}^-(A^{(j)}, A^{(j)'}) :=& \frac{1}{2} \int_\Sigma   h_{(j)}^\sharp(\pi^{(j)'}, H^{-1}_{(j)}\overline{\pi^{(j)}})    +  h^\sharp_{(j)}  (a^{(j)'},  H_{(j)}\overline{a^{(j)}}  )\: \vol_h
 + \frac{i}{2}\sigma^{(j)}(A^{(j)'},\overline{A^{(j)}})
\end{align}
where $H_{(j)} :=  \overline{\Delta^{(j)} + m^2}^{1/2}$,
$\sigma^{(j)}$ are the symplectic forms of the corresponding Klein-Gordon fields taking place in the right-hand side of  (\ref{NEWSP}), now evaluated on complex fields.
Above, $a^{(j)}, \pi^{(j)} \in \Omega_c^j(\Sigma)_\CC$ are the Cauchy data on $\Sigma$ of $A^{(j)}$ respectively
and  $a^{(j)'}, \pi^{(j)'} \in \Omega_c^j(\Sigma)_\CC$ are the Cauchy data on $\Sigma$ of $A^{(j)'}$ respectively. Notice that we are not imposing constraints on these initial data since we are dealing with independent Klein-Gordon fields.  
$\lambda^\pm_{(j)}$ are evidently positive because, if all involved forms in the right-hand side are smooth and compactly supported, then the right-hand side of the identity above is well-defined and
$$\lambda_{(j)}^+(A^{(j)}, A^{(j)'}) := \frac{1}{2} \int_\Sigma   h^\sharp_{(j)}  ( \overline{H^{1/2}a^{(j)} +i H^{-1/2} \pi^{(j)}},  H_{(j)}^{1/2} a^{(j)'}
+i H^{-1/2} \pi^{(j)'})\: \vol_h \:.$$
The case of $\lambda_{(j)}^-$ is strictly analogous.
Furthermore
$$ \lambda_{(j)}^+(A^{(j)}, A^{(j)'}) - \lambda_{(j)}^-(A^{(j)}, A^{(j)'})  =   i\sigma^{(j)}(\overline{A^{(j)}}, A^{(j)'})\:.$$
Therefore $\lambda^\pm_{(j)}$ satisfy the hypotheses of \cite[Proposition 4.14]{gerardBook}\footnote{The reader should pay attention to the fact that 
the Cauchy data used in  \cite{gerardBook}, {\em in the complex case},  are defined as $(f_0,f_1):= (a, -i\pi)$ {\em instead of our} $(a, \pi)$!
This is evident by comparing (2.4) and (2.20) in \cite{gerardBook}.
 With the choice of \cite{gerardBook},
 $i\overline{(f_0,f_1)}^t\cdot q (f'_0,f'_1) = \int \overline{f_0}f_1' +\overline{ f_1}f_0' \vol_h = i \sigma(\overline{(a, \pi)}, (a', \pi'))$, where $\cdot q \equiv  \sigma_1$  (the Pauli matrix) according to  \cite{gerardBook}.}
 so that  they define a pair, for $j=0,1$, of gauge-invariant quasifree states for the complex Klein-Gordon fields respectively associated to Equations~(\ref{EQ1}) and~(\ref{EQ2}). We pass to prove that both states are Hadamard exploiting the fact that $(\Sigma,h)$ is of bounded geometry.
By rewriting the covariances $\lambda_{(j)}^\pm$ as $\lambda_{(j)}^\pm = \pm qc_{(j)}^\pm$ ($q= i\sigma^{(j)}$) a quick computation shows that 
\begin{equation*}
c_{(j)}^\pm = \frac{1}{2}
\begin{bmatrix}
I & \pm H_{(j)}^{-1}  \\
\pm H_{(j)} & I  
\end{bmatrix}\,.
\end{equation*}
We can immediately realize that the operator $c_{(j)}^\pm $ is the same Hadamard projector obtained in~\cite[Section 5.2]{gravity}\footnote{ It follows immediately since $b^+(t)= -b^-(t)=H:= \overline{\Delta^{(j)}+ m^2I}^{1/2}$.} -- see also~\cite[Section 11]{gerardBook} for a more introductory  explanation for the scalar case. This operator belongs to the necessary class of pseudodifferential operators $C^\infty_b(\RR;\Psi^1_b(\Sigma))$ because $(\Sigma, h)$ is of bounded geometry. Therefore, on account of~\cite[Proposition 5.4]{gravity}, the two quasifree states associated to $\lambda_{(j)}^+$, for $A^{(j)}$ and $j=0,1$, are Hadamard.
In other words,  the Schwartz kernels provided by the  two-point functions  $\lambda^+_{(j)}(\G^{(j)}\cdot, \G^{(j)}\cdot)$, viewed as distributions of $\Gamma(\V_g^{(j)}\boxtimes \V_g^{(j)})'$,  satisfy
$$WF(\lambda^+_{(j)}(\G^{(j)}\cdot, \G^{(j)}\cdot)) = \mathcal{H} \:,$$
where $\mathcal{H}$ is defined in (\ref{WHadMMV}) and $\G^{(i)}$, $i=0,1$ are the causal propagators  for the normally hyperbolic operators 
$$\N^{(i)}:=\partial_t^2  + \Delta^{(i)}_h + m^2I  : \Gamma_{sc}(\V_g^{(i)}) \to  \Gamma_{sc}(\V_g^{(i)})\quad i=0,1\,.$$
 Above and from now on we use the same notation to indicate a bidistribution and the associated Schwartz kernel.
Notice that we have used the same symbol $\G^{(j)}$  of the  causal propagator we used for  the real vector field case. This is  because the causal propagators  for the complex fields are the direct complexification of the scalar case (see Remark \ref{REMSYMP}). We  pass now to focus on the expression of $\omega_{\mu 2}$ provided  in (\ref{omega2R}) taking the usual decomposition $\Omega_c^{1}(\M)_\CC \ni \ff= \ff^{(0)}dt + \ff^{(1)}$ into account.  It can be written
\begin{equation*} \omega_{\mu2}(\ff, \ff') = \omega^{(1)}_{\mu2}(\ff^{(1)}, \ff^{(1)'})-  \omega^{(0)}_{\mu2}(\ff^{(0)}, \ff^{(0)'})
\end{equation*}
where, comparing (\ref{DEFMU}) and (\ref{omega2R}) with (\ref{lambdaplus}) for {\em real} arguments $\ff, \ff' \in \Gamma(\V_g)$, we find
$$ \omega^{(j)}_{\mu2}(\ff^{(j)}, \ff^{(j)'})= \lambda_{(j)}^+(\G^{(j)}\ff^{(j)}, \G^{(j)}\ff^{(j)'})\:.$$
We have   $$WF(\pm \omega^{(j)}_{\mu 2}) = WF(\pm \lambda^+_{(j)}(\G^{(j)}\cdot, \G^{(j)}\cdot)) =
 WF(\lambda^+_{(j)}(\G^{(j)}\cdot, \G^{(j)}\cdot)) = \mathcal{H}\quad \mbox{for $j=0,1$.}$$
Taking (\ref{decVV}) into account, we now  observe that 
$\omega_{\mu2} \in \Gamma(\V_g \boxtimes \V_g)' =  \Gamma((\V^{(0)}_g\oplus V^{(1)}_g) \boxtimes (\V^{(0)}_g\oplus V^{(1)}_g))' $. As a matter of fact,  however, 
$\omega_{\mu2}$ does not have mixed components acting on sections of   $ V^{(1)}_g \boxtimes \V^{(0)}_g$ and  $V^{(0)}_g \boxtimes \V^{(1)}_g$ and the only components of that distribution are those which  work on sections of  $V^{(0)}_g \boxtimes \V^{(0)}_g$ and  $V^{(1)}_g \boxtimes \V^{(1)}_g$. These are respectively represented by $-\omega^{(0)}_{\mu2}$
and $\omega^{(1)}_{\mu2}$ whose wavefront set is $\mathcal{H}$ in both cases. The remaining two components have empty wavefront set since they are the zero distributions.
Applying the definition of wavefront set of a  vector-valued distribution \cite{SV},  we conclude that 
$$WF(\omega_{\mu2})= WF(-\omega^{(0)}_{\mu2}) \cup  WF(\omega^{(1)}_{\mu2})  \cup \emptyset \cup \emptyset = \mathcal{H} \cup  \mathcal{H} \cup \emptyset \cup \emptyset
=  \mathcal{H}\:,$$
concluding the proof.
\end{proof}

Combining the results obtained so far, we get the main result of this paper.
\begin{theorem}
Let $(\M, g)$ be a globally hyperbolic spacetime and refer to the
 $CCR$-algebra $\mathcal A_g$ of the real Proca field. Then there exists a quasifree Hadamard state on $\mathcal A_g$.
\end{theorem}
\begin{proof}[Proof of Theorem~\ref{thm:main intro 3}]
As already explained in the beginning of Section~\ref{sec:existence}, for any globally hyperbolic spacetime  $(\M,g)$, there exists a paracausally related  globally hyperbolic  spacetime $(\M,g_0)$ which is ultrastatic and whose spatial metric is of bounded geometry.
In particular, in this class of spacetimes, the quasifree states defined in Proposition~\ref{PROPHAD1} satisfy the microlocal spectrum condition, as proved in Proposition~\ref{prop:ultraHada}.
Therefore, since the pull-back along a M\o ller $*$-isomorphism preserves the Hadamard condition on account of Theorem~\ref{thm:main intro Had}, we can conclude.
\end{proof}

\section{Comparison with Fewster-Pfenning's definition of Hadamard states}\label{SECHADFP}
Though the paper \cite{Fewster} by Fewster and Pfenning concerns {\em quantum energy inequalities}, it also offers a general theoretical discussion about the algebraic quantization of the Proca and the Maxwell fields in curved spacetime. In particular, the authors propose a definition of a Hadamard state which appears to be technically different from ours at first glance, even if it shares a number of important features with ours. This section is devoted to a comparison of the two definitions for the Proca field.

\subsection{Proca Hadamard states according to Fewster and Pfenning}
The definition of Hadamard state stated in \cite[Equation~(35)]{Fewster} is formulated in terms of causal normal neighborhoods of smooth spacelike Cauchy surfaces (see also below) and the {\em global Hadamard parametrix} for distributions which are bisolutions of the vectorial Klein-Gordon equation. Our final goal is to prove an equivalence theorem of our definition of Hadamard state Definition~\ref{HadamardMMV} and the one adopted in \cite{Fewster}.

As a first step, we translate the original  Fewster-Pfenning's definition of a Hadamard state  into an equivalent form which will turn out to be more useful for our comparison. 
The equivalence of  the version  stated below of Fewster-Pfenning's definition and  the original one 
was  established  in~\cite[Section III C]{Fewster} (see also the comments under Definition \ref{Hadamard}).

\begin{definition}\label{Hadamard} [{\bf Fewster-Pfenning's definition of Proca Hadamard state}] Consider the globally hyperbolic spacetime $(\M,g)$ and a  state $\omega:\mathcal{A}_g\to\CC$ 
for the Proca algebra of observables on $(\M,g)$. $\omega$ is  called
	  \textbf{Hadamard} if it is quasifree and its two-point function has the form
\begin{equation}  \omega(\hat{\fa}(\ff)\hat{\fa}(\fh))=W_g(\ff,Q\fh) \label{HadF} \end{equation} 
$\forall\ff,\fh\in\Gamma_c(\V_g)$, where $Q:\Gamma(\V_g)\to\Gamma(\V_g)$ in the differential operator $Q=\Id+ m^{-2}(d\delta_g) .$
Above 
	$W_g \in \Gamma_c'(\V_g \boxtimes\V_g)$ is a  Klein-Gordon distributional bisolution 
such that \begin{equation} W_g(\ff,\fg)- W_g(\fg,\ff)= i\G_\N(\ff,\fg) \quad   \mbox{mod $C^\infty$}\:,\label{HadF2}\end{equation} 
$G_\N$ being the causal propagator of the Klein-Gordon operator  (\ref{NtoP}) and 
which satisfies the  microlocal spectrum condition
	\begin{equation} \label{WHad} WF(W_g)=\{(x,k_x;y,-k_y)\in T^*\M^2\backslash\{0\}\:|\:(x,k_x)\sim_{\parallel}(y,k_y), k_x\triangleright0\}\:.\end{equation}
\end{definition}

 \remark The equivalence of Definition~\ref{Hadamard}  and the original one  stated in \cite{Fewster} relies on  Sahmann -Verch's  \cite{SV} generalization  to  vector (and spinor)  fields of some classic Radzikowski results originally formulated for scalar fields. In practice, (a) if a distribution which is a bisolution of the  vectorial Klein-Gordon equation and it is of Hadamard form in a normal causal neighborhoods of a smooth spacelike Cauchy surface, then it necessarily has the wavefront set of the form  (\ref{WHad}) ((a)  \cite[Theorem 5.8]{SV})  and its antisymmetric part satisfies (\ref{HadF2}) directly from the definition of parametrix; (b) if  a distribution which is a bisolution of the  vectorial Klein-Gordon equation satisfies (\ref{WHad}) and (\ref{HadF2}), then it is of Hadamard form in some normal causal neighborhoods of a smooth spacelike Cauchy surface (see \cite[Remark 5.9. (i)]{SV}). \\

For the Proca fields, it has been established  in \cite{Fewster} the property of propagation of the Hadamard condition stated in the next proposition. That result was  already established for the Hadamard states of  scalar and vector (including spinor)  fields in \cite{FSW,KW,SV} (see \cite{IgorValter,Valter} for a general recap for the KG scalar field).  The pivotal tool is the already mentioned notion of {\em causal normal neighborhood}
$\mathcal{N}$ {\em  of a 
smooth spacelike Cauchy surface} $\Sigma$  in a globally hyperbolic spacetime $(\M;g)$.  The notion  introduced in \cite{KW} has been recently improved (closing a gap in the geometric  definition of Hadamard states) in \cite{Valter}\footnote{Where these open sets  are named normal neighborhoods  of smooth spacelike Cauchy surfaces, omitting ``causal''.}. The propagation  results established in \cite{KW,SV} and \cite{Fewster} are  valid with the improved notion of causal normal neighborhoods and Hadamard states of \cite{Valter}.

\begin{proposition}\label{propagation}
Let $\omega:\mathcal{A}_g\to\mathbb{C}$ be a quasifree state for the Proca field 
in the globally hyperbolic spacetime $(\M,g)$. Let  $\mathcal{N}$  be a causal normal  neighborhood of a Cauchy surface $\Sigma$ of  $(\M,g)$. Suppose that the restriction of $\omega$  to $(\mathcal{N}, g|_\mathcal{N})$ is Hadamard according to Definition \ref{Hadamard}. Then $\omega$ is Hadamard in $(\M,g)$ according to the same definition.
\end{proposition}

\remark \label{remarkN} In order to compare Proposition~\ref{propagationMMV} and Proposition~\ref{propagation}  we stress that the requirement that the neighborhood  $\mathcal{N}$ of a Cauchy surface  is causal normal can be relaxed also in  Proposition~\ref{propagation} to make contact with our Proposition~\ref{propagationMMV}. One may only assume that $(\mathcal{N}, g|_{\mathcal{N}})$ is globally hyperbolic also therein. That is a consequence of the following facts.
\begin{itemize}

\item[(a)]  Every  causal normal neighborhood $\mathcal{N}\subset \M$ of a Cauchy surface $\Sigma$ of $(\M, g)$ is, by definition \cite{KW,Valter}, a globally hyperbolic spacetime  with respect to the restriction of the metric and $\Sigma$ is also a Cauchy surface in $(\mathcal{N}, g|_\mathcal{N})$.

\item[(b)] Every smooth spacelike Cauchy surface admits a causal normal neighborhood \cite{KW,Valter}. 

\item[(c)]  According  to the proof of \cite[Lemma 2.2 ]{KW} whose validity extends to \cite{Valter},  
every neighborhood of a smooth spacelike Cauchy surface includes   a causal normal neighborhood of that Cauchy surface\footnote{Essentially because  convex normal neighborhoods  of points form a topological basis of any  spacetime and in view of \cite[Proposition 9]{Valter}}.

\end{itemize}

The smoothness  property corresponding to our Proposition~\ref{smoothnessMMV} also holds for  Hadamard bisolutions in the sense of Fewster-Pfenning. In \cite{Fewster}, it is  an immediate consequence of (\ref{HadF}) and  the analogous  feature of   Klein-Gordon   bisolutions  (see the discussion on p. 4488 in \cite{Fewster}).

\begin{proposition}\label{smoothness}
Let $\omega, \omega' \in \Gamma'_c(V_g\boxtimes V_g)$ be a pair of bisolutions of the Proca equation satisfying the Hadamard condition~(\ref{HadF}) for corresponding Klein-Gordon bisolutions  which, in turn, satisfy (\ref{HadF2}). Then,  the differences between the two bisolutions is smooth: $\omega-\omega' \in \Gamma(V_g\boxtimes V_g)$.
\end{proposition}

Finally, \cite{Fewster} also contains a proof of the existence  of Hadamard states for the Proca (and the Maxwell) field in globally hyperbolic spacetimes with compact Cauchy surfaces (whose first homology group is trivial when treating the Maxwell field). This proof establishes first the existence in ultrastatic spacetimes and next it exploits a standard deformation argument \cite{W}. 

\subsection{An (almost)  equivalence theorem}  We are in a position to state and prove our equivalence result.

\begin{theorem} \label{teoequiv}Consider the globally hyperbolic spacetime $(\M,g)$ and a quasifree  state $\omega:\mathcal{A}_g\to\CC$ 
for the $*$-algebra of observables on $(\M,g)$ of the real Proca field. Let $\omega_2\in \Gamma_c'(\V_g\boxtimes \V_g)$ be the two-point function of $\omega$.   The following facts are true.
\begin{itemize}
\item[(a)] If $\omega$ is Hadamard according to Definition~\ref{Hadamard}, then  it is also Hadamard according to Definition~\ref{HadamardMMV}. 
\item[(b)] If $(\M,g)$ admits a Proca quasifree Hadamard state according to Definition~\ref{Hadamard}  and $\omega$  is Hadamard according to Definition~\ref{HadamardMMV},  then $\omega$ is Hadamard in the sense of  Definition~\ref{Hadamard}.
\end{itemize}
\end{theorem}

\begin{proof}Tha following argument is identical to the one used in \ref{propagationMMV} to prove $WF(\G_\P)=WF(\G_\N)$, but we repeat it here to keep this section self-contained.\\
First of all notice that, since $\omega_2(\ff,\fg) = W_g(\ff, Q\fg)$, then viewing $\omega_2$ and $W_g$ as bidistributions, we have
$\omega(x,y) = (Id_x\otimes Q_y)\: W(x,y)$  (where we have used the fact that $Q$ is formally selfadjoint)  taking Remark \ref{remdist} into account).\\
Now  suppose that $\omega$ is Hadamard according to Definition~\ref{Hadamard}.
Since  $W_g$ satisfies the microlocal spectrum condition  and the differential operator $I\otimes Q$ is smooth, we have 
$$ WF(\omega_2)\subset WF(W_g) = \{(x,k_x;y,-k_y)\in T^*\M^2\backslash\{0\}\:|\:(x,k_x)\sim_{\parallel}(y,k_y), k_x\triangleright0\}\:.$$
Notice that, in particular, $k_x$ and $k_y$ cannot vanish (simultaneously or separately) if they take part of $WF(W_g)$.
Let us prove the converse inclusion to complete  the proof of (a).
Again from known results, from $\omega_2(x,y) = (Id_x\otimes Q_y) W_g(x,y)$,  we have 
$$WF(W_g) \subset Char(I\otimes Q) \cup WF(\omega_2)\:.$$
However, by direct inspection, one sees that
 $$ Char(I\otimes Q)= \{(x, k_x; y, 0) \:|\: (x,k_x) \in \T^*\M\:, y \in \M\}\:,$$
so that
\begin{equation} WF(\omega_2)\subset WF(W_g) \subset  WF(\omega_2) \cup  \{(x, k_x; y, 0) \:|\: (x,k_x) \in \T^*\M\:, y \in \M\}\:.\label{chain}\end{equation}
However $WF(W_g) \cap  \{(x, k_x; y, 0) \:|\: (x,k_x) \in \T^*\M\:, y \in \M\} = \emptyset$  and thus  we can re-write the chain of inclusions (\ref{chain}) as 
$$ WF(\omega_2)\subset WF(W_g) \subset  WF(\omega_2) \quad \mbox{so that}\quad WF(\omega_2) = WF(W_g)\:.$$ 
This is the thesis of (a) because we have established that Definition~\ref{HadamardMMV} is satisfied by $\omega$.\\
To prove (b), let us assume that $\omega$ satisfies Definition~\ref{HadamardMMV}.  By hypotheses the  antisymmetric part of $\omega_2$ is $-iG_\P$. Let $\Omega$ be another quasifree state of the Proca field which satisfies Definition~\ref{Hadamard}. Also the antisymmetric part of $\Omega_2$ is $-iG_\P$.

Due to Proposition~\ref{smoothnessMMV}, 
$$F(x,y) := \omega_2(x,y)-\Omega_2(x,y)\:.$$
 is a smooth function.  Furthermore it is a symmetric bisolution of the Proca equation. In particular it therefore satisfies\footnote{We are grateful to C. Fewster for this observation.} $F(\ff, d \fh^{(0)}) = 0$, where $\fh^{(0)}\in \Omega^0_c(\M)$, so that 
$$F(\ff, Q\fg) = F(\ff, \fg) + \frac{1}{m^2}F(\ff, d (\delta_g \fg)) = F(\ff, \fg) \:.$$
Collecting everything  together, we can assert that, for some distributional bisolution of the Klein-Gordon equation $W_g$ 
which satisfies (\ref{HadF2}), (\ref{WHad}), and is 
associated to the Hadamard state $\Omega$, it holds
$$\omega_2(\ff,\fg) = W_g(\ff,Q\fg) + F(\ff,\fg)  = W_g(\ff,Q\fg) + F(\ff,Q\fg)\:.$$
If we re-absorb $F$ in the definition of $W_g$,
$$W'_g(\ff,\Q\fg) =  W_g(\ff,Q\fg) + F(\ff,Q\fg)\:.$$
the new $W'_g$ is again a distributional bisolution of the Klein-Gordon equation  which satisfies (\ref{HadF2}), (\ref{WHad})
and 
$$\omega_2(f,g) = W'_g(\ff,Q\fg)\:.$$
In other words, the Hadamard state  $\omega$ according to Definition~\ref{HadamardMMV} is also  Hadamard in the sense of Definition~\ref{Hadamard} concluding the proof of (b).
\end{proof}

\remark Regarding (b), the existence of Hadamard states in the sense of Definition~\ref{Hadamard} has been established in \cite{Fewster} for globally hyperbolic spacetimes whose Cauchy surfaces are compact: in those types of spacetimes at least, the two definitions are completely equivalent. We expect that actually the equivalence is  complete, even dropping the compactness hypothesis (see the conclusion section). This issue will be investigated elsewhere.

\section{Conclusion and future outlook}\label{sec:concl}
We conclude this paper by discussing some open issues which are raised in this paper and we leave for future works.

On an ultrastatic spacetime  $\M = \RR \times \Sigma$, 
 the one-parameter group of isometries given by  time-translations
has an associated  action on $\mathcal{A}_g$ in terms of $*$-algebras isomorphisms $\alpha_u$ completely induced by $$\alpha_u(\hat{\fa}(\ff)) := \hat{\fa}(\ff_u)$$ for every $\ff \in \Gamma_c(\M)$, where $\ff_u(t, p) := \ff(t-u,p)$ for every $t,u\in \RR$ and 
$p\in \Sigma$. 
It is shall not be difficult to prove that the Hadamard state constructed in Theorem~\ref{thm:intro intermezzo} is  invariant under the action of $\alpha_u$
  $$\omega_\mu(\alpha_u(a)) = \omega_\mu(a) \quad \forall u \in \RR \quad \forall  a \in \mathcal{A}_g$$
 It should be also true that the map $$\RR \ni u \mapsto \omega_\mu(b\alpha_u(a)) \in \CC$$ 
 is continuous for every $a,b\in  \mathcal{A}_g$ which would assure (see, e.g. \cite{FFM}) that $\alpha := \{\alpha_h\}_{h\in \RR}$ is unitarily implementable by a strongly continuous unitary representation of $\RR$ in the GNS representation of $\omega_\mu$ and that the vacuum vector of the Fock-GNS representation is left invariant under the said unitary representation. We expect that the selfadjoint generator of that unitary group has a positive spectrum where, necessarily,  the vacuum state is an  eigenvector with eigenvalue $0$. In other words $\omega_\mu$ should be a {\em ground state} of $\alpha$. We finally expect that $\omega_\mu$ is {\em pure} (on the Weyl algebra associated to the symplectic space $((\Ker \P) \cap \Gamma_{sc}(\M), \sigma^{(\P)})$  and it is the {\em unique quasifree algebraic state which is invariant under $\alpha$}. We can summarize the previous discussion in the following question.

\begin{question} Is the Hadamard state defined in Theorem~\ref{thm:intro intermezzo} a {\em ground state} for the time-translation? More precisely, is it the unique, pure, quasifree algebraic state  which is invariant under action of  $\alpha$?
\end{question}

Last, but not least, we have seen in Section~\ref{SECHADFP} that if a globally hyperbolic manifold admits a Proca quasifree Hadamard state according to the definition of Fewster-Pfenning, then Definition~\ref{HadamardMMV} and~\ref{Hadamard} are equivalent. 
This is the case for example for globally hyperbolic spacetimes whose Cauchy surfaces are compact. We do expect to extend this result for the whole class of globally hyperbolic spacetime. 
\begin{conjecture}
Deifnition~\ref{HadamardMMV} and~\ref{Hadamard} are equivalent on any globally hyperbolic spacetime.
 \end{conjecture}

As is evident from our quasi equivalence theorem, a complete equivalence would take place if a Hadamard state  according to \cite{Fewster} is proved to exist for every globally hyperbolic spacetime.
As a matter of fact, we expect that every globally hyperbolic spacetime $(\M,g)$ admits a quasifree Proca Hadamard state $\omega$ according to Fewster and Pfenning. This state should exist in every  paracausally related  ultrastatic spacetime $(\RR\times \Sigma, -dt^2+h)$ with complete Cauchy surfaces of bounded geometry.  With the same argument used for  our existence proof  of Hadamard states or the deformation argument exploited in \cite{Fewster}, it should be  possible to export this state to the original space  $(\M,g)$. 
We expect that the  Hadamard Klein-Gordon bisolution for the real Proca field  on $(\RR\times \Sigma, -dt^2+h)$ used to define $\omega$ according to (\ref{HadF}) in  Definition \ref{Hadamard} should have this form.
$$W_g(\ff,\ff') := \mu(\G_\N\ff, \G_\N\ff') + \frac{i}{2}\sigma^{(N)}(\G_\N\ff, \G_\N\ff')\:,\quad \ff,\ff' \in \Gamma_c(\RR\times \Sigma)\:,$$
where $\N$ is the Klein-Gordon operator (\ref{NtoP}) associated to $\P$ and $\G_\N$ its causal propagator. The bilinear symmetric form  $\mu : \left( (\Omega^0_c(\Sigma))^2 \times (\Omega^1_c(\Sigma))^2\right) \times \left( (\Omega^0_c(\Sigma))^2 \times (\Omega^1_c(\Sigma))^2\right) \to \RR$ is defined as in (\ref{DEFMU}), but with the crucial difference that here its arguments are not restricted to $C_\Sigma \times C_\Sigma$.

\end{document}